\documentclass[a4paper,11pt]{article}
\usepackage{a4wide}
\usepackage[utf8]{inputenc}
\usepackage{hyperref}
\usepackage{ss}
\usepackage{amsmath}
\usepackage{amsthm}
\usepackage{graphicx}
\usepackage[linesnumbered,lined,ruled,commentsnumbered]{algorithm2e}

\makeatletter

\newcommand{\Rmnum}[1]{\expandafter\@slowromancap\romannumeral #1@}
\makeatother

\newcommand{\wgfp}{\operatorname{WGFP}}
\newcommand{\gfp}{\operatorname{GFP}}
\newcommand{\opt}{\operatorname{Opt}}
\newcommand{\RR}{\mathbb{R}}
\newcommand{\sD}{\mathscr{D}}
\newcommand{\sF}{\mathscr{F}}

\newtheorem*{theorem*}{Theorem}

\hypersetup{colorlinks=true,citecolor=blue,linkcolor=blue,urlcolor=blue}

\begin{document}

\title{\bf A Strongly Polynomial-Time Algorithm for
Weighted General Factors with Three Feasible Degrees\thanks{
An extended abstract of this paper appeared in the Proceedings of the 34th
International Symposium on Algorithms and Computation (ISAAC
2023)~\cite{sz23:isaac}.
The research leading to these results has received funding from the European Research Council (ERC) under the European Union's Horizon 2020 research and innovation programme (grant agreement No 714532). This work was also supported by UKRI EP/X024431/1. For the purpose of Open Access, the authors have applied a CC BY public copyright licence to any Author Accepted Manuscript version arising from this submission. All data is provided in full in the results section of this paper. Part of the work was done while the first author was a postdoctoral research associate at the University of Oxford.}
}

\author{Shuai Shao\\University of Science and Technology of China\\\texttt{shao10@ustc.edu.cn}
\and Stanislav \v{Z}ivn\'y\\University of Oxford\\\texttt{standa.zivny@cs.ox.ac.uk}}

\maketitle

\begin{abstract}
General factors are a generalization of matchings. Given a graph $G$ with a set $\pi(v)$ of feasible degrees, called a degree constraint, for each vertex $v$ of $G$, the general factor problem is to find a (spanning) subgraph $F$ of $G$ such that $\deg_F(v) \in \pi(v)$ for every $v$ of $G$. When all degree constraints are symmetric $\Delta$-matroids, the problem is solvable in polynomial time. The weighted general factor problem is  to find a general factor of the maximum total weight in an edge-weighted graph. 
In this paper, we present a strongly polynomial-time algorithm for a type of weighted general factor problems with real-valued edge weights that is provably not reducible to the weighted matching problem by gadget constructions.
\end{abstract}

\section{Introduction}

A matching in an undirected graph is a subset of the edges that have no vertices
in common, and it is perfect if its edges cover all vertices of the graph. Graph matching is one of the most studied problems both in graph theory and combinatorial
optimization, with beautiful structural results and efficient algorithms
described, e.g., in the monograph of Lov\'asz and
Plummer~\cite{lovasz2009matching} and in relevant chapters of standard
textbooks~\cite{schrijver2003combinatorial,korte2018combinatorial}. In particular,  the
\emph{weighted (perfect) matching problem} 
is to find a (perfect) matching of the maximum total weight for a given graph of which each edge is assigned a weight. 
This problem can be solved in polynomial time by the celebrated Edmonds' blossom
algorithm~\cite{edmonds1965maximum, edmonds1965paths}.
Since then, a number of more  efficient algorithms have been
developed~\cite{gabow1974implementation, lawler2001combinatorial, karzanov1976efficient, cunningham1978primal, gabow1985scaling, galil1986ev, gabow1989efficient, gabow1990data, gabow1991faster, huang2012efficient}.
Table \Rmnum{3} of~\cite{duan2014linear}  gives a detailed review of these algorithms.

The \emph{$f$-factor problem} is a generalization of the perfect matching problem in
which one is given a non-negative integer $f(v)$ for each vertex $v\in V$ of
$G=(V,E)$. 
The task is to find a (spanning) subgraph $F=(V_F, E_F)$ of $G$ 
such that $\deg_F(v)=f(v) $ for every $v\in V$.\footnote{In graph theory, a graph factor is usually a spanning subgraph. 
Here, without causing ambiguity, we allow $F$ to be an arbitrary subgraph including the empty graph and 
we adapt the convention that $\deg_F(v)=0$ if $v\in
V\setminus V_F$.}
The case $f(v)=1$ for every $v\in V$
is the perfect matching problem.
This problem, as well as the weighted
version, can be solved efficiently by a gadget reduction to the perfect matching
problem~\cite{edmonds70}.
In addition,
Tutte gave a characterization of graphs having an $f$-factor~\cite{tutte1952factors}, which generalizes his characterization theorem for perfect matchings~\cite{tutte1950factorization}. 
Subsequently, the study of graph factors 
has attracted much attention with
many variants of graph factors, e.g., $b$-matchings, $[a,b]$-factors, $(g, f)$-factors, parity $(g, f)$-factors, and anti-factors introduced, and various types of characterization theorems proved for the existence of such factors.
We refer the reader to the book~\cite{akiyama2011factors} and the survey~\cite{plummer2007graph}
for a comprehensive treatment of the developments on the topic of graph factors. 

In the early 1970s, Lov\'asz introduced  a generalization of the above factor
problems~\cite{lovasz70,lovasz1972factorization}, for which we will need a few
definitions. For any nonnegative integer $n$, let $[n]$ denote $\{0, 1, \ldots,
n\}$. A degree constraint $D$ of arity $n$ is a subset of $[n]$.\footnote{We
always associate a degree constraint with an arity. Two degree constraints are
different if they have different arities although they may be the same set of
integers.} We say that a degree constraint $D$ has a gap of length $k$ if there
exists $p\in D$ such that $p+1, \ldots, p+k\notin D$ and $p+k+1\in
D$.
An instance of the \emph{general factor
problem} (GFP)~\cite{lovasz70,lovasz1972factorization} 
 is given by a graph $G=(V,E)$
and a mapping $\pi$ that maps every vertex $v\in V$ to a degree constraint
$\pi(v)\subseteq[\deg_G(v)]$ of arity $\deg_G(v)$. The task is to find a
subgraph, if one exists, $F$ of $G$ such that $\deg_F(v)\in \pi(v)$ for every
$v\in V$.
The case $\pi(v)=\{0, 1\}$ for every $v\in V$ is the matching problem, and the case $\pi(v)=\{1\}$ for every $v\in V$ is the perfect matching problem.
Cornu{\'e}jols showed that the GFP
is
solvable in polynomial time if each degree constraint has gaps of length at most 1~\cite{cornuejols1988general}.
When a degree constraint having a gap of length at least $2$ occurs, the GFP is
NP-complete~\cite{lovasz1972factorization,cornuejols1988general} except for the
case when all constraints are either \emph{$0$-valid} or \emph{$1$-valid}. 
A degree constraint $D$ of arity $k$ is $0$-valid if $0\in D$,  and $1$-valid if $k\in D$.
When all constraints are \emph{$0$-valid}, the empty graph is a factor.
When all constraints are \emph{$1$-valid}, the underlying graph is a factor of
itself. In both cases, the GFP is trivially tractable. 

In this paper, we consider the \emph{weighted  general factor problem} (WGFP) where each edge is assigned a real-valued weight and the task is to find a general factor of the maximum total weight. 
We suppose that each degree constraint has  gaps of length at most $1$ for which the unweighted GFP is known to be polynomial-time solvable.
Some cases of the WGFP are reducible to the weighted matching or perfect matching  problem 
by gadget constructions, and hence are polynomial-time solvable.
In these cases, the degree constraints are called \emph{matching-realizable} (see Definition~\ref{def:matchgate}).
For instance, the degree constraint $D=[b]$  where $b>0$, for $b$-matchings is matching realizable~\cite{tutte1954short}.
The weighted $b$-matching problem is interesting in its own right in combinatorial optimization and has been well studied with many elaborate algorithms developed~\cite{pulleyblank1973faces, marsh1979matching, gabow1983efficient, anstee1987polynomial, gabow2013algebraic}. 
Besides $b$-matchings, Cornu{\'e}jols showed that
the \emph{parity interval} constraint $D=\{g, g+2, \ldots, f\}$ where $f\geq g\geq 0$ and $f\equiv g \mod 2$, 
is matching realizable~\cite{cornuejols1988general}, 
and Szab{\'o} showed that the \emph{interval} constraint $D=\{g, g+1, \ldots, f\}$  where $f\geq g\geq 0$, for  $(g, f)$-factors is matching realizable~\cite{Szabo2009}.
Thus, the WGFP
where each degree constraint is an interval or a parity interval is reducible to the weighted matching problem (with some vertices  required to have degree exactly $1$) and hence solvable in polynomial-time by Edmonds' algorithm, although Szab{\'o} gave a different algorithm for this problem~\cite{Szabo2009}.
By reducing the WGFP with interval and parity interval constraints to the weighted $(g, f)$-factor problem, 
a faster algorithm  was obtained in~\cite{Dudycz2018} based on Gabow's algorithm~\cite{gabow1983efficient}.

In~\cite{Szabo2009}, Szab{\'o} further conjectured that the WGFP is solvable in polynomial time without 
requiring each degree constraint should be an interval or a parity interval, as long as each degree constraint has gaps of length at most $1$.
To prove the conjecture, 
a natural question is then the following: \emph{Are there other WGFPs that are polynomial-time solvable by a gadget reduction to weighted matchings?}
In other words, \emph{are there other degree constraints that are matching realizable?}
In this paper, we show that the answer is \emph{no}.
\begin{theorem}
A degree constraint with  gaps of length at most $1$  is matching realizable if and only if it is an interval or a parity interval.
\end{theorem}

\smallskip\noindent\textbf{Previous results beyond matchings realizable degree constraints}\quad
With the answer to the above question being negative, new algorithms need to be devised for the WGFP with degree constraints that are not intervals or parity intervals. 
Unlike the weighted matching problem and the weighted $b$-matching problem for which various types of algorithms have been developed,
only a few algorithms have been presented for the more general and challenging WGFP:
For the cardinality version of WGFP, i.e., the WGFP where each edge is assigned weight $1$, 
Dudycz and Paluch introduced a polynomial-time algorithm for this problem with degree constrains having gaps of length at most $1$, which leads to a pseudo-polynomial-time algorithm for the WGFP with non-negative integral edge weights~\cite{Dudycz2018}. 

The algorithm in~\cite{Dudycz2018} is based on a structural result showing that if a factor is not optimal, then a factor of larger weight can be found  by a local search,
which can be done in polynomial time. 
However, it is not clear how much larger the weight of the new factor is. 
In order to get an optimal factor, the algorithm needs to repeat local searches iteratively until no better factors can be found, and the number of local searches is bounded by the total edge weight, which makes the algorithm pseudo-polynomial-time. 
Later, in an updated version~\cite{dudycz2017optimal},  by carefully assigning edge weights, the algorithm was improved to be weakly polynomial-time with a running
time $O(\log W mn^6)$, where $W$ is the largest edge weight, $m$ is the number of edges and $n$ is the number of vertices.
Later, Kobayash extended the algorithm to a more general setting called jump system intersections~\cite{kobayashi2023optimal}.

When degree constraints do not have gap integers of different
parities\footnote{An integer $p$ is a gap integer of a constraint $D$ if $p\notin D$ and
there exist some integers $p_1$ and $p_2$ such that $p_1<p<p_2$ and $p_1, p_2\in
D$.}, Szab\'{o} gives a strongly polynomial  time algorithm for the WGFP with rational weights~\cite{Szabo2009}.
The algorithm is based on a reduction to the weighted $H$-factor problem, which
can be solved by a primal-dual
algorithm~\cite{pap2004tdi}.\footnote{\cite{Szabo2009} claims the result for
real weights while Pap's result only works for integer weights but we believe it
can be extended to the rationals.}

\smallskip\noindent\textbf{Our main contribution}\quad 
In this paper, we give another \emph{strongly} polynomial-time algorithm for WGPF with real-valued edge weights. 
Let $p\geq 0$ be an arbitrary integer. Consider the following two types of degree
constraints $\{p, p+1, p+3\}$ and $\{p, p+2, p+3\}$ (of arbitrary arity). We will call them  \emph{type-1}
and  \emph{type-2} respectively.
These are the ``smallest'' degree constraints
that are provably not  matching realizable.

\begin{theorem}\label{thm:tract}
There is a strongly polynomial-time algorithm
for the $\wgfp$ with real-valued edge weights where each degree constraint is an interval, a parity interval, a
  type-1, or a type-2 (of arbitrary arities).
  The algorithm runs in time $O(n^6)$
  for a graph with $n$ vertices. 
\end{theorem}

We remark that for type-1 and type-2 constraints, there are no gap integers of different parities. Thus, the WGPF with type-1 and type-2 constraints and rational edge weights are known to be strongly-polynomial time solvable by a reduction to the weighted $H$-factor problem. 
However, our result gives a direct  algorithm for WGFP based on combinatorial structures of graph factors, which works for all real-valued edge weights. 

The requirement  of  degree constraints in our result may look overly specific. 
However,  the scope of our algorithm is not narrow. 
First, our result implies a complexity dichotomy for the WGFP on all subcubic graphs (see the following Theorem~\ref{thm:dich}), which for many is a large and interesting class of graphs. 
More importantly, there are interesting problems arising from applied areas that are encompassed by the WGPF with constraints considered in this paper. 

For instance, the \emph{terminal backup problem} from network design is the following problem. Given a graph consisting of terminal nodes, non-terminal nodes, and edges with non-negative costs. 
The goal is to find a subgraph with the minimum total cost such that each
terminal node is connected to at least one other terminal node (for the purpose
of backup in applications). 
It is known that an optimal solution of the terminal backup problem consists of edge-disjoint  paths containing 2 terminals and stars containing 3 terminals~\cite{xu2008survivable}. 
In other words, in an optimal subgraph of the terminal backup problem, each terminal node has degree $1$ and each non-terminal node has degree $0$, $2$ or $3$. 
Thus, the terminal backup problem can be expressed as a WGFP with degree constraints $\{1\}$ and $\{0, 2, 3\}$ (both of arbitrary arities). 
A weakly polynomial-time algorithm was given for the terminal backup problem in~\cite{anshelevich2011terminal}. Our result gives a  strongly polynomial-time algorithm for this problem.

Our algorithm is a recursive algorithm, reducing the problem to a smaller sub-problem of itself by fixing the parity of degree constraints on vertices.
Its correctness is based on a delicate structural result, which is stronger than that of~\cite{dudycz2017optimal}.\footnote{The result in~\cite{dudycz2017optimal} holds for the more general WGFP with  all degree constraints having gaps of length at most $1$, while our result only works for the WGFP with interval, parity interval, type-1 and type-2 degree constraints.}
 Equipped with this result, 
 our algorithm can directly find an \emph{optimal} factor (not just a better one) of an instance of a larger size by performing only one local search from an optimal factor of a smaller instance. 
 Here, the important part is not how to find a better factor by local search (the
 main result of \cite{dudycz2017optimal}) but rather how to ensure that the
 better factor obtained by only one local search is actually optimal under certain assumptions. 
 This is the key to making our algorithm strongly polynomial.
In addition, as a by-product, we give a simple proof of the  result of~\cite{dudycz2017optimal} for the special case of WGFP with interval, parity interval, type-1 and type-2 degree constraints by reducing the problem to WGFP on subcubic graphs and utilizing the equivalence between 2-vertex connectivity and 2-edge connectivity of subcubic graphs.


%

\smallskip\noindent\textbf{Relation with (edge) constraint satisfaction problems}\quad
The graph factor problem is encompassed by a special case of the Boolean
constraint satisfaction problem (CSP), called edge-CSP, in which
every variable appears in exactly two constraints~\cite{istrate1997looking, Feder01:tcs}.
When every constraint is symmetric (i.e, the value of the constraint only depends on the Hamming weight of its input), the Boolean edge-CSP is a graph factor problem. 

For general Boolean edge-CSPs, 
Feder showed that the problem is NP-complete if a constraint that is not a
$\Delta$-matroid occurs, except for those that are tractable by Schaefer's
dichotomy theorem for Boolean CSPs~\cite{schaefer1978complexity}. 
In a subsequent
line of work~\cite{Dalmau03:mfcs,geelen2003linear,Feder06:sidma,Dvorak15:icalp},
tractability of Boolean edge-CSPs has been established for special classes of
$\Delta$-matroids, most recently for even $\Delta$-matroids~\cite{kkr2018even}.
A complete complexity classification of Boolean edge-CSPs is still open with the conjecture that all $\Delta$-matroids are tractable. 
A degree constraint (i.e., a symmetric constraint) is a $\Delta$-matroid if and only if it has  gaps of length at most $1$.
Thus, the above conjecture holds for symmetric Boolean edge-CSPs by Cornu{\'e}jols' result on the general factor problem~\cite{cornuejols1988general}.
A complexity classification for weighted Boolean edge-CSPs is certainly a more challenging goal: 
The complexity of weighted Boolean edge-CSPs with
even $\Delta$-matroids as constraints is still open.
%
Our result in Theorem~\ref{thm:tract} gives a tractability result for weighted
Boolean edge-CSPs with certain
symmetric $\Delta$-matroids as constraints.
Combining our main result with known results on Boolean valued CSPs~\cite{cohen2006complexity}, 
we obtain a complexity dichotomy for weighted Boolean edge-CSPs with symmetric constraints of arity no more than $3$, i.e., the WGFP on subcubic graphs. 

Let $D$ be a degree constraint of arity at most $3$. If $D\neq\{0, 3\}$, then $D$ is an interval, a parity interval, a type-1, or a type-2. Thus, if the constraint $\{0, 3\}$ (of arity 3) does not occur, then the WGFP is strongly polynomial-time solvable by our main theorem.
Otherwise, the constraint $\{0, 3\}$ occurs. 
In this case, for a vertex $v$ labeled by $\{0, 3\}$,  the three edges incident
to it must take the same assignments in a feasible factor (i.e., the three edges
are all either present or absent in any factor).
Thus, the vertex $v$ can be viewed as a Boolean variable  and it appears in three other degree constraints connected to it. 
By viewing all vertices with $\{0, 3\}$ as variables appearing three times and the other edges as variables appearing twice, 
the WGFP becomes a special case of valued CSPs where some variables appear three times and the other variables appear twice.
It is known that once variables are allowed to appear three times in a CSP, then
they can appear arbitrarily many times \cite{Dalmau03:mfcs}. 
Thus, the WGPF with $\{0, 3\}$ occurring is equivalent to a standard (non-edge) CSP~\cite{Feder98:monotone}.
By the dichotomy theorem for valued CSPs~\cite{cohen2006complexity}, one can check that the problem is tractable if and only if  for every degree constraint $D$ of arity $k\leq 3$, $D\subseteq \{0, k\}$. Thus, we have the following complexity dichotomy.

\begin{theorem}\label{thm:dich} 
The $\wgfp$ on subcubic graphs is strongly polynomial-time solvable if \begin{enumerate}
    \item the degree constraint $\{0,3\}$ of arity $3$ does not occur, 
    \item or for every degree constraint $D$ of arity $k\leq 3$, $D\subseteq \{0, k\}$.
\end{enumerate} Otherwise, the problem is NP-hard.
\end{theorem} 


\smallskip\noindent\textbf{Organization}\quad 
In Section~\ref{sec:pre}, we present basic definitions and notation. 
In Section~\ref{sec:alg}, we describe our algorithm and give a structural result
for the WGFP that ensures the correctness and the strongly polynomial-time running time of our algorithm. 
In Section~\ref{sec:thm}, we introduce basic augmenting subgraphs as an analogy of augmenting paths for weighed matchings and give a proof of the structural result. The proof is based on a result regarding the existence of certain basic factors for subcubic graphs, for which we give a proof sketch in Section~\ref{sec:thm1}.
Finally, we discuss matching realizability and its relation with $\Delta$-matroids in Appendix~\ref{sec:matroid}.

\section{Preliminaries}\label{sec:pre}

Let $\sD$ be a (possibly infinite) set of degree constraints. 

\begin{definition}
The \emph{weighted general factor problem} parameterized by
  $\sD$, denoted by $\wgfp(\sD)$, is the following computational problem. An
  instance is a triple $\Omega=(G, \pi, \omega)$, where $G=(V, E)$ is a graph, 
 $\pi: V\to\sD$ assigns to every $v\in V$ a degree constraint
  $D_v\in\sD$ of arity $\deg_G(V)$, 
  and $\omega:E\to\RR$ assigns to every $e\in E$ a real-valued weight $w(e)\in\RR$. The
  task is to find, if one exists, a general factor $F$ of $G$  such that the
  total weight of edges in $F$ is maximized. 
  
  The \emph{general factor problem}  $\gfp(\sD)$ is 
 the decision version of $\wgfp(\sD)$; i.e., deciding whether a general factor
  exists or not.
\end{definition}

Suppose that $\Omega=(G, \pi, \omega)$ is a WGFP instance.
If $F$ is a general factor of $G$ under $\pi$, then we say that $F$ is a factor of $\Omega$, denoted by $F\in \Omega$.
In terms of this inclusion relation, $\Omega$ can be viewed as a set of subgraphs of $G$.
We extend the edge weight function $\omega$ to subgraphs of $G$.
For a subgraph $H$ of $G$, its weight $\omega(H)$ is $\sum_{e\in E(H)} \omega(e)$ ($\omega(H)=0$ if $H$ is the empty graph). 
If $H$ contains an isolated vertex $v$, then $\omega(H)=\omega(H')$, where $H'$ is the graph obtained from $H$ by removing $v$.
Moreover, $H\in \Omega$ if and only if $H'\in \Omega$. 
In the following, without other specification, we always assume that a factor does not contain any isolated vertices. 
The optimal value of $\Omega$, denoted by $\opt(\Omega)$, is $\max_{F\in \Omega} \omega(F)$.
We define $\opt(\Omega)=-\infty$ if $\Omega$ has no factor. 
A factor $F$ of $\Omega$ is \emph{optimal} in $\Omega$ if $\omega(F)=\opt(\Omega)$.
For a  WGFP instance $\Omega'=(G', \pi', \omega')$, 
where $G'\subseteq G$\footnote{We use the term
``subgraph'' and notation $G'\subseteq G$ throughout for the standard meaning of a ``normal'' subgraph i.e., if $G=(V', E')$ and $G=(V, E)$ then $G'\subseteq G$ means $V'\subseteq V$ and $E'\subseteq E$.} and $\omega'$ is the restriction of $\omega$ on the edges of $G'$,
we say $\Omega'$ is a \emph{sub-instance} of $\Omega$, 
denoted by $\Omega'\subseteq \Omega$, if $F\in \Omega$ for every  $F\in \Omega'$.
In particular, $\Omega'$ is a subset of $\Omega$ by viewing them as two sets of subgraphs of $G$.
If $\Omega'\subseteq \Omega$, then $\opt(\Omega')\leq \opt(\Omega)$. 

For two WGFP instances $\Omega_1=(G, \pi_1, \omega)$ and $\Omega_2=(G, \pi_2, \omega)$, we use $\Omega_1\cup \Omega_2$ to denote the union of factors of these two instances, i.e.,  $\Omega_1\cup \Omega_2=\{F\subseteq G \mid F\in \Omega_1 \text{ or }  F\in \Omega_2\}$, and $\Omega_1\cap \Omega_2$ to denote the intersection, i.e.,  $\Omega_1\cap \Omega_2=\{F\subseteq G \mid F\in \Omega_1 \text{ and }  F\in \Omega_2\}$.
Note that  $\Omega_1\cup \Omega_2$ and  $\Omega_1\cap \Omega_2$  are just sets
of subgraphs of $G$ and may not define  WGFP instances on $G$.

We use $\mathscr{G}_1$ and $\mathscr{G}_2$ to denote the set of  degree constraints that are intervals
and parity intervals, respectively, and
 $\mathscr{T}_1$ and $\mathscr{T}_2$ to denote the set of degree constraints
 that are  type-1 and type-2, respectively.
Let $\mathscr{G}=\mathscr{G}_1\cup \mathscr{G}_2$ and  $\mathscr{T}= \mathscr{T}_1\cup  \mathscr{T}_2$.
In this paper, we study the problem  $\wgfp(\mathscr{G}\cup\mathscr{T})$.

Let $H_1=(V_1, E_1)$ and $H_2=(V_2, E_2)$ be two subgraphs of $G$.
The symmetric difference graph $H_1\Delta H_2$ is the induced subgraph of $G$ induced
by the edge set $E_1\Delta E_2$. Note that there are no isolated vertices in a symmetric difference graph.
When $E_1\cap E_2=\emptyset$, we may write  $H_1\Delta H_2$ as $H_1\cup H_2$. 
When $E_2\subseteq E_1$, we may write  $H_1\Delta H_2$ as $H_1\backslash H_2$.

A \emph{subcubic} graph
is defined to be a graph where every vertex has degree $1, 2$ or $3$.
Unless stated otherwise,
we use $V_G$  and $E_G$ to denote the vertex set and the edge set of a graph
$G$, respectively.

\begin{definition}[2-vertex-connectivity]
A connected graph $G$ is 2-vertex-connected (or $2$-connected) if it has more than $2$ vertices and remains connected by removing any vertex.  
\end{definition}

Menger's Theorem gives an equivalent definition of 2-connectivity, cf.~\cite{diestel10:graph} for a proof.

\begin{theorem}[Menger's Theorem]
A connected graph $G$ is $2$-connected if and only if for any two vertices of $G$, there exists two vertex disjoint paths connecting them (i.e., there is a cycle containing these two vertices). 
\end{theorem}

\begin{definition}[Bridge and 2-edge-connectivity]
A bridge of a connected graph is an edge whose deletion makes the graph disconnected.
A connected graph is 2-edge-connected  if it has no bridge. 
\end{definition}

The following theorem is the edge version of Menger's Theorem.

\begin{theorem}
A connected graph $G$ is $2$-edge-connected if and only if for any two vertices of $G$, there exists two edge disjoint paths connecting them. 
\end{theorem}

If two paths connecting a pair of vertices are vertex-disjoint, then they are also edge-disjoint.
Thus, 2-vertex-connectivity implies 2-edge-connectivity.
For subcubic graphs, one can check that 
two edge-disjoint paths are also vertex-disjoint.
Thus, for subcubic graphs, 2-vertex-connectivity is equivalent to 2-edge-connectivity.
In particular, we have the following result. 

\begin{lemma}\label{lem:bridge}
If a connected subcubic graph is not 2-connected, then it contains a bridge. 
\end{lemma}

The following fact regarding 2-connected  graphs will also be used. 
\begin{lemma}\label{lem:extra-path}
Let $G=(V_G, E_G)$ be a 2-connected  graph, $H=(V_H, E_H)\subseteq G$,
and $u\in V_H$. If $\deg_H(u)=2<\deg_G(u)=3$, then there exists 
a path $p_{uw}=(V_{p_{uw}}, E_{p_{uw}})\subseteq G$ with endpoints $u$ and $w$ for some $w\in V_H$   such that $E_{p_{uw}}\cap E_H=\emptyset$.
\end{lemma}
\begin{proof}
 Since $\deg_H(u)=2<\deg_G(u)=3$, there is an edge $e_{vu}=(v, u)\in E_G$ incident to $u$ such that $e_{vu}\notin E_H$.
 If $v\in V_H$, then the edge $e_{vu}$ is the desired path. 
 Thus, we may assume that $v\notin V_H$.
Since $G$ is 2-connected, there is a path $p_{vu}$ with endpoints $v$ and $u$ such
that $e_{vu}\notin E_{p_{vu}}$.
Since $u\in V_H$, $V_{p_{vu}}\cap V_H\neq \emptyset$.
Let $w$ be the first vertex in the path $p_{vu}$ (within the order of traversing the path from $v$ to $u$) belonging to $V_H$.
Then, $w\neq u$ since $e_{vu}\notin E_{p_{vu}}$ and $\deg_G(u)=3$.
Also, $w\neq v$ since $v\notin V_H$.
Let $p_{vw}\subsetneq p_{vu}$ be the segment with endpoints $v$ and $w$.
Then, $E_{p_{vw}}\cap E_H=\emptyset$.
Let $p_{uw}$ be the path consisting of $e_{vu}$ and $p_{vw}$.
It has endpoints  $u, w\in V_H$, and $E_{p_{uw}}\cap E_H=\emptyset$.
\end{proof}

\section{Algorithm}\label{sec:alg}

We give a recursive algorithm for the problem $\wgfp(\mathscr{G}\cup\mathscr{T})$, using the problems $\wgfp(\mathscr{G})$ and the decision problem $\gfp(\mathscr{G}\cup\mathscr{T})$ as oracles.
Given an instance $\Omega=(G, \pi, \omega)$ of $\wgfp(\mathscr{G}\cup\mathscr{T})$, 
we define the following sub-instances of $\Omega=(G, \pi, \omega)$ that will be used in the recursion.
Recall that $V_G$ denotes the vertex set of the underlying graph $G$.
Let $T_{\Omega}$  denote the set $\{v\in V_G\mid \pi(v)\in \mathscr{T}\}$.
(We may omit the subscript ${\Omega}$ of $T_{\Omega}$ when it is clear from the context.)

For every vertex $v\in T_{\Omega}$, we split the instance $\Omega$ in two by
splitting the degree constraint $\pi(v)$ in two parity intervals. More
precisely, we define 
\begin{align*}
    &D^0_v=\{p_v+1, p_v+3\}  \text{ ~and~ }  D^1_v=\{p_v\}   &&\text{if~~~ } \pi(v)=\{p_v, p_v+1, p_v+3\}\in \mathscr{T}_1;\\
   &D^0_v=\{p_v, p_v+2\}  \text{ ~and~ }  D^1_v=\{p_v+3\}   &&\text{if~~~ } \pi(v)=\{p_v, p_v+2, p_v+3\}\in \mathscr{T}_2.
\end{align*}
We have $D_{v}^0, D_v^1\in \mathscr G_2$.
For $i\in \{0, 1\}$  and $v\in T_{\Omega}$, we define $\Omega_v^i=(G, \pi_v^i, \omega)$ to be the sub-instance of $\Omega$ where $\pi_v^i(x)=\pi(x)$ for every $x\in V_G\backslash \{v\}$ and  $\pi_v^i(v)=D_v^i$.
Then, for every $v\in T_{\Omega}$, we have $\Omega_v^0\cap \Omega_v^1=\emptyset$ and $\Omega_v^0\cup \Omega_v^1=\Omega$.
Moreover, $T_{\Omega_v^0}=T_{\Omega_v^1}=T_{\Omega}\backslash\{v\}$.

Let $F$ be a factor of $\Omega$.
Similarly to above, one can partition $\Omega$  into $2^{|T_\Omega|}$ many
sub-instances according to $F$ such that  each one is an instance of
$\wgfp(\mathscr{G})$ -- for each $v\in T_\Omega$, we choose one of the two splits of
$\pi(v)$ as above. (We note that the algorithm will not consider all exponentially many sub-instances.)
In detail, for  every vertex $v\in T_\Omega$, we define $D^F_v=D^i_v$ where
$\deg_F(v)\in D^i_v$ as follows:
\begin{align*}
    &D^F_v=\{p_v\}  &&\text{if~~~ } \pi(v)=\{p_v, p_v+1, p_v+3\}\in \mathscr{T}_1 \text{ ~and~ } \deg_F(v)=p_v, \\
    &D^F_v=\{p_v+1, p_v+3\}  &&\text{if~~~ } \pi(v)=\{p_v, p_v+1, p_v+3\}\in \mathscr{T}_1 \text{ ~and~ } \deg_F(v)\neq p_v; \\
    & D^F_v=\{p_v+3\}  &&\text{if~~~ } \pi(v)=\{p_v, p_v+2, p_v+3\}\in \mathscr{T}_2 \text{ ~and~ } \deg_F(v)=p_v+3, \\
   &D^F_v=\{p_v, p_v+2\}  &&\text{if~~~ } \pi(v)=\{p_v, p_v+2, p_v+3\}\in \mathscr{T}_2 \text{ ~and~ } \deg_F(v)\neq p_v+3.
\end{align*}
By definition, 
$\deg_F(v)\in D_v^F\subseteq \pi(v)$ and  $D_v^F\in \mathscr{G}_2$.
In fact, $D^F_v$ is the maximal set such that $\deg_F(v)\in D_v^F\subseteq \pi(v)$ and  $D_v^F\in \mathscr{G}_2$.
One can also check that for every $v\in T$,  $\pi(v)\backslash D_v^F\in \mathscr{G}_2$, and moreover for every $p\in D^F_v$ and $q\in \pi(v)\backslash D_v^F$, $p \not\equiv q \mod 2$. 

For every $W\subseteq T_{\Omega}$, we define $\Omega^F_W=(G, \pi^F_W, \omega)$ to be the sub-instance of $\Omega$ where 
\begin{equation}\label{equ:omega-f-w}
\begin{aligned}
&\pi^F_W(v)=\pi(v)\backslash D_v^F && \hspace{3ex}\text{for~~~} v\in W,\\
&\pi^F_W(v)= D_v^F&& \hspace{3ex}\text{for~~~} v\in T_{\Omega}\backslash W,\\
&\pi^F_W(v)=\pi(v) && \hspace{3ex}\text{for~~~} v\in V\backslash T_{\Omega}.
\end{aligned}
\end{equation}
By definition, for every $W$, $\Omega_W^F$ is an instance of $\wgfp(\mathscr{G})$. 
Moreover, we have $\cup_{W\subseteq T}\Omega_{W}^F=\Omega$ and $\Omega_{W_1}^F\cap \Omega_{W_2}^F=\emptyset$
for every $W_1 \neq W_2$. 
Thus, $\{\Omega^F_W\}_{W\subseteq T_\Omega}$ is a partition of $\Omega$ (viewed as a set of subgraphs of $G$).
When $W=\emptyset$, we write $\Omega_W^F$ as $\Omega^F$, and when $W=\{s\}$ or $W=\{s, t\}$, 
we write  $\Omega_W^F$ as $\Omega^F_s$ or $\Omega^F_{s,t}$ respectively for simplicity. 

Our algorithm is given in Algorithm~\ref{alg-main}.
\begin{algorithm}[!h]\label{alg-main}
\caption{Finding an optimal factor for an instance of $\wgfp(\mathscr{G}\cup\mathscr{T})$}
\SetKwInOut{Input}{Input}\SetKwInOut{Output}{Output}\SetKwInOut{Return}{Return}
\SetKwFunction{DEC}{Decision}\SetKwFunction{OPT}{Optimization}
\SetKwFunction{Main}{Main}
\SetKwProg{Fn}{Function}{:}{}
\Fn{\DEC}{
\BlankLine
\Input {An instance $\Omega=(G, \pi, \omega)$ of $\wgfp(\mathscr{G}\cup\mathscr{T})$.}
\Output {A factor of $\Omega$, or ``No'' if $\Omega$ has no factor.}
}

\BlankLine

\Fn{\OPT}{
\BlankLine
\Input {An instance $\Omega=(G, \pi, \omega)$ of $\wgfp(\mathscr{G})$.}
\Output {An optimal factor of $\Omega$, or ``No'' if $\Omega$ has no factor.}
}

\BlankLine

\Fn{\Main}{
\BlankLine
\Input {An instance $\Omega=(G, \pi, \omega)$ of $\wgfp(\mathscr{G}\cup\mathscr{T})$.}
\Output {An optimal factor $F\in \Omega$, or ``No'' if $\Omega$ has no factor.}
\BlankLine
$T \leftarrow \{v\in V\mid \pi(v)\in \mathscr{T}\}$\;
\eIf{$T$ {\rm is the empty set}}
{\KwRet \OPT($\Omega$)\;}
{Arbitrarily pick $u\in T$\;
\eIf{\DEC$(\Omega^0_u)$ {\rm returns ``No''}}
{\KwRet \Main($\Omega^1_u$)\;}
{$F^{\rm opt} \leftarrow$ \Main($\Omega^0_u$)\;
\ForEach{$v\in T$}{\tcp*[h]{Elements of $T$ can be traversed in an arbitrary order.}\\
$W\leftarrow \{u\}\cup\{v\}$\;
\lIf{$\OPT(\Omega^{F^{\rm opt}}_W) \neq $  {\rm``No''}}{$F'\leftarrow \OPT(\Omega^{F^{\rm opt}}_W)$}
\lIf{$\omega(F')>\omega(F^{\rm opt})$}
{$F^{\rm opt}\leftarrow F'$}
}
\KwRet $F^{\rm opt}$;} 
}}
\end{algorithm}

The key that makes our algorithm running in strongly polynomial time is the  following structural result (Theorem~\ref{Thm:local-global-optimal}) for the problem WGFP$(\mathscr{G}\cup \mathscr{T})$.
It says that given an optimal factor $F$ of $\Omega_u^0$ for some $u\in T_\Omega$, if $F$ is not optimal in $\Omega$, 
then we can directly find an optimal factor of $\Omega$ by searching at most $n$ sub-instances of $\Omega$ which are in $\wgfp(\mathscr{G})$.
Note that the number of searches is independent of the edge weights. 
Thus, the problem of finding an optimal factor in $\Omega$  can be reduced to finding an optimal factor in $\Omega_u^0$, where there is one fewer vertex $u$ with constraints in $\mathscr{T}$. 
By recursively reducing an instance to another with fewer vertices  with constraints in $\mathscr{T}$, we eventually get an instance of $\wgfp(\mathscr{G})$ 
which can be solved in polynomial-time.  
This leads to a strongly polynomial time algorithm for finding an optimal factor.

\begin{theorem}\label{Thm:local-global-optimal}
Suppose that  $\Omega=(G, \pi, \omega)$ is an instance of $\wgfp(\mathscr{G}\cup\mathscr{T})$, $F$ is a factor of $\Omega$ and  $F$ is optimal in $\Omega_u^0$ for some $u\in T_{\Omega}$.
Then a factor $F'$ is optimal in $\Omega$ if and only if  $\omega(F')\geq
  \omega(F)$ and  $\omega(F')\geq \opt(\Omega^F_W)$ for every $W$ where $u\in
  W\subseteq T_{\Omega}$ and $|W|=1$ or $|W|=2$.

In other words, if $F$ is not optimal in $\Omega$, then there is an optimal
  factor of $\Omega$ which belongs to $\Omega^F_W$ for some $W$ where $u\in
  W\subseteq T_{\Omega}$ and $|W|=1$ or $|W|=2$.

\end{theorem}

\begin{remark}
This result is \emph{stronger} than the main result (Theorem 2) of \cite{dudycz2017optimal}, and it is \emph{not} simply implied by \cite{dudycz2017optimal}.
To clarify this, we give a simple proof outline of Theorem~\ref{Thm:local-global-optimal} here.

In order to prove Theorem~\ref{Thm:local-global-optimal}, it suffices to prove the direction that  if  $\omega(F')\geq
  \omega(F)$ and  $\omega(F')\geq \opt(\Omega^F_W)$ for every $W$ where $u\in
  W\subseteq T_{\Omega}$ and $|W|=1$ or $2$, then $F'$ is optimal in $\Omega$.  We prove this by contradiction. 
  Suppose that $F'$ is not optimal in $\Omega$, and 
$F^\ast$ is an optimal factor of $\Omega$. 
Then, $\omega(F^\ast)>\omega(F')\geq\opt(\Omega^F_W)$ for every  $W\subseteq T_\Omega$ where $|W|\leq 2$.
Also, $\omega(F^\ast)\notin \Omega_u^0$ since $\omega(F^\ast)>\omega(F) =\opt(\Omega_u^0)$. Thus, $\deg_{F^\ast}(u)\not\equiv \deg_F(u) \mod 2$.

By \cite{dudycz2017optimal}, a \emph{canonical path} $M\subseteq F\Delta F^\ast$ with positive weight\footnote{See definition 3 of \cite{dudycz2017optimal}. They are defined as basic augmenting subgraphs (Definition~\ref{def-basic-argumenting}) in this paper.} can be found, and 
then $F\Delta M$ is a factor of $\Omega$ with larger weight than $F$ and $F\Delta M\in \Omega^F_W$ for some $W\subseteq T_\Omega$ where $|W|\leq 2$. 
However, this does not lead to a contradiction. 
To get a contradiction, we need  to show that
the positive weighted canonical path  $M$ (a basic augmenting subgraph) further satisfies
$\deg_M(u)\equiv 0 \mod 2$.
Then, $\deg_{F\Delta M}(u)\equiv \deg_F(u) \mod 2$.
Thus, $F\Delta M$ is a factor with larger weight than $F$ and $F\Delta M \in \Omega_u^0$, which contradicts with $F$ being optimal in $\Omega_u^0$.

The existence of a basic augmenting subgraph  $M$  satisfying $\deg_M(u)\equiv 0 \mod 2$ is formally stated in the second property of Lemma~\ref{lem:basic-exist}. The main technical part of the paper (Section 5.2 of the full paper) is devoted to prove it. 
In Section~\ref{sec:thm1} of this short version, we give an example to illustrate the proof ideas. 
The existence of such a basic augmenting subgraph is highly non-trivial.
In fact, it does \emph{not} hold anymore after a subtle change of the condition ``$F$ is optimal in  $\Omega_u^{\bf 0}$''  to ``$F$ is optimal in $\Omega_u^{\bf 1}$'' for some $u\in T_{\Omega}$. 
We give the following example (see Figure~\ref{fig:1.cc}) to show this.

\begin{figure}[!htbp]
    \centering
    \includegraphics[height=1.6cm]{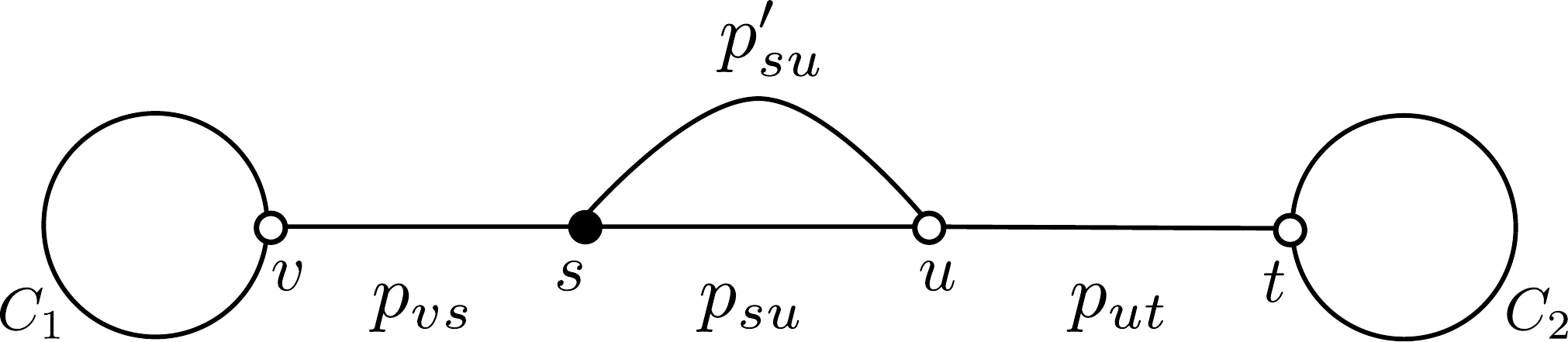}
    \caption{An example that violates Theorem~\ref{Thm:local-global-optimal} when $F$ is optimal in $\Omega_u^1$ instead of $\Omega_u^0$}
    \label{fig:1.cc}
\end{figure}

\vspace{-1ex}

In this instance, $\pi(u)=\pi(v)=\pi(t)=\{0, 1, 3\}$ (denoted by hollow nodes) and $\pi(s)=\{0, 2, 3\}$ (denoted by the solid node), and $\omega(C_1)=\omega(p_{vs})=\omega(p_{su})=\omega(p'_{su})=\omega(p_{ut})=\omega(C_2)=1.$
Inside the cycles $C_1$ and $C_2$, and the paths $p_{vs}$, $p_{su}$, $p_{ut}$, and $p'_{su}$, 
there are other vertices of degree $2$ with the degree constraint  $\{0, 2\}$ so that the graph $G$ is simple. 
We omit these vertices of degree $2$ in Figure~\ref{fig:1.cc}.
In this case, $T_\Omega=\{u,v,s,t\}$.
Consider the sub-instance $\Omega_u^1=(G, \pi_u^1, \omega)$.
We have $\pi_u^1(u)=D_u^1=\{0\}$ since $\pi(u)=\{0, 1, 3\}$.
Then, the only factor $F$ of $\Omega_u^1$ is the empty graph (assuming there are no isolated vertices in factors), and $F$ is not optimal in $\Omega$.
Also, the only optimal factor of $\Omega$ is the graph $G$ and $G\in \Omega_{T_{\Omega}}^F$ where $|T_{\Omega}|=4$. 
Clearly, $\deg_G(u)\not\equiv \deg_F(u) \mod 2$.
One can check that for any factor $F'$ of $\Omega$ with larger weight than $F$, $\deg_{F'}(u)\not\equiv \deg_F(u) \mod 2$.
In other words, there is no basic augmenting subgraph $M$ such that $\deg_M(u)\equiv 0 \mod 2$.
Moreover, one can check that in this case, 
Theorem~\ref{Thm:local-global-optimal} also does not hold.
In other words, the existence of a basic augmenting subgraph satisfying $\deg_M(u)\equiv 0 \mod 2$ is crucial for the correctness of Theorem~\ref{Thm:local-global-optimal}.

\end{remark}

Using Theorem~\ref{Thm:local-global-optimal}, we now prove that Algorithm~\ref{alg-main} is correct.

\begin{lemma}
Given an instance $\Omega=(G, \pi, \omega)$ of $\wgfp(\mathscr G, \mathscr{T})$,
Algorithm~\ref{alg-main} 
returns either an optimal factor of $\Omega$, or ``No'' if $\Omega$ has no factor. 
 \end{lemma}
 
 \begin{proof}
 Recall that for an instance $\Omega=(G, \pi, \omega)$, we define $T_\Omega=\{v\in V_G\mid \pi(v)\in \mathscr{T}\}$ where $V_G$ is the vertex set of $G$ .
 We prove the correctness by induction on the $|T_\Omega|$. 
 
 If $|T_\Omega|=0$, $\Omega$ is an instance of  $\wgfp(\mathscr G)$.
 Algorithm~\ref{alg-main} simply  returns \OPT$(\Omega)$.
 By the definition of the function \OPT, the output is correct. 
 
 Suppose that Algorithm~\ref{alg-main} returns correct results for all instances $\Omega'$ of $\wgfp(\mathscr G, \mathscr T)$ where $|T_{\Omega'}|=k$. 
 We consider an instance $\Omega$ of $\wgfp(\mathscr G, \mathscr T)$ where $|T_\Omega|=k+1$.
 Algorithm~\ref{alg-main} first calls the function \DEC$(\Omega_u^0)$ for some arbitrary $u\in T$. 
 
 We first consider the case that \DEC$(\Omega_u^0)$ returns ``No''. By the definition, $\Omega_u^0$ has no factor. 
 Moreover, since $\Omega=\Omega_u^0\cup \Omega_u^1$, we have $F\in \Omega$ if and only if $F\in \Omega_u^1$.
 Then, a factor $F\in \Omega_u^1$ is optimal in $\Omega$ if and only if it is optimal in $\Omega_u^1$. 
 Note that $\Omega_u^1$ is an instance of $\wgfp(\mathscr G, \mathscr T)$ where $|T_{\Omega_u^1}|=k$.
 By the induction hypothesis, Algorithm~\ref{alg-main} returns a correct result \Main$(\Omega_u^1)$ for the instance  $\Omega_u^1$, which is also a correct result for the instance $\Omega$. 
 
 Now, we consider the case that \DEC$(\Omega_u^0)$ returns  a factor  of $\Omega_u^0$.
 Then, \Main$(\Omega_u^0)$ returns an optimal factor $F$ of $\Omega_u^0$. 
 After the loop (lines 13 to 17)  in Algorithm~\ref{alg-main}, we get a factor
   $F^{\rm opt}$ of $\Omega$ such that $\omega(F^{\rm opt})\geq
   \opt(\Omega_W^F)$ for every $u\in W\subseteq T_{\Omega}$ where $|W|=1$ (when
   $u=v$) or $|W|=2$ (when $u\neq v$) and $\omega(F^{\rm opt})\geq \omega(F)$. 
 By Theorem~\ref{Thm:local-global-optimal}, $F^{\rm opt}$ is an optimal factor of $\Omega$.
 Thus,  Algorithm~\ref{alg-main} returns a correct result. 
 \end{proof}
 
 Now, we consider the time complexity of Algorithm~\ref{alg-main}.
  The size of an instance is defined to be the number of vertices of the underlying graph of the instance. 
 \begin{lemma}
 Run Algorithm~\ref{alg-main} on an instance $\Omega=(G, \pi, \omega)$ of size $n$. Then,
 \begin{itemize}
     \item the algorithm will stop the recursion after at most $n$ recursive steps;
     \item the algorithm will call  \DEC at most $n$ many times, call  \OPT at most $\frac{n(n+1)}{2}+1$ many times, and perform at most $\frac{n(n+1)}{2}$ many comparisons;
     \item the algorithm runs in time $O(n^6)$.
 \end{itemize}
 \end{lemma}
 \begin{proof}
 Let  $\Omega^k=\{G, \pi^k, \omega\}$ be the instance after $k$ many recursive steps. Here $\Omega^0=\Omega$.  
  Recall that $T_{\Omega^k}=\{v\in V\mid \pi^k(v)\in \mathscr T\}$.
    For an instance $\Omega^k$ with $|T_{\Omega^k}|>0$, the recursive step will then go to the instance $(\Omega^k)_u^0$ or $(\Omega^k)_u^1$ for some $u\in T_{\Omega^k}$. 
    Thus, $\Omega^{k+1}=(\Omega^k)_u^0$ or $(\Omega^k)_u^1$.
    In both cases, $T_{\Omega^{k+1}}=T_{\Omega^{k}}\backslash \{u\}$ and hence $|T_{\Omega^{k+1}}|=|T_{\Omega^{k}}|-1$.
      By design, the algorithm will stop the recursion and return \OPT$(\Omega^m)$ when it reaches an instance $\Omega^m$ with $|T_{\Omega^m}|=0$.
Thus, $\#{\rm recursive ~ steps}=m= |T_{\Omega}|-0\leq |V|=n.$
 
 To prove the second item, we consider the number of operations inside the recursive step for the instance  $\Omega^k=\{G, \pi^k, \omega\}$. 
 Note that $k\leq n$ and $|T_{\Omega^k}|=|T_{\Omega}|-k\leq n-k$. 
 If $|T_{\Omega^k}|=0$, then the algorithm will simply call \OPT once. 
 If $|T_{\Omega^k}|>0$, then inside the recursive step, the algorithm will call \DEC once, and call \OPT once or $|T_{\Omega^k}|$ many times depending on the answer of \DEC.
 Moreover, in the later case, the algorithm will also perform $|T_{\Omega^k}|$ many comparisons. 
 Thus, 
 \begin{equation*}
 \begin{aligned}
     \#{\rm calls ~ of ~ \DEC}&=\sum_{|T_{\Omega^k}|>0}1=\sum_{i=1}^{|T_{\Omega}|}1=|T_\Omega|\leq n.\\
     \#{\rm calls ~ of ~ \OPT}&\leq 1+\sum_{|T_{\Omega^k}|>0}|T_{\Omega^k}|=1+\sum_{i=1}^{|T_{\Omega}|}i\leq \frac{n(n+1)}{2}+1.\\
     \#{\rm comparisons}&\leq \sum_{|T_{\Omega^k}|>0}|T_{\Omega^k}|\leq \frac{n(n+1)}{2}
 \end{aligned}
 \end{equation*}
 
 Let $t_{\tt Main}(n)$  denote the running time of Algorithm~\ref{alg-main} on an instance of size $n$, and $t_{\tt Dec}(n)$ and $t_{\tt Opt}(n)$ denote the running time of algorithms for the functions \DEC and \OPT, respectively. 
 Then, $t_{\tt Dec}(n)=O(n^4)$ by the algorithm in~\cite{cornuejols1988general} and $t_{\tt Opt}(n)=O(n^4)$ by the algorithm in~\cite{Dudycz2018}.
Thus, 
$t_{\tt Main}(n)\leq nt_{\tt Dec}(n)+\frac{n(n+1)+2}{2}t_{\tt Opt}(n)+\frac{n(n+1)}{2}=O(n^6)$.
\qedhere
\end{proof}

\section{Proof of Theorem~\ref{Thm:local-global-optimal}}\label{sec:thm}

In this section, we give a proof  of Theorem~\ref{Thm:local-global-optimal}. 
The general strategy is that starting with a non-optimal factor $F$ of an instance $\Omega=(G, \omega, \pi)$, 
we want to find a subgraph $H$ of $G$ such that by taking the symmetric difference $F\Delta H$, 
we get another factor of $\Omega$ with larger weight. 
The existence of such subgraphs is trivial (Lemma~\ref{lem-trivial}).
However, the challenge is how to find one efficiently. 
As an analogy of  augmenting paths in the weighted matching problem, 
we introduce basic augmenting subgraphs (Definition~\ref{def-basic-argumenting}) for  the weighted graph factor problem, which can be found efficiently. 
We will show that given a non-optimal factor $F$,
a basic augmenting subgraph always exists (Lemma~\ref{lem:basic-exist}, property 1).
Then, we can efficiently improve the factor $F$ to another factor with larger weight.
As shown in~\cite{Dudycz2018}, this already gave a weakly-polynomial time algorithm. 
However, the existence of  basic augmenting subgraphs is not enough to get a strongly polynomial-time algorithm, which requires the number of improvement steps being independent of edge weights. 
Thus, in order to prove
Theorem~\ref{Thm:local-global-optimal}, which leads to a strongly polynomial-time algorithm, 
we further establish that there exists a  basic augmenting subgraph 
that satisfies certain stronger properties under suitable assumptions (Lemma~\ref{lem:basic-exist}, property 2).
This result will imply Theorem~\ref{Thm:local-global-optimal}.

\begin{definition}[$F$-augmenting subgraphs]
Suppose that $F$ is a factor of an instance $\Omega=(G, \pi, \omega)$. 
A subgraph $H$ of $G$ 
is $F$-augmenting if $F\Delta H\in \Omega$ and $\omega(F\Delta H)-\omega (F)>0$.
\end{definition}

\begin{lemma}\label{lem-trivial}
Suppose that  $F$ is a factor of an instance $\Omega$. 
If $F$ is not optimal in $\Omega$, then there exists an $F$-augmenting subgraph. 
\end{lemma}
\begin{proof}
Since $F$ is not optimal, there is some $F'\in \Omega$ such that $\omega(F')>\omega(F)$.
Let $H=F\Delta F'$.
We have $F\Delta H=F'\in \Omega$ and $\omega(H)=\omega(F')-\omega(F)>0$.
Thus, $H$ is $F$-augmenting.
\end{proof}

Recall that for an instance $\Omega=(G, \pi, \omega)$ of $\wgfp(\mathscr{G}\cup\mathscr{T})$,
$T_\Omega$ is the set $\{v\in V_G\mid \pi(v)\in \mathscr{T}\}$.
 For two factors $F, F^\ast\in \Omega$,
we define $T_\Omega^{F\Delta F^\ast}=\{v\in T_\Omega\mid \deg_{F\Delta F^\ast}(v)\equiv 1 \mod 2\}=\{v\in T_\Omega\mid \deg_F(v)\not\equiv \deg_{F^\ast}(v) \mod 2\}$.

\begin{definition}[Basic augmenting subgraphs]\label{def-basic-argumenting}
Suppose that $\Omega=(G, \pi, \omega)$ is an instance of $\wgfp(\mathscr G, \mathscr T)$, and $F$ and $F^\ast$  are  factors of $\Omega$ with $\omega(F)<\omega(F^\ast)$.
An $F$-augmenting subgraph $H=(V_H, E_H)$ is \emph{$(F, F^\ast)$-basic} if $H\subseteq F\Delta F^\ast$, $|V^{\rm odd}_H|\leq 2$, and $V^{\rm odd}_H\cap T_\Omega\subseteq T_\Omega^{F\Delta F^\ast}$ where $V^{\rm odd}_H=\{v\in V_H\mid \deg_H(v)\equiv 1 \mod 2\}$.
\end{definition}

\begin{lemma}\label{lem:basic-exist}
\label{claim-local-optimal}
Suppose that  $\Omega=(G, \pi, \omega)$ is an instance of $\wgfp(\mathscr{G}\cup\mathscr{T})$, and $F$ and $F^\ast$  are two factors of $\Omega$.

\begin{enumerate}
    \item If  $\omega(F^\ast)>\omega(F)$, then there exists an  $(F, F^\ast)$-basic subgraph.
    \item If  $\omega(F^\ast)>\opt(\Omega_W^F)$ for every
    $W\subseteq T_\Omega^{F\Delta F^\ast}$ with $|W|\leq 2$, and $T_\Omega^{F\Delta F^\ast}$ contains a vertex $u$ such that $F\in\Omega_u^0$ (i.e., $\deg_F(u)\in D_u^0$),
    then there exists an $(F, F^\ast)$-basic subgraph $H$ where $\deg_H(u)\equiv 0 \mod 2$.
\end{enumerate}
\end{lemma}
\begin{remark}
The first property  of Lemma~\ref{lem:basic-exist}
implies the following: a factor $F\in\Omega$ is optimal if and only if  $\omega(F)\geq\opt(\Omega_W^F)$ for every $W\subseteq T_\Omega$ with $|W|\leq 2$.
This is a special case of the main result (Theorem 2) of~\cite{Dudycz2018} where the authors consider the  $\wgfp$ for all constraints with gaps of length at most 1. 
The second property  of Lemma~\ref{lem:basic-exist} is more refined than the first property  and it implies our main result (Theorem~\ref{Thm:local-global-optimal}).
In this paper, as a by-product of the proof of property 2,
we give a simple proof of Theorem 2 of~\cite{Dudycz2018} for the special case $\wgfp(\mathscr{G}\cup\mathscr{T})$ based on certain properties of cubic graphs. 
\end{remark}

Using the second property of Lemma~\ref{lem:basic-exist}, we can prove Theorem~\ref{Thm:local-global-optimal}.

\begin{theorem*}[Theorem~\ref{Thm:local-global-optimal}]
Suppose that  $\Omega=(G, \pi, \omega)$ is an instance of $\wgfp(\mathscr{G}\cup\mathscr{T})$, $F$ is a factor of $\Omega$ and  $F$ is optimal in $\Omega_u^0$ for some $u\in T_\Omega$.
Then a factor $F'$ is optimal in $\Omega$ if and only if  $\omega(F')\geq
  \omega(F)$ and  $\omega(F')\geq \opt(\Omega^F_W)$ for every $W$ where $u\in
  W\subseteq T_\Omega$ and $|W|=1$ or $2$.
\end{theorem*}

\begin{proof}
If $F'$ is optimal in $\Omega$, then clearly $\omega(F')\geq \omega(F)$ and
  $\omega(F')\geq \opt(\Omega^F_W)$ for every $W$ where $u\in W\subseteq
  T_\Omega$ and $|W|=1$ or $2$.
Thus, to prove the theorem, it  suffices to prove the other direction.
 Since $\omega(F')\geq\omega(F)$ and $F$ is optimal in $\Omega_u^0$, we have $\omega(F')\geq\opt(\Omega^F_W)$ for every  $W\subseteq T_\Omega$  where $u\notin W$ and $|W|\leq 2$.
 Also, since $\omega(F')\geq \opt(\Omega^F_W)$ for every $W$ where $u\in W
  \subseteq T_\Omega$ and $|W|=1$ or $2$,
  we have $\omega(F')\geq\opt(\Omega^F_W)$ for every  $W\subseteq T_\Omega$ where $|W|\leq 2$.
  
 For a contradiction, suppose that $F'$ is not optimal in $\Omega$.
Let $F^\ast$ be an optimal factor of $\Omega$. 
Then, $\omega(F^\ast)>\omega(F')$.
Thus, $\omega(F^\ast)>\omega(F')\geq\opt(\Omega^F_W)$ for every  $W\subseteq T_\Omega$ where $|W|\leq 2$.
Also, $\omega(F^\ast)\notin \Omega_u^0$ since $\omega(F^\ast)>\omega(F)$ and $F$ is optimal in $\Omega_u^0$.
Thus, $\deg_{F^\ast}(u)\not\equiv \deg_F(u) \mod 2$. 
Then, $T_\Omega^{F\Delta F^\ast}$ contains the vertex $u$ such that $F\in \Omega_u^0$.
By Lemma~\ref{claim-local-optimal}, there exists an $(F, F^\ast)$-basic subgraph $H$ where $\deg_H(u)\equiv 0 \mod 2$.
Let $F''=F\Delta H$.
Then $F''\in \Omega$ and $\omega(F'')>\omega(F)$.
Also, $F''\in \Omega_u^0$ since $\deg_{F''}(u)\equiv \deg_F(u) \mod 2$.
This is a contradiction with $F$ being optimal in $\Omega_u^0$.
\end{proof}

Now it suffices to prove Lemma~\ref{lem:basic-exist}.
By a type of normalization  maneuver, we can transfer any instance of  $\wgfp(\mathscr G, \mathscr T)$ to
an instance of  $\wgfp(\mathscr G, \mathscr T)$ defined on subcubic graphs, called a key instance (Definition~\ref{def:key-ins}).
Recall that a subcubic graph is a graph where every vertex has degree $1$, $2$ or $3$.
For key instances, there are five possible forms of basic augmenting subgraphs, called basic factors (Definition~\ref{def-basic-factor}). 
Then,
the crux of the proof of Lemma~\ref{lem:basic-exist}
is to establish 
the existence of certain basic factors of key instances (Theorem~\ref{lem-basic-factor-star}).

\begin{definition}[Key instance]\label{def:key-ins}
A \emph{key instance} $\Omega=(G, \pi, \omega)$ is an instance of $\wgfp(\mathscr G, \mathscr T)$ where $G$ is a subcubic graph, and for every $v\in V_G$, $\pi(v)=\{0, 1\}$  if $\deg_G(v)=1$, $\pi(v)=\{0, 2\}$ if $\deg_G(v)=2$, and $\pi(v)=\{0, 1, 3\}$ (i.e., type-1) or $\{0, 2, 3\}$ (i.e., type-2) if $\deg_G(v)=3$.
We say a vertex $v\in V_G$ of degree $3$ is of type-1 or type-2 if $\pi(v)$ is type-1 or type-2 respectively.
We say a vertex $v\in V_G$ of any degree is 1-feasible or  2-feasible if $1\in \pi(v)$ or $2\in \pi(v)$ respectively.
\end{definition}

\begin{definition}[Basic factor]\label{def-basic-factor}
Let $\Omega$ be a key instance.
A factor  of $\Omega$ is a \emph{basic factor} if it is in one of the following five forms.
\begin{enumerate}
     \item A path, i.e., a tree with two vertices of degree 1 (called endpoints) and all other vertices, if there exists any, of degree 2.
     \item A cycle, i.e., a graph consisting of two vertex disjoint paths with the same two endpoints.
    \item 
    A tadpole graph, i.e., a graph consisting of a cycle and a path such that they intersect at one endpoint of the path.
    \item 
    A dumbbell graph, i.e.,  a graph consisting of  two vertex disjoint cycles and a path such that the path intersects with each cycle at one of its endpoints. 
    \item A theta graph (i.e., a graph consisting of three vertex disjoint paths with the same two endpoints) where one vertex of degree 3 is of type-1, and the other vertex of  degree 3 is of type-2.
\end{enumerate}
\end{definition}

\begin{theorem}\label{lem-basic-factor-star}
Suppose that $\Omega=(G, \pi, \omega)$ is a key instance.
\begin{enumerate}
    \item If $\omega(G)>0$,
then there is a basic factor $F$ of $\Omega$ such that $\omega(F)>0$.
\item If $\omega(G)>0$,  $\omega(G)>\omega(F)$ for every basic factor $F$ of $\Omega$, and
$G$ contains a vertex $u$ with $\deg_G(u)=1$ or $\deg_G(u)=3$ and  $\pi(u)=\{0, 2, 3\}$,
then
there is a basic factor $F^\ast$ of $\Omega$ such that $\omega(F^\ast)>0$ and $\deg_{F^\ast}(u)\equiv 0 \mod 2$. (Recall that  $\deg_{F^\ast}(u)=0$ if $u\notin V_{F^\ast}$).
\end{enumerate}
\end{theorem}

\begin{remark}\label{remark:example}
For the second property of Theorem~\ref{lem-basic-factor-star},
the requirement of $\pi(u)=\{0, 2, 3\}$ when $\deg_G(u)=3$ 
is crucial. 
Consider the instance $\Omega=(G, \pi, \omega)$ as shown in Figure~\ref{fig:1.cc}.
It is easy to that $\Omega$ is a key instance. 
In this case, it can be checked that $\omega(G)=6>0$ and  $\omega(G)>\omega(F)$ for every basic factor $F$ of $\Omega$.
However,
there is no basic factor $F^\ast$ of $\Omega$ such that $\omega(F^\ast)>0$ and $\deg_{F^\ast}(u)\equiv 0 \bmod 2$.
Thus, 
the second property does \emph{not} hold for a vertex $u$ where $\deg_G(u)=3$ and $\pi(u)=\{0, 1, 3\}$. 
\end{remark}

We will now describe the normalization maneuver, and use Theorem \ref{lem-basic-factor-star} to prove Lemma~\ref{lem:basic-exist}.

\begin{proof}[Proof of Lemma~\ref{lem:basic-exist}]
 Recall that $F$ and $F^\ast$ are two factors of the instance $\Omega=(G, \pi, \omega)$ (not necessarily a key instance).
Consider the subgraph $G_\Delta=F\Delta F^\ast$ of $G$. 
 $G_\Delta=(V_{G_\Delta}, E_{G_\Delta})$ is not necessarily a subcubic graph. 
In order to invoke Theorem~\ref{lem-basic-factor-star},
we modify $G_\Delta$ to a subcubic graph $G^s$, and  construct a key instance $\Omega^s=(G^s, \pi^s, \omega^s)$ on it.

For every $v\in V_{G_\Delta}$, 
we consider  the set of edges incident to $v$ in $G_\Delta$, denoted by $E_v$.
Since $G_\Delta=F\Delta F^\ast$, 
we have $E_v\subseteq E_{G_{\Delta}}=E_F\Delta E_{F^\ast}$, where $E_F$ and $E_{F^\ast}$ are the edge sets of the factors $F$ and $F^\ast$ respectively.
If there is  a pair of edges $e, e^\ast \in E_v$ such that 
$e\in E_F$ and $e^\ast\in E_{F^\ast}$,
then we perform the following \emph{separation} operation for this pair of edges. 
Suppose that $e=(v, u)$ and $e^\ast=(v, u^\ast)$; 
we add a new vertex $v^1$ to the graph, and replace the edges $e$ and $e^\ast$ by $(v^1, u)$ and $(v^1, u^\ast)$ respectively. 
We label the vertex $v^1$ (of degree 2) by $\pi^s(v^1)=\{0, 2\}$.
With a slight abuse of notation, we may still use $e$ and $e^\ast$ to denote these two new edges, and also use $E_{G_\Delta}$ to denote the set of all edges of the new graph. 

For each $E_v$, keep doing the separation operations for pairs of edges of which one is in $E_F$ and the other is in $E_{F^\ast}$ until all the remaining edges in $E_v$ are in $E_F$ or in $E_{F^\ast}$ 
We use $E_v^{\tt r}$ to denote the set  of  remaining edges. It is possible that $E_v^{\tt r}$ is empty. 
Let $P_v^1, \ldots, P_v^k$ be the pairs of edges that have been separated, 
and $v^1, \ldots, v^k$ be the added vertices ($k$ can be zero).
Note that all these new vertices are of degree $2$, and are labeled by $\{0, 2\}$.
Now, we have the partition $E_v=P_v^1\cup \cdots\cup  P_v^k\cup E_v^{\tt r}$. 
Let $r=|E_v^{\tt r}|$.
Then $r=|\deg_{F}(v)-\deg_{F^\ast}(v)|$.
Note that $r$ is even if $\pi(v) \in \mathscr{G}_2$, and 
$r\leq 3$ if $\pi(v) \in \mathscr{T}$.
We deal with edges in $E_v^{\tt r}$ according to $r$ and $\pi(v)$. 
\begin{itemize}
    \item If $r=0$, then $v$ is an isolated vertex in the current graph, and we simply remove it. 
    Consider an arbitrary subgraph  $H$   of the original $G_\Delta$ induced by a union of some pairs of edges in $P_v^1, \ldots, P_v^k$. 
    Then, for the subgraph $F\Delta H$ of $G_\Delta$, we have
    $$\deg_{F\Delta H}(v)= \deg_F(v)\in \pi(v).$$ 
   \item If $r\neq 0$ and $\pi(v)\in \mathscr{G}_1$, then we replace the vertex $v$ with $r$ many new vertices, and replace the $r$ many edges incident to $v$ by $r$ many edges incident to these new vertices such that each vertex has degree $1$.
    We label every new vertex by $\{0, 1\}$.
    Suppose that $L=\min\{\deg_F(v), \deg_{F^\ast}(v)\}$ and $U=\max \{\deg_F(v), \deg_{F^\ast}(v)\}$.
    Since $\pi(v)\in \mathscr{G}_1$, $\{L, L+1, \ldots, U\}\subseteq \pi(v)$. 
    Consider  an arbitrary subgraph $H\subseteq G_\Delta$ induced by a union of some pairs of edges in $P_v^1, \ldots, P_v^k$ and a subset  of $E_v^{\tt r}$.
     Then, for the subgraph $F\Delta H$ of $G_\Delta$, we have $$\deg_{F\Delta H}(v)\in \{L, L+1, \ldots, U\}\in \pi(v).$$ 
       \item If $r\neq 0$ and $\pi(v)\in \mathscr{G}_2\backslash \mathscr{G}_1$, then 
       we replace the vertex $v$ with ${r}/{2}$ many vertices, 
       and replace the $r$ many edges incident to $v$ by $r$ many edges incident to these new vertices such that each vertex has degree $2$.
       (We can partition these $r$ many edges into arbitrary pairs.)
    We label every new vertex by $\{0, 2\}$.
    Suppose that $L=\min\{\deg_F(v), \deg_{F^\ast}(v)\}$ and $U=\max \{\deg_F(v), \deg_{F^\ast}(v)\}$.
    Since $\pi(v)\in \mathscr{G}_2$, $\{L, L+2, \ldots, U\}\subseteq \pi(v)$. 
    Consider an arbitrary subgraph $H\subseteq G_\Delta$ induced by a union of some pairs of edges in $P_v^1, \ldots, P_v^k$ and an even-size subset of $E_v^{\tt r}$.
     Then, for the subgraph $F\Delta H$ of $G_\Delta$, we have $$\deg_{F\Delta H}(v)\in \{L, L+2, \ldots, U\}\in \pi(v).$$ 
    \item If $r\neq 0$ and  $\pi(v)\in \mathscr{T}$, then there are three subcases. 
  If $r=1$, then $v$ has degree $1$ in the current graph. We label it by $\pi^s(v)=\{0, 1\}$.
  If $r=2$, then $v$ has degree $2$ in the current graph. We label it by $\pi^s(v)=\{0, 2\}$.
    If $r=3$, then $v$ has degree $3$ in the current graph. We label it by $\pi^s(v)=\{0, 1, 3\}$ if 
    $\deg_F(v)\in D_v^1$,
    and $\pi^s(v)=\{0, 2, 3\}$ if $\deg_F(v)\in D_v^0$.
     Consider an arbitrary subgraph $H\subseteq G_\Delta$ induced by a union of some pairs of edges in $P_v^1, \ldots, P_v^k$ and a subset $I$ of $E_v^{\tt r}$ where $|I|\subseteq \pi^s(v)$. 
      Then, for the subgraph $F\Delta H$ of $G_\Delta$, we have $$\deg_{F\Delta H}(v)\in \pi(v).$$ 
     
\end{itemize}

Now, we get a subcubic graph $G^s=(V_{G^s}, E_{G^s})$ from $G_\Delta$.
Each vertex $v$ in $G_\Delta$ is replaced by a set of new vertices in $G^s$, denoted by $S(v)$. 
\begin{itemize}
    \item 
If $\pi(v)\in \mathscr{G}_1$, then  $S(v)$ consists of vertices of degree 2 or 1. 
\item If $\pi(v)\in \mathscr{G}_2$, then  $S(v)$ consists of vertices of degree 2.
\item If $\pi(v)\in \mathscr{T}$, then $S(v)$ consists of vertices of degree 2 and possibly a vertex of degree $r$ where $r=|\deg_{F}(v)-\deg_{F^\ast}(v)|\leq 3$ (there is no such a vertex if $r=0$).
In particular, if $\deg_{F}(v)-\deg_{F^\ast}(v)\equiv 0 \mod 2$, then   $S(v)$ consists of vertices of degree 2. 
\end{itemize}
In all cases, we have $\deg_{G_\Delta}(v)=\sum_{x\in S(v)}\deg_{G^s}(x)$.
Each edge $(u, v)$ in $G_\Delta$ is replaced by an edge $(u^s, v^s) \in G_\Delta$ where $u^s\in S(u)$ and $v^s\in S(v)$.
Once we get $G^s$ from $G_\Delta$, it is clear that there is a natural one-to-one correspondence between edges in $G^s$ and edges in $G_\Delta$. 
Without causing ambiguity, when we say an edge or an edge set in $G^s$, we may also refer it to the corresponding edge or edge set in $G_\Delta$.

As we constructed $G^s$, we have already defined the mapping $\pi^s$ which labels each vertex in $G^s$ with a degree constraint. 
For $x\in V_{G^s}$, we have $\pi^s(x)=\{0, 1\}$ if $\deg_{G^s}(x)=1$, $\pi^s(x)=\{0, 2\}$ if $\deg_{G^s}(x)=2$, and 
$\pi^s(x)=\{0, 1, 3\}$ or $\{0, 2, 3\}$ if $\deg_{G^s}(x)=3$.
Moreover, as we have discussed above,
for a vertex $v\in V_{G_\Delta}$ and  a subgraph $H\subseteq G_\Delta$ induced be a set $E$ of edges incident to $v$ in $G_\Delta$, we have $\deg_{F\Delta H}(v)\in \pi(v)$ if $\deg_{H^s}(x)\in \pi^s(x)$ for every $x\in S(v)$ where 
$H^s$ is the subgraph of $G^s$ induced by the edge set $E$ (viewed as edges in $G^s$).

Now, we define the function $\omega^s$ for edges in $G^s$ as follow.
Recall that for every edge in $G^s$, its corresponding edge in $G_\Delta$ is either in the factor $F$ or the factor $F^\ast$ but not in both since $G_{\Delta}=F\Delta F^\ast$. 
For $e\in E_{G^s}$, we define $\omega^s(e)=\omega({e})$ if ${e}\in E_{F^\ast}$ and $\omega^s(e)=-\omega({e})$ if ${e}\in E_F$.
We can extend $\omega^s$ to any subgraph of $G^s$ by defining its weight to be the total weight of all its edges. 
Then, for any subgraph $H^s\subseteq G^s$, 
$\omega^s(H^s)=\omega(F\Delta H)-\omega(F)$ where $H$ is the subgraph of $G_\Delta$ corresponding to $H^s$.
In particular, $\omega^s(G^s)=\omega(F^\ast)-\omega(F)>0$.
Thus,  we get a key instance $\Omega^s=(G^s, \pi^s, \omega^s)$ where $\omega^s(G^s)>0$.

Suppose that $F^s$ is a factor of $G^s$ with $\omega^s(F^s)>0$.
We consider the subgraph $H$ of $G_\Delta$ induced by the edge set ${E_{F^s}}$ (viewed as edges in $G_\Delta$).
We show that $H$ is an $(F, F^\ast)$-basic subgraph of $G$.
We have $H\subseteq G_\Delta=F\Delta F^\ast$.
As we have discussed above, for every vertex $v\in V_{F\Delta H}$, $\deg_{F\Delta H}(v)\in \pi(v)$. 
Thus, $F\Delta H\in \Omega$.
Also, $\omega^s(F^s)=\omega(F\Delta H)-\omega(F)>0.$
Then, $H$ is an $F$-augmenting subgraph.
For every $v\in V_H$, $\deg_{H}(v)=\sum_{x\in S(v)}\deg_{F^s}(x)$.
Then, $\deg_H(v)$ is odd only if there is a vertex $x\in S(v)$ such that $\deg_{F^s}(x)$ is odd.
Thus, the number of odd vertices in $H$ is no more than the number of odd vertices in $F^s$.
Since $F^s$ is a basic factor, it has at most $2$ vertices of odd degree. 
Thus, $H$ has at most  $2$ vertices of odd degree.
Moreover, for a vertex $v\in V_H\cap T_\Omega$, if $\deg_F(v)\equiv \deg_{F^\ast}(v) \mod 2$, then $S(v)$ consists of vertices of degree 2.
Thus, $\deg_{F^s}(x)\in \{0, 2\}$ for every $x\in S(v)$. 
Then, $\deg_H(v)=\sum_{x\in S(v)}\deg_{F^s}(x)$ is even.
Thus, for a vertex $v\in V_H\cap T_\Omega$, $\deg_H(v)$ is odd only if $\deg_F(v)\not\equiv \deg_{F^\ast}(v) \mod 2.$
Then, $V_H^{\rm odd}\cap T_{\Omega}\subseteq T^{F\Delta F^\ast}_{\Omega}$ where $V^{\rm odd}_H=\{v\in V_H\mid \deg_H(v)\equiv 1 \mod 2\}$.
Thus, $H$ is an $(F, F^\ast)$-basic subgraph of $G$.

By the first part of Theorem~\ref{lem-basic-factor-star}, there exists a basic factor $F^s\in \Omega^s$ with $\omega^s(F^s)>0$.
Thus,  there exists an $(F, F^\ast)$-basic subgraph $H\subseteq G$ induced by the edge set $E_{F^s}$. The first part is done.

Now, we prove the second part. 
Suppose that $\omega(F^\ast)>\opt(\Omega_W^F)$ for every
    $W\subseteq T_\Omega^{F\Delta F^\ast}$ where $|W|\leq 2$, and $T_\Omega^{F\Delta F^\ast}$ contains a vertex $u$ where $\deg_F(u)\in D^0_u$.
    Consider the instance $\Omega^s$.
    First, we prove that $\omega^s(G^s)>\omega(F^s)$ for every basic factor $F^s$ of $\Omega^s$.
    For a contradiction, suppose that there is some $F^s\in \Omega^s$ such that $\omega^s(G^s)\leq \omega (F^s)$.
    Still consider the subgraph $H$ of $G_\Delta$ inducted by $E_{F^s}$.    
    We know that $H$ is an $(F, F^\ast)$-basic subgraph of $G$ and $\omega^s(F^s)=\omega(F\Delta H)-\omega(F)$.
    Let $W=V^{\rm odd}_H\cap T_\Omega$.
    Then, $W\subseteq T_\Omega^{F\Delta F^\ast}$ and $|W|\leq 2$.
    For every $x\in W$, since $\deg_H(x)$ is odd,
    we have $\deg_{F\Delta H}(x)\not\equiv \deg_F(x) \mod 2$, and then $\deg_{F\Delta H}(x)\in \pi(x)\backslash D_v^F$.
    For  every $x\in T_\Omega\backslash W$, since $\deg_H(x)$ is even,
    we have $\deg_{F\Delta H}(x)\equiv \deg_F(x) \mod 2$ and then $\deg_{F\Delta H}(x)\in  D_v^F$. 
    Consider the sub-instance $\Omega^F_W=(G, \pi^F_W, \omega)$ of $\Omega$ (see Equation~(\ref{equ:omega-f-w}) for the definition of $\Omega^F_W$).
 Then, $F\Delta H\in \Omega^F_W$.
 Thus, $\omega(F\Delta H)\leq \opt(\Omega^F_W)$.
    Since $$\omega^s(G^s)=\omega(F^\ast)-\omega(F)\leq \omega^s(F^s)=\omega(F\Delta H)-\omega(F),$$ we have $\omega(F^\ast)\leq \omega(F\Delta H)$.
    Then, $\omega(F^\ast)\leq \opt(\Omega^F_W)$.
    A contradiction with the assumption that $\omega(F^\ast)>\opt(\Omega_W^F)$ for every
    $W\subseteq T_\Omega^{F\Delta F^\ast}$ where $|W|\leq 2$.
    Thus, $\omega^s(G^s)>\omega^s(F^s)$ for every basic factor $F^s$ of $\Omega^s$.
    
    Since $T_\Omega^{F\Delta F^\ast}$ contains a vertex $u$ where $\deg_F(u)\in D^0_u$.
    Consider the vertex set $S(u)$ in $G^s$ that corresponds to $u$.
    Since $u\in T_\Omega^{F\Delta F^\ast}$, $\deg_F(u)\not\equiv\deg_{F^\ast}(u) \mod 2$.
    Thus, $S(u)$ consists of vertices of degree $2$ and a vertex $u^s$ of degree $\deg_{G^s}(u^s)=|\deg_F(u)-\deg_{F^\ast}(u)|$ which is $1$ or $3$.
    If $|\deg_F(u)-\deg_{F^\ast}(u)|=3$, 
    then $\pi^s(u^s)=\{0, 2, 3\}$ since $\deg_F(u)\in D^0_u$.
    Thus, $G^s$ contains a vertex $u^s$ where $\deg_{G^s}(u^s)=1$ or $\deg_{G^s}(u^s)=3$ and $\pi^s(u^s)=\{0, 2, 3\}$.
    Then, by the second part of Theorem~\ref{lem-basic-factor-star}, there is a basic factor $F^s\in \Omega^s$ such that $\omega^s(F^s)>0$ and $\deg_{F^s}(u_s)\equiv 0 \mod 2$.
    Again, consider the subgraph $H$ of $G_\Delta$ inducted by $E_{F^s}$.   
    We have proved that $H$ is an $(F, F^\ast)$-basic subgraph of $G$.
    Also, $$\deg_H(u)=\sum_{x\in S(u)\backslash\{u^s\}}\deg_{F^s}(x)+\deg_{F^s}(u^s)\equiv 0 \mod 2$$ since 
    $\deg_{F^s}(x)\in \pi^s(x)=\{0, 2\}$ for every $x\in S(u)\backslash\{u^s\}$, and $\deg_{F^s}(u_s)\equiv 0 \mod 2$.
Thus, there is an  $(F, F^\ast)$-basic subgraph $H$ of $G$ such that $\deg_H(u)\equiv 0 \mod 2$.
\end{proof}

\section{Proof of Theorem~\ref{lem-basic-factor-star}}
\label{sec:thm1}
We first prove the first property (restated in Lemma~\ref{lem-MBFC}), 
and then prove the second property (restated in Lemma~\ref{lem-2}) 
using the first property. 
In this section, for two points $x$ and $y$, we use $p_{xy}$, $p'_{xy}$ or $p''_{xy}$ to denote a path with endpoints $x$ and $y$.
Recall that $V_{p_{xy}}$ and $E_{p_{xy}}$ denotes the vertex set and the edge set of $p_{xy}$ respectively. 

\subsection{Proof of the first property}\label{subsec1}

\begin{lemma}\label{lem:not-connect}
Suppose that $\Omega=(G, \pi, \omega)$ is a key instance with $\omega(G)>0$. 
If $G$ is not connected, 
then there is a factor $F\in \Omega$ such that $\omega(F)>0$ and $|E_F|<|E_G|$.
\end{lemma}

\begin{proof}
Suppose that $G_1$ is  a connected component of $G$, and $G_2=G\Delta G_1$ is the rest of the graph. 
Note that $G_1$ and $G_2$ are both factors of $G$.
By the definition of subcubic graphs, there are no isolated vertices in $G$.
Thus, neither $G_1$ nor $G_2$ is a single vertex. 
Then, $|E_{G_1}|, |E_{G_2}|\geq 1$.
Since $E_G$ is the disjoint union of $E_{G_1}$ and $E_{G_2}$, $|E_{G_1}|, |E_{G_2}| < |E_{G}|$,
and 
$\omega(G)=\omega(G_1)+\omega(G_2)$.
Since $\omega(G)>0$, among $\omega(G_1)$ and $\omega(G_2)$, one is positive. 
Thus, 
we are done. 
\end{proof}

\begin{lemma}\label{lem:path}
Suppose that $\Omega=(G, \pi, \omega)$ is a key instance with $\omega(G)>0$. 
Then, there is a factor $F\in \Omega$ such that $\omega(F)>0$ and $|E_F|<|E_G|$ if one of the following conditions holds:
\begin{enumerate}
    \item There is a path $p_{uv}\subseteq G$ with endpoints $u$ and $v$ where  $u$ and $v$ are the only two vertices in $p_{uv}$ of type-2 (i.e., $\deg_G(u)=\deg_G(v)=3$ and $\pi(u)=\pi(v)=\{0, 2, 3\}$) and $\omega(p_{uv})\leq 0$.
    \item There is a cycle $C\subseteq G$ where no vertex is of type-2 and $\omega(C)\leq 0$.
\end{enumerate}
\end{lemma}

\begin{proof}
Suppose that the first condition holds. 
Consider the subgraph $F=G\backslash p_{uv}$ of $F$.
Then, $|E_F|=|E_G|-|E_{p_{uv}}|<|E_G|$, and $\omega(F)=\omega(G)-\omega(p_{uv})\geq \omega(G)>0$.
Now we only need to show that $F$ is a factor of $\Omega$. 
The vertex set $V_F$ consists of three parts:
$$V_1=V_{G}\backslash V_{p_{uv}}, ~~~ V_2=\{x\in V_{p_{uv}}\backslash\{u, v\}\mid \deg_G(x)=3\}, \text{ ~~~and~~~ } V_3=\{u, v\}.$$
Since $u$ and $v$ are the only two vertices of type-2 in $p_{uv}$, for every $x\in V_2$, $x$ is of type-1 (i.e., $\pi(x)=\{0, 1, 3\}$).
Then, for every $x\in V_F$, we have $\deg_F(x)=\deg_G(x)\in \pi(x)$ if $x\in V_1$, $\deg_F(x)=1\in \pi(x)$ if $x\in V_2$, and $\deg_F(x)=2\in \pi(x)$ if $x\in V_3$. 
Thus, $F$ is a factor of $G$. We are done.

Suppose that the second condition holds.
Consider the subgraph $F=G\backslash C$. Then $|E_F|<|E_G|$ and $\omega(F)>0$.
Similar to the above proof, one can check that $F$ is a factor $\Omega$. We are done. 
\end{proof}

\begin{lemma}\label{lem:cycle}
Suppose that $\Omega=(G, \pi, \omega)$ is a key instance with $\omega(G)>0$, $G$ is not a basic factor of $\Omega$ and $C\subseteq G$ is a cycle.
Let $k$ be the number of type-1 vertices  and $\ell$ be the number of type-2 vertices in $C$.
If $k\neq 1$ and $\ell\neq 1$, then there is a factor $F\in \Omega$ such that $\omega(F)>0$ and $|E_F|<|E_G|$.
\end{lemma}

\begin{proof}
We prove this lemma in two cases depending on whether $\omega(C)>0$ or $\omega(C)\leq 0$. 

We first consider the case that $\omega(C)>0$.
If $k=0$, then all vertices in $C$ are 2-feasible (see Definition~\ref{def:key-ins}). 
Thus, $C$ is a factor of $\Omega$.
Since $G$ is not a basic factor of $\Omega$, we have $G\neq C$. 
Also, since $G$ has no isolated vertices, $C\subsetneq G$ implies that $|E_C|<|E_G|$. We are done.
Thus, we may assume that $k\geq 2$.
Suppose that $\{u_1, u_2, \ldots, u_k\}$ are the type-1 vertices in $C$.
We list them
in the order of traversing the cycle starting from $u_1$ in an arbitrary direction.
Then, these $k$ many vertices split the cycle into $k$ many paths $p_{u_1u_2}, \ldots, p_{u_ku_{k+1}}$ ($u_{k+1}=u_1$).
For each path, all its vertices are 2-feasible except for its two endpoints which are $1$-feasible. 
Thus, each path is a basic factor of $G$.
We have $|E_{p_{u_iu_{i+1}}}|<|E_G|$ for every $i \in [k]$.
Since $$\omega(C)=\sum^k_{i=1}\omega(p_{u_iu_{i+1}})>0,$$  there is a path $p_{u_iu_{i+1}}$ such that $\omega(p_{u_iu_{i+1}})>0$.
Thus, we are done.

Then we consider the case that $\omega(C)\leq 0$.
If $\ell=0$, then $C\subseteq G$ is cycle with no type-2 vertices.
By Lemma~\ref{lem:path}, we are done.
Thus, we may assume that $\ell\geq 2$.
Suppose that $\{v_1, v_2, \ldots, v_\ell\}$ are the type-2 vertices in $C$.
We list them
in the order of traversing the cycle starting from $v_1$ in an arbitrary direction.
Then, these $\ell$ many vertices split the cycle into $\ell$ many paths $p_{v_1v_2}, \ldots, \ldots, p_{v_\ell v_{\ell+1}}$ ($v_{\ell+1}=v_1$).
For each path, it has no vertex of type-2 except for its
two endpoints which are of type-2. 
Since $$\omega(C)=\sum^k_{i=1}\omega(p_{v_iv_{i+1}})\leq 0,$$  there is a path $p_{v_iv_{i+1}}$ such that $\omega(p_{v_iv_{i+1}})\leq 0$.
Thus, there is a path $p_{v_iv_{i+1}}\subseteq G$ where $v_i$ and $v_{i+1}$ are the only two vertices of type-2 in $p_{v_iv_{i+1}}$   and $\omega(p_{v_iv_{i+1}})\leq 0$.
Then, by Lemma~\ref{lem:path}, we are done. 
\end{proof}

\begin{lemma}\label{lem:2-connect}
Suppose that $\Omega=(G, \pi, \omega)$ is a key instance with $\omega(G)>0$, and $G$ is not a basic factor of $\Omega$.
If $G$ is $2$-connected, then there is a factor $F\in \Omega$ such that $\Omega(F)>0$ and $|E_F|<|E_G|$.
\end{lemma}

\begin{proof}
Since $G$ is 2-connected, it contains at least three vertices and it contains no vertex of degree $1$.
Consider the number of type-1 vertices in $G$.
There are three cases. 
\begin{itemize}
    \item $G$ has no type-1 vertex. 
    
Since $G$ is 2-connected, there is a cycle $C\subseteq G$. Clearly,
$C$ has no type-1 vertex.
If $C$ has exactly one type-2 vertex, denoted by $v$,
then $v$ is the only vertex in $C$ such that $\deg_G(v)=3$. 
Then, there is an edge $e\in E_{G}$ incident to $v$ such that $e\notin E_C$.
It is easy to see that $e$ is a bridge of $G$, a contradiction with $G$ being 2-connected.
Thus, $C$ has no type-2 vertex, or it has at least two type-2 vertices.
Then, by Lemma~\ref{lem:cycle}, we are done.

\item $G$ has exactly one type-1 vertex.

Let $u$ be the type-1 vertex of $G$.
Since $G$ is 2-connected, there is a cycle  $C\subseteq G$ containing the vertex $u$. 
Since $\deg_C(u)=2<\deg_G(u)=3$, by Lemma~\ref{lem:extra-path}, 
there is a path $p_{uw}\subseteq G$  with endpoints $u, w\in V_C$ such that $E_{p_{uw}}\cap E_C=\emptyset$.

Consider the subgraph $H=p_{uw}\cup C$ of $G$.
$H$ is a theta graph where $\deg_H(u)=\deg_H(w)=3$.
All vertices of $H$ are 2-feasible except for $u$ which is $1$-feasible.
Note that $H$ is a basic factor of $\Omega$.
Since $G$ is not a basic factor of $\Omega$, $H\neq G$. 
Also since $G$ is connected, 
there exists an edge $e_{ts}=(t, s)$ incident  to a vertex $s\in V_{H}$ such that $e_{ts}\notin E_H$. 
Clearly,  $s$ is a vertex of type-2, $\deg_G(s)=3$ and $\deg_H(s)=2$.
Then, by Lemma~\ref{lem:extra-path},
there is a path  $p_{sr}$ with endpoints $s, r\in V_H$ such that $E_{p_{sr}}\cap E_H=\emptyset$.
Clearly, $\deg_G(r)=3$ and $r$ is a  vertex of type-2.
Since $s, r\in V_H$ and $H$ is a theta graph which is 2-connected,  we can find a path $p'_{sr}\subseteq H$ with endpoints $s$ and $r$ such that the only type-1 vertex $u$ in $H$ is not in  $p'_{sr}$.
Consider the cycle $C'=p_{sr}\cup p'_{sr}$.
It has no vertex of type-1, and it has at least two vertices $s$ and $r$ of type-2.
By Lemma~\ref{lem:cycle}, we are done.

\item $G$ has at least two type-1 vertices. 

Since $G$ is 2-connected and it contains at least two type-1 vertices, we can find a cycle $C\subseteq G$ that contains at least two  type-1 vertices.
Consider the number of type-2 vertices in $C$.
If the number is not 1, then by Lemma~\ref{lem:cycle},
we are done. 
Thus, we may assume that $C$ contains exactly one vertex of type-2, denoted by $v$.
Since $G$ is 2-connected and $\deg_G(v)=3>\deg_C(v)=2$, 
we can find a path $p_{vu}$ for some $u\in V_C$ such that $E_{p_{vu}}\cap E_C=\emptyset$.
We have $\deg_G(u)=3$.
Since $v$ is the only vertex of type-2 in $C$, $u$ is a vertex of type-1.
Vertices $v$ and $u$ split $C$ into two paths $p'_{vu}$ and $p''_{vu}$. 
Since $C$ contains at least two type-1 vertices, 
there exists some $w\in V_C$ where $w\neq u$ such that $w$ is of type-1.
Also, $w\neq v$ since $v$ is of type-2.
Since $w\in V_C=V_{p'_{vu}}\cup V_{p''_{vu}}$ and $V_{p'_{vu}}\cap V_{p''_{vu}}=\{u, v\}$, without loss of generality, we may assume that $w\in V_{p'_{vu}}$.

Consider the path $p_{vu}$.
If $p_{vu}$ contains at least two vertices of type-2, 
then the cycle $C'=p_{vu}\cup p'_{vu}$ contains at least two vertices  of type-2 and at least two vertices $u$ and $w$ of type-1.
Then, by Lemma~\ref{lem:cycle}, we are done.
Thus, we may assume that 
$v$ is the only vertex of type-2 in $p_{vu}$.
Consider the theta graph $H=p_{vu}\cup C$.
Then  $v$ is the only vertex of type-2 in $H$.
Note that $w\in V_H$, $\deg_H(w)=2<\deg_G(w)=3$.
Since $G$ is 2-connected, by Lemma~\ref{lem:extra-path}, we can  find a path $p_{ws}$ for some $s\in V_H$ such that $E_{p_{ws}}\cap E_H=\emptyset$.
Clearly $s\neq v$. Then, $s$ is of type-1 since $v$ is the only vertex of type-2 in $H$.

Consider the number of type-2 vertices in $p_{ws}$.
 Suppose that there is no  vertex of type-2 in $p_{ws}$.
    Since $H$ is  2-connected and $H$ contains only one vertex $v$ of type 2,
    we can find a path $p'_{ws}
    \subseteq H$ such that  $p'_{ws}$ does not contain the vertex $v$ of type-2.
    Then, 
    the cycle $p_{ws}\cup p'_{ws}$ has no vertex of type-2 and at least two vertices $w$ and $s$ of type-1.
    By Lemma~\ref{lem:cycle}, we are done. 
    Otherwise, there is at least one vertex of type-2 in $p_{ws}$.
    Since $H$ is 2-connected, 
        we can find a path $p''_{ws}
    \subseteq H$ such that $p''_{ws}$ contains the vertex $v$ of type-2.
    Then, 
    the cycle $p_{ws}\cup p''_{ws}$ has at least two vertices of type-2 and at least two vertices $w$ and $s$ of type-1.
    By Lemma~\ref{lem:cycle}, we are done.  \qedhere
    \end{itemize}
    \end{proof}

\begin{definition}[Induced sub-instance]
    For a key instance $\Omega=(G, \pi, \omega)$, and a factor $F\in \Omega$, the sub-instance of $\Omega$ induced by $F$, denoted by $\Omega_F$, is a key instance $(F, \pi_F, \omega_F)$  defined on the subgraph $F$ of $G$
where $\pi_F(x)=\pi(x)\cap [\deg_F(x)]\subseteq \pi(x)$ for every $x\in V_F$ and $\omega_F$ is the restriction of $\omega$ on $E_F$ (we may write $\omega_F$ as $\omega$ for simplicity).
\end{definition}

We are now ready to prove the first property of Theorem~\ref{lem-basic-factor-star} as restated in the next lemma.

\begin{lemma}\label{lem-MBFC}
Suppose that $\Omega=(G, \pi, \omega)$ is a key instance.
If $\omega(G)>0$,
then there is a basic factor $F$ of $\Omega$ such that $\omega(F)>0$.
\end{lemma}

\begin{proof}
We prove this lemma by induction on the number of edges in $G$.

If $|E_G|=1$, then $G$ is a single edge.
Thus, $G$ is a basic factor of $\Omega$, and $\omega(G)>0$.
We are done. 

We assume that the lemma holds for all key instances where the underlying graph has no more than $n$ many edges. 
We consider a key instance $\Omega=(G, \pi, \omega)$ where $|E_G|=n+1$.

If $G$ is a basic factor of $\Omega$, then clearly we are done.
Thus, we may assume that $G$ is not a basic factor of $\Omega$. 
Suppose that we can find a factor $F\in \Omega$ such that $\omega(F)>0$ and $|E_F|<|E_G|=n+1$.
Then, consider the  sub-instance $\Omega_F$
of $\Omega$ induced by $F$.
Since $|E_F|<n+1$ and $\omega(F)>0$, 
by the induction hypothesis, 
there is basic factor $F'\in \Omega_F$ such that $\omega(F')>0$.
Since $\Omega_F\subseteq \Omega$, $F'\in \Omega$.
Then, we are done.
Thus, in order to establish the inductive step, it suffices to prove that 
there is a factor $F\in \Omega$ such that $|E_F|<|E_G|$ and $\omega(F)>0$.

By Lemmas~\ref{lem:not-connect}
and~\ref{lem:2-connect}, if $G$ is not connected or $G$ is 2-connected, then we are done.
Thus, we may assume that $G$ is a connected graph but not 2-connected.
By Lemma~\ref{lem:bridge}, $G$ contains at least a bridge.
Fix such a bridge of $G$. 
Let $p_{uv}$ be the path containing the bridge such that for every vertex $x\in V_{p_{uv}}
\backslash\{u, v\}$, $\deg_G(x)=2$ and $\deg_G(u), \deg_G(v)\neq 2$;
observe that such a path exists and it is unique.
In fact, the whole path can be viewed as a ``long bridge'' of the graph $G$. 
Then, $G\backslash p_{uv}$ is not connected and it has two connected components. 
Let $G_u\subseteq G\backslash p_{uv}$ be the part that contains $u$ and $G_v\subseteq G\backslash p_{uv}$  be the part that contains $v$. 

If both $G_u$ and $G_v$ are single vertices, then the graph $G$ is a path.
If both $G_u$ and $G_v$ are cycles, then $G$ is a dumbbell graph.
If one of $G_u$ and $G_v$ is a single vertex and the other one is a cycle, then $G$ is a tadpole graph.
In all these cases, $G$ is a basic factor of $\Omega$. A contradiction with our assumption. 
Thus,  among $G_u$ and $G_v$, at least one is neither a cycle nor a single vertex. 
Without loss of generality, we may assume that $G_u$ is neither a cycle nor a single vertex.

Since $G_u$ is not a single vertex, $\deg_G(u)\neq 1$.
By assumption, $\deg_G(u)\neq 2$.
Then $\deg_G(u)=3$, and hence $\deg_{G_u}(u)=2$.
Let $e_1=(u, w_1)$ and $e_2=(u, w_2)$ be the two edges incident to $u$ in $G_u$.
We slightly modify $G_u$ to get a new graph.
We replace the vertex $u$ in $G_u$ by two vertices $u_1$ and $u_2$, and replace the edges $(u, w_1)$ and $(u, w_1)$ in $G_u$ by two new edges $(u_1, w_1)$ and $(u_2, w_2)$ respectively.
We denote the new graph by $G'$.
With a slight abuse of notations, we still use $e_1$ and $e_2$ to denote the edges $(u_1, w_1)$ and $(u_2, w_2)$ in $G'$ respectively, and we say $E_{G_u}=E_{G'}$.
Then, the edge weight function $\omega$ can be adapted to $E_{G'}$.
We define the following instance $\Omega'=(G', \pi', \omega')$ where $\pi'(u_1)=\pi'(u_2)=\{0, 1\}$ and $\pi'(x)=\pi(x)$ for every $x\in V_{G'}\backslash\{u_1, u_2\}$, and $\omega'(e_1)=\omega(e_1)+\omega(G\backslash G_u)$, 
and $\omega'(e)=\omega(e)$ for every $e\in E_{G'}\backslash\{e_1\}$.
In other words, we add the total weight of the subgraph $G\backslash G_u$ to the edge $e_1$. 
Then, $\omega'(G')=\omega(G)>0$ and $|E_{G'}|=|E_{G_u}|<|E_G|$.
By the induction hypothesis, 
there is a basic factor $F\in \Omega'$ such that $\omega'(F)>0$.
We will recover  a factor of $\Omega$ from $F$ such that it has positive weight and fewer edges than $G$.
This will finish the proof of the inductive step. 

There are four cases depending on the presence of $e_1$ and $e_2$ in $F$.

\begin{itemize}
    \item 
 $e_1, e_2\notin E_F$.  
 Then, $u_1, u_2\notin V_F$.
 For every $x\in V_F$, $\deg_F(x)\in \pi'(x)=\pi(x)$. 
Thus,  $F$ is a basic factor of $\Omega$.
Clearly, $\omega(F)=\omega'(F)>0$ and $|E_F|=|E_{F'}|<|E_G|$.
We are done. 

\item  $e_1\in E_F$ and $e_2\notin E_F$. 
We can view $F$ as a subgraph of $G_u$ by changing the edge $(u_1, w_1)$ in $G'$ back to the edge $(u, w_1)$ in $G_u$. 
Then, the edge $(u, w_2)\notin E_F$.
Consider the subgraph $H=F\cup (G\backslash G_u)$ of $G$.
Since $(u, w_2)\notin E_F$, we have $(u, w_2)\notin E_H$. Then, $|E_H|<|E_G|$.  
Also, we have $$\omega(H)=\omega(F)+\omega(G\backslash G_u)=\omega'(F)>0.$$
The vertex set $V_H$ consists of three parts $V_1=V_F\backslash\{u\}$, $V_2=\{u\}$, and $V_3=V_{G\backslash G_u}\backslash\{u\}.$
For every $x\in V_1$, $\deg_H(x)=\deg_F(x)\in \pi(x)$.
For every $x\in V_3$, 
$\deg_H(x)=\deg_{G\backslash G_u}(x)=\deg_G(x)\in \pi(x)$.
Now, we consider the vertex $u$.
\begin{itemize}
    \item 

If $u$ is 2-feasible, then $\deg_H(u)=2\in \pi(x)$.
Thus, $H$ is a factor of $\Omega$ where $\omega(H)>0$ and $|E_H|<|E_G|$.

\item If $u$ is 1-feasible, then $F$ and $G\backslash G_u$ both are factors of $\Omega$ since $\deg_F(u)=\deg_{G\backslash G_u}(u)=1\in \pi(u)$.
Since $\omega(H)=\omega(F)+\omega(G\backslash G_u)>0$,
among $\omega(F)$ and $\omega(G\backslash G_u)$, at least one is positive.
Also, $|E_F|, |E_{G\backslash G_u}|<|E_H|<|E_G|$. 
We are done.
\end{itemize}

\item $e_2\in E_F$ and $e_1\notin E_F$. 
Again, we can view $F$ as a subgraph of $G_u$ where $(u, w_2)\in E_F$ and $(u, w_1)\notin E_F$. Then, we have $|E_F|<|E_{G_u}|<|E_G|$, and $\omega(F)=\omega'(F)>0$.
\begin{itemize}
    \item 
If $u$ is 1-feasible, then 
$F$ is a factor of $G$ where $|E_F|<|E_G|$ and $\omega(F)>0$. We are done. 

\item If $u$ is 2-feasible, then $G_u$ is a factor of $\Omega$ since $\deg_{G_u}(u)=2$. 
If $\omega(G_u)>0$, then we are done. 
Thus, we may assume that $\omega(G_u)\leq 0$.
Then, $\omega(G\backslash G_u)=\omega(G)-\omega(G\backslash G_u)\geq \omega(G)>0$.
Still consider the subgraph $H=F\cup (G\backslash G_u)$.
Then, $H$ is a factor of $\Omega$ since $\deg_H(u)=2\in \pi(u)$.
Also, $\omega(H)=\omega(F)+\omega(G\backslash G_u)>0$ and $|E_H|<|E_G|$.
We are done. 
\end{itemize}

\item $e_1, e_2\in E_F$. 
Then, $F$ (as a subgraph of $G'$) contains two vertices $u_1$ and $u_2$ of degree $1$.
Since $F$ is a basic factor, it is a path. 
Still we can view $F$ as a subgraph of $G_u$ by changing edges $(u_1, w_1)$ and $(u_2, w_2)$ in $G'$ to  edges  $(u, w_1)$ and $(u, w_2)$ in $G$.
Then, $F$ is a cycle in $G_u$.
Since $G_u$ is not a cycle and it has no isolated vertices, $|E_F|<|E_{G_u}|$.
Consider the subgraph $H=F\cup (G\backslash G_u)$ of $G$.
We have $|E_H|<|E_G|$ and $\omega(H)=\omega(F)+\omega(G\backslash G_u)=\omega'(F)>0$.
Also, one can check that $H$ is a factor of $G$ no matter whether $u$ is 1-feasible or 2-feasible since $\deg_H(u)=3\in \pi(u)$.
We are done.  \qedhere
\end{itemize}
\end{proof}

\subsection{Proof of the second property}
\label{subsec2}
Now we prove the second property of Theorem~\ref{lem-basic-factor-star} using the first property (Lemma~\ref{lem-MBFC}).

\begin{lemma}\label{lem-2}
Suppose that $\Omega=(G, \pi, \omega)$ is a key instance, and
 $u$ is a vertex of $G$ where $\deg_G(u)=1$ or $\deg_G(u)=3$ and  $\pi(u)=\{0, 2, 3\}$.
If $\omega(G)>0$ and $\omega(G)>\omega(F)$ for every basic factor $F$ of $\Omega$,
then
there is a basic factor $F^\ast$ of $\Omega$ such that $\omega(F^\ast)>0$ and $\deg_{F^\ast}(u)\equiv 0 \mod 2$. (Recall that we agree  $\deg_{F^\ast}(u)=0$ if $u\notin V_{F^\ast}$.)
\end{lemma}

\begin{proof}
By Lemma~\ref{lem-MBFC}, there exists at least one basic factor of $\Omega$ such that its weight is positive.
Among all such basic factors,  we pick an $F$ such that $\omega(F)$ is the largest.
We have $0<\omega(F)<\omega(G)$.
If $\deg_{F}(u)$ is even, then we are done.
Thus, we may assume that $\deg_{F}(u)$ is odd.
Since $F$ is a basic factor and it contains a vertex $u$ of odd degree, $F$ is not a cycle.
 By the definition of basic factors,  $F$ contains exactly one more vertex $v$ of odd degree.
Since $F$ is a factor of $\Omega$, $\deg_{F}(u)\subseteq \pi(u)$. 
Recall that $\deg_G(u)=1$ or $3$.
If $\deg_G(u)=1$, then $\pi(u)=\{0, 1\}$, and hence $\deg_F(u)=1$.
If $\deg_G(u)=3$, then $\pi(u)=\{0, 2, 3\}$, and hence $\deg_F(u)=3$.
Thus, $\deg_{F}(u)$ always equals $\deg_{G}(u)$.

Consider the graph $G'=G\backslash F$, i.e., the subgraph of $G$ induced by the edge set $E_G\backslash E_{F}$.
Consider the instance $\Omega'=(G', \pi', \omega')$ where for every $x\in V_{G'}$, $\pi'(x)=\{0, 1\}$ if $\deg_{G'}(x)=1$, $\pi'(x)=\{0, 2\}$ if $\deg_{G'}(x)=2$ and $\pi'(x)=\pi(x)$ if $\deg_{G'}(x)=3$, and $\omega'$ is the weight function $\omega$ restricted to $G'$.  
Note that $\Omega'$ is also a key instance, but it is not necessarily a sub-instance of $\Omega$.
Since $\omega(G)>\omega(F)$, we have $\omega'(G')=\omega(G')=\omega(G)-\omega(F)>0$.
Without causing ambiguity, we may simply write $\omega'$ as $\omega$ in the instance $\Omega'$.
By Lemma~\ref{lem-MBFC}, there exists a basic factor $F'$ of $\Omega'$ such that $\omega(F')>0$.
Since $E_{F'}\subseteq E_{G}\backslash E_{F}$, $F$ and $F'$ are edge-disjoint.
Let $H=F\cup F'$, which is the subgraph of $G$ induced by the edge set $E_F\cup E_{F'}$. 
We show that $H$ is a factor of $\Omega$.

Let $V_{\cap}=V_F\cap V_{F'}$. 
First we show that for every $x\in V_H\backslash V_{\cap}$, $\deg_H(x)\in \pi(x)$.
If $x\in V_F\backslash  V_{\cap}$, then $\deg_H(x)=\deg_F(x)$.
Since $F\in \Omega$, $\deg_F(x)\in \pi(x)$.
Then, $\deg_H(x)\in \pi(x)$. 
If $x\in V_{F'}\backslash  V_{\cap}$, then $\deg_H(x)=\deg_{F'}(x)$.
Since $x\notin V_{F}$ and $G'=G\backslash F$, $\deg_{G'}(x)=\deg_G(x)$.
Then,  by the definition of $\Omega'$, we have $\pi'(x)=\pi(x)$.
Since $F'$ is a factor of $\Omega'$, $\deg_{F'}(x)\in \pi'(x)$.
Thus, $\deg_H(x)\in \pi(x)$. 
Now, we consider vertices in $V_\cap$. 
Since $F$ and $F'$ are edge disjoint, for every $x\in V_{\cap}$  we have $\deg_H(x)=\deg_F(x)+\deg_{F'}(x)\leq \deg_{G}(x)\leq 3$. 
Also, $\deg_F(x), \deg_{F'}(x)\geq 1$ since $F$ and $F'$ are subcubic graphs which have no isolated vertices. 
\begin{itemize}
    \item 
If $\deg_F(x)=1$, then $1\in \pi(x)$. The vertex $x$ is 1-feasible. Thus, $\deg_G(x)\neq 2$.
Since $\deg_G(x)>\deg_F(x)=1$, $\deg_G(x)=3$.
Then, $\deg_{G'}(x)=\deg_G(x)-\deg_{F}(x)=2$, $\pi'(x)=\{0, 2\}$ and $\deg_{F'}(x)=2$.
\item
If $\deg_F(x)=2$, then $\deg_G(x)=3$ since  $\deg_G(x)>\deg_F(x)$. 
Then, $\deg_{G'}(x)=\deg_G(x)-\deg_{F}(x)=1$, $\pi'(x)=\{0, 1\}$ and $\deg_{F'}(x)=1$.
\end{itemize}
Thus, for every $x\in V_{\cap}$, $\deg_H(x)=\deg_F(x)+\deg_{F'}(x)=3\in \pi(x)$. Thus, $H$ is a factor of $\Omega$.

Consider the sub-instance $\Omega_H=(H, \pi_H, \omega_H)$ of $\Omega$ induced by $H$ (we will write $\omega_H$ as $\omega$ for simplicity).
We will show that we can find a 
a basic factor $F^\ast$ of $\Omega_H$ such that $\omega(F^\ast)>0$ and $\deg_{F^\ast}(u)\equiv 0 \mod 2$.
Clearly, $F^\ast$ is also a factor of $\Omega$. 

Consider the set $V_{\cap}$ of intersection points. 
If $V_{\cap}=\emptyset$, then for every $x\in V_{F'}$, $\deg_{F'}(x)=\deg_H(x)\in \pi(x)$.
Thus, $F'$ is a basic factor of $\Omega$ where $\omega(F')>0$ and $\deg_{F'}(u)=0$. That is, $F'$ is the desired $F^\ast$.
We are done.
Thus, we may assume that $V_{\cap}$ is non-empty. 
For every $x\in V_{\cap}$, $\deg_F(x)=1$ and $\deg_{F'}(x)=2$, or $\deg_F(x)=2$ and $\deg_{F'}(x)=1$.
Recall that $F$ is a basic factor containing two vertices $u, v$ of odd degree, and $\deg_F(u)=\deg_G(u)$.
Clearly,  $u\notin V_\cap$. 

We consider the possible forms of $F$ and $F'$. Recall that $F$ is not a cycle. We show that $F'$ is also not a cycle.
For a contradiction, suppose that $F'$ is a cycle. Then, all vertices of $F'$ have degree $2$. 
Thus, the only possible vertex in $V_\cap$ is $v$. 
Since $V_\cap$ is non-empty, $V_\cap=\{v\}$. 
Then, $\deg_{F}(v)=1$ and $\deg_{F'}(v)=2$. 
If $\deg_F(u)=1$, then $F$ is a path. The graph $H$ is a tadpole graph where $v$ is the only vertex of degree $3$.
If $\deg_F(u)=3$, then $F$ is a tadpole graph.
    The graph $H$ is a dumbbell graph where $v$ and $u$ are  the two vertices of degree $3$.
    In both cases, $H$ is a basic factor of $\Omega$.
    Since $\omega(F')>0$, we have
$\omega(H)=\omega(F)+\omega(F')>\omega(F)$ which leads to a contraction with $F$ being a basic factor with the largest weight. 
Thus, $F'$ is a basic factor which is not a cycle.
Then, it contains exactly two vertices $s, t$ of odd degree. 
Then, $V_\cap\subseteq\{v, s, t\}$. 

We consider the graph $H$ depending on the forms of $F$ and $F'$, and the vertices in $V_\cap$. 
There are 5 main cases. 
\begin{enumerate}
         \item[\Rmnum{1}.] $F$ is a path.
         \item[\Rmnum{2}.] $F$ is a tadpole graph and $\deg_F(u)=3$.
         \item[\Rmnum{3}.] $F$ is a tadpole graph and $\deg_F(u)=1$.
         \item[\Rmnum{4}.] $F$ is a dumbbell graph.
         \item[\Rmnum{5}.] $F$ is a theta graph.
\end{enumerate}

Recall that for two points $x$ and $y$, we use $p_{xy}$, $p'_{xy}$ or $p''_{xy}$ to denote a path with endpoints $x$ and $y$.
We also use $q_{xy^3}$ or $q'_{xy^3}$ to denote a tadpole graph where $x$ is the  vertex of degree $1$ and $y$ is the  vertex of degree $3$, and $\theta_{xy}$  to denote a theta graph where $x$ and $y$ are the two points of degree $3$.
In the following Figures~\ref{fig_1.a} to~\ref{fig:3.2.b}, we use hollow nodes to denote  $1$-feasible vertices, solid nodes to denote $2$-feasible vertices, semisolid nodes to denote vertices that are possibly $1$-feasible or $2$-feasible, red-colored lines to denote paths in $F$, and blue-colored lines to denote paths in $F'$.

\vspace{1ex}
\noindent {\bf Case \Rmnum{1}:} $F$ is a path. 
There are 4 subcases depending on the form of $F'$.
\begin{enumerate}
    \item[\Rmnum{1}.1]\label{case1.1} $F$ and $F'$ are both paths. Then, $V_\cap\subseteq \{v, s, t\}$.
    There are $5$ subcases: $V_\cap=\{v\}$, $V_\cap=\{s\}$ or $\{t\}$,
    $V_\cap=\{v, s\}$ or $\{v, s\}$, $V_\cap=\{s, t\}$, and $V_\cap=\{v, s, t\}$.
    \begin{enumerate}
        \item\label{case1.1.a} $V_\cap=\{v\}$.
        
               \begin{figure}[!htbp]
            \centering
            \includegraphics[width=7.2cm]{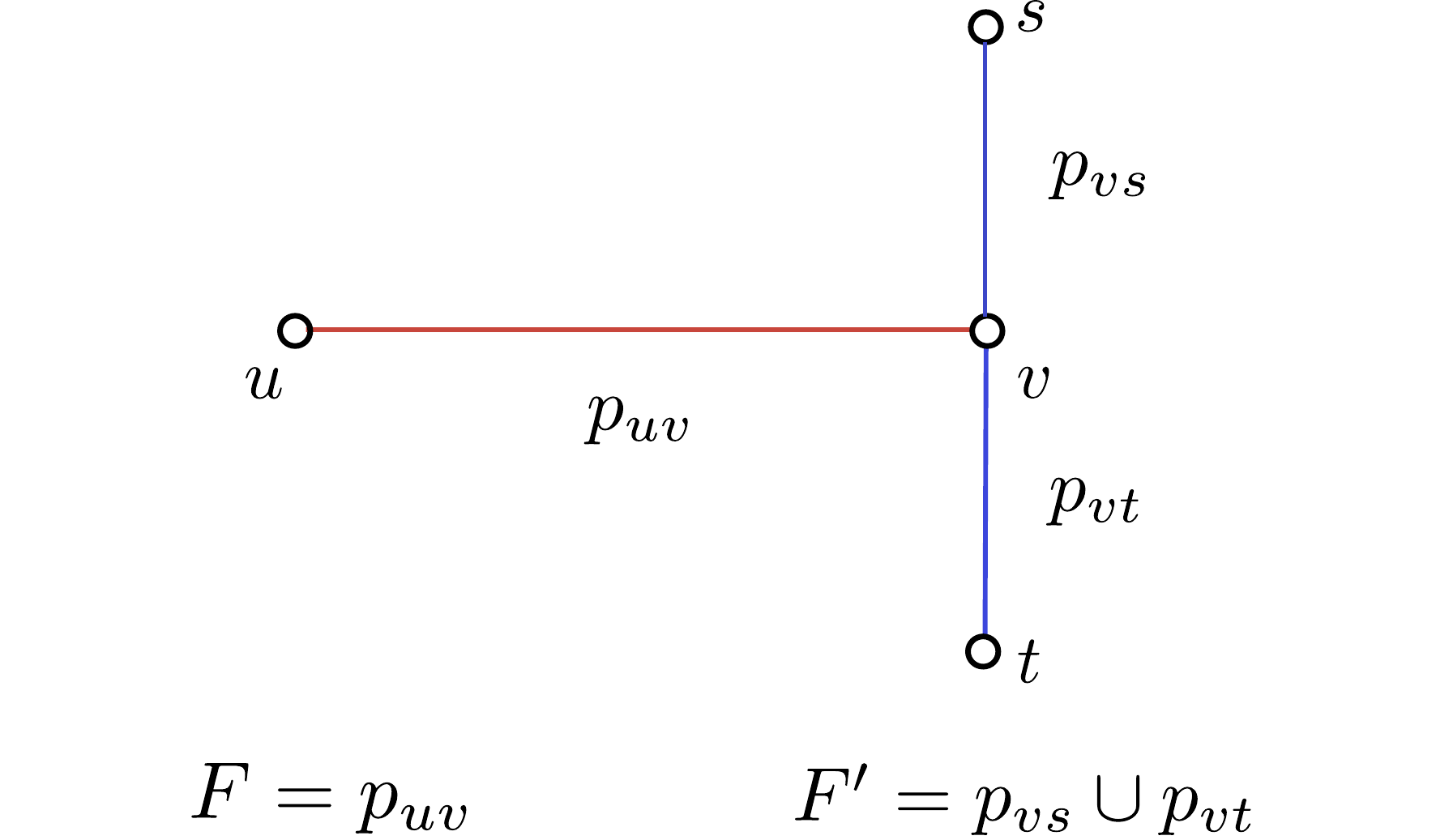}
            \caption{The graph $H$ in Case \Rmnum{1}.1.(a)}
            \label{fig_1.a}
        \end{figure}
        In this case,   $\deg_H(u)=\deg_H(s)=\deg_H(t)=1$, $\deg_H(v)=3$, and $\pi(v)=\{0, 1, 3\}$ since $\deg_F(v)=1\in \pi(v)$. 
        The graph $H$ consists of three edge-disjoint paths $p_{uv}$, $p_{vs}$ and $p_{vt}$.
         Then, $F=p_{uv}$ and $F'=p_{vs}\cup p_{vt}$.
         (See Figure~\ref{fig_1.a}.)

        Since $\omega(F')=\omega(p_{vs})+\omega(p_{vt})>0$, among $\omega(p_{vs})$ and $\omega(p_{vt})$, at least one is positive. 
        Without loss of generality, we may assume that $\omega(p_{vs})>0$. 
        Since $u$ does not appear in $p_{vs}$, we have $\deg_{p_{vs}}(u)=0$.
        For every vertex $x$ in $p_{vs}$ where $x\neq v$, $\deg_{p_{vs}}(x)=\deg_H(x)\in \pi(x)$.
        Also, $\deg_{p_{vs}}(v)=1\in \pi(v)$. 
        Thus, the path $p_{vs}$ is a basic factor of $\Omega$ where $\omega(p_{vs})>0$ and $\deg_{p_{vs}}(u)=0$. 
               
        \item\label{case1.1.b} $V_\cap=\{s\}$ or $\{t\}$.
        
        These two cases are symmetric. We only consider the case that $V_\cap=\{s\}$.

            \begin{figure}[!hbpt]
             \centering
             \includegraphics[width=7.2cm]{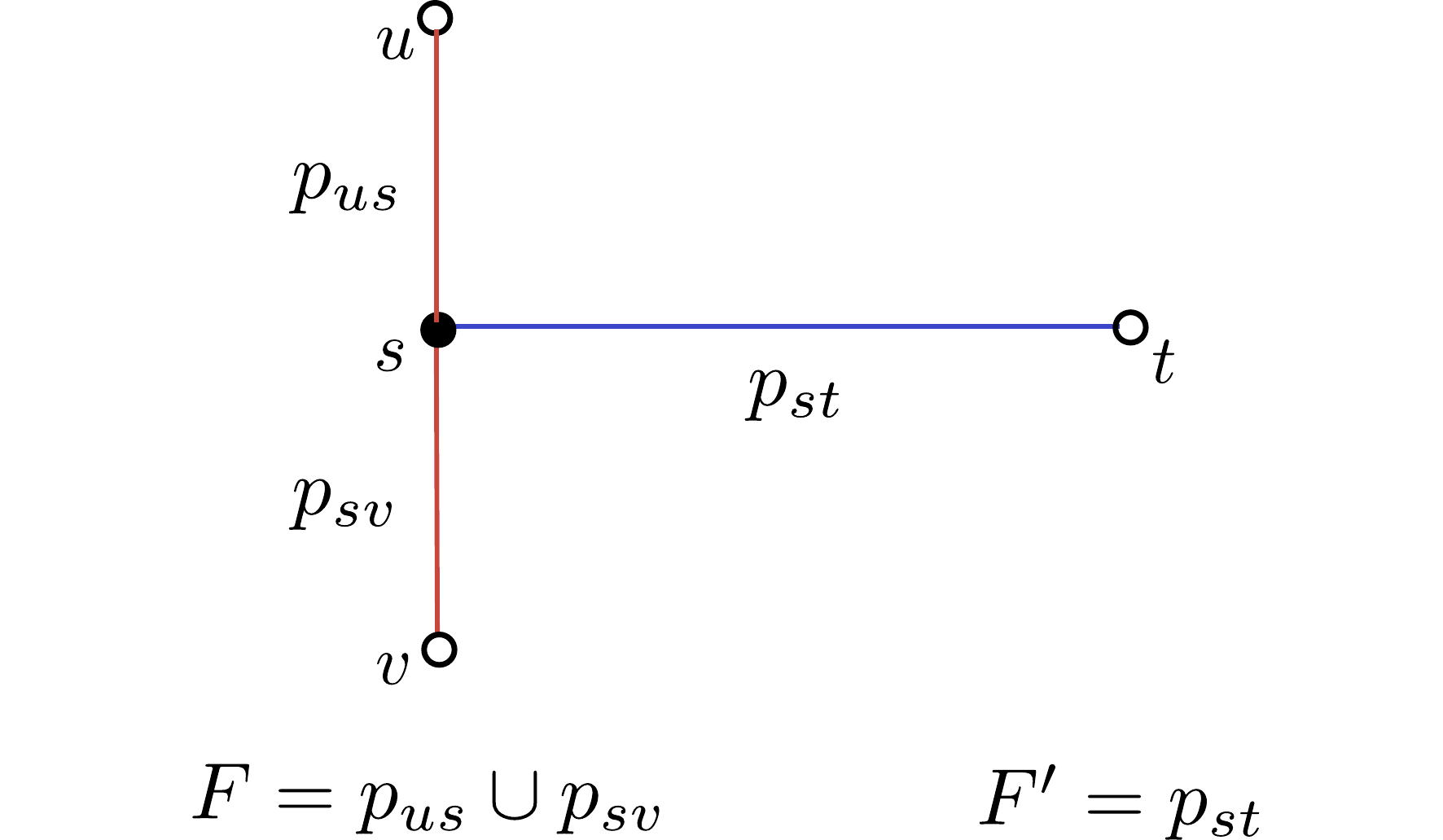}
             \caption{The graph $H$ in Case \Rmnum{1}.1.(b)}
             \label{fig:1.b}
         \end{figure}
        
        In this case, $\deg_H(s)=3$, $\deg_H(u)=\deg_H(v)=\deg_H(t)=1$, and $\pi(s)=\{0, 2, 3\}$ since $\deg_{F'}(s)=1$ and $\deg_F(s)=2\in \pi(s)$. 
             The graph $H$ consists of three edge-disjoint paths $p_{us}$, $p_{sv}$ and $p_{st}$.
           Then, $F'=p_{st}$ and  $F=p_{us}\cup p_{sv}$. (See Figure~\ref{fig:1.b}.)

             Note that $\omega(p_{st})=\omega(F')>0$.
             Let $p_{ut}=p_{us}\cup p_{st}$
            be the path with endpoints $u$ and $t$.
             For every vertex $x$ in $p_{ut}$ where $x\neq s$, we have $\deg_{p_{ut}}(x)=\deg_H(x)\in \pi(x)$.
             Also, $\deg_{p_{ut}}(s)=2\in \pi(s)$. 
             Thus, $p_{ut}$ is a  basic factor of $\Omega$. 
             Then, $\omega(F)\geq \omega(p_{ut})$ since $F$ is a basic factor of $\Omega$ with the largest weight $\omega(F)$.
             Then, $$\omega(F)=\omega(p_{us})+\omega(p_{sv})\geq \omega(p_{us})+\omega(p_{st})=\omega(p_{ut}).$$
             Thus, $\omega(p_{sv})\geq \omega(p_{st})>0$.
             Let $p_{vt}=p_{sv}\cup p_{st}$ be the path
             with endpoints $v$ and $t$.
             Then, $$\omega(p_{vt})=\omega(p_{sv})+\omega(p_{st})>0.$$
             Since $u$ is not in $p_{vt}$, $\deg_{p_{vt}}(u)=0$.
             Similar to the proof of $p_{ut}\in \Omega$, we have $p_{vt}\in \Omega$. 
             Thus, the path $p_{vt}$ is a basic factor of $\Omega$ where $\omega(p_{vt})>0$ and $\deg_{p_{vt}}(u)=0$. 
             \item\label{case1.1.c} $V_\cap=\{v, s\}$ or $\{v, t\}$. 
         
           These two cases are symmetric. We only consider the case that $V_\cap=\{v, s\}$. 
           
              \begin{figure}[!htpb]
               \centering
               \includegraphics[height=2.8cm]{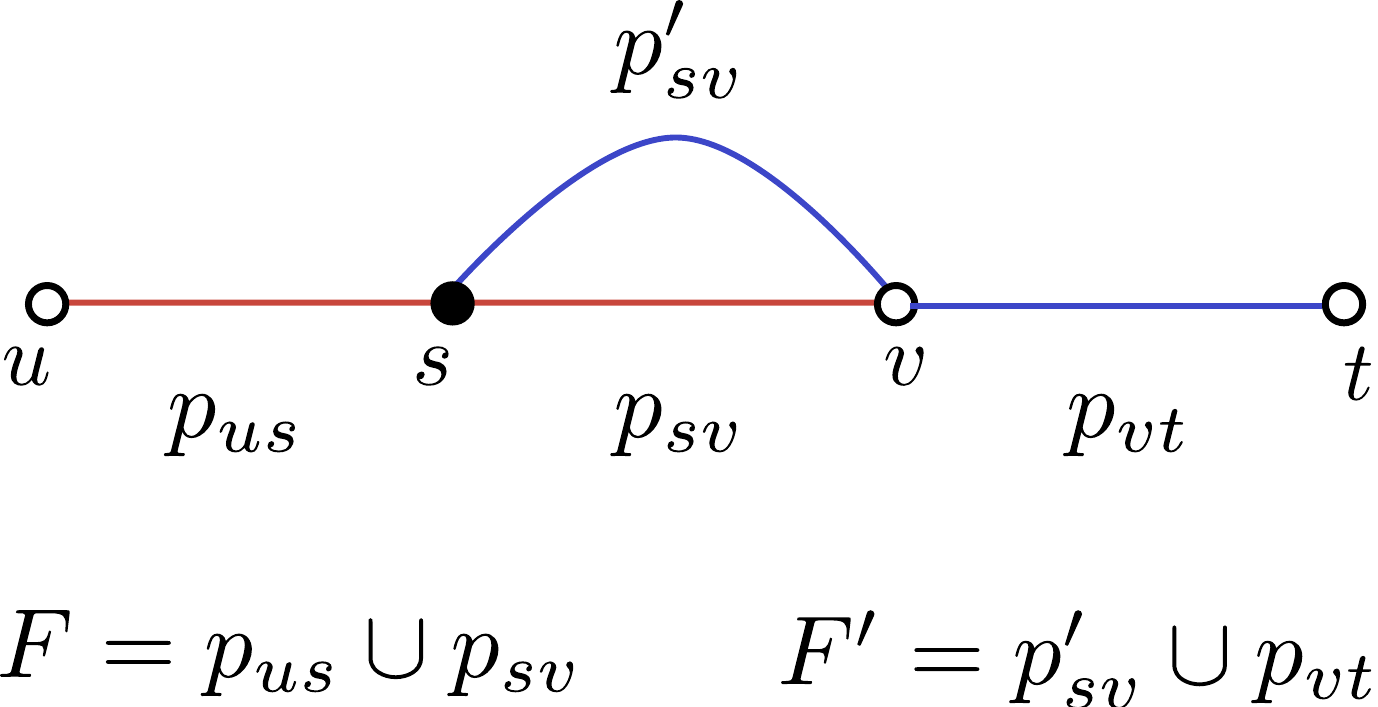}
               \caption{The graph $H$ in Case \Rmnum{1}.1.(c)}
               \label{fig:1.c}
           \end{figure}
           
           In this case, $\deg_H(v)=\deg_H(s)=3$,  $\deg_H(u)=\deg_H(t)=1$, $\pi(v)=\{0, 1, 3\}$ since $\deg_F(v)=1\in \pi(v)$, and $\pi(s)=\{0, 2, 3\}$ since $\deg_F(s)=2\in \pi(s)$.  
           The point $s$ splits $F$ into two paths $p_{us}$ and $p_{sv}$.
           Then, $F=p_{us}\cup p_{sv}$.
           The point $v$ splits  $F'$ into two paths $p'_{sv}$ and $p_{vt}$. 
            Then, $F'=p'_{sv}\cup p_{vt}$. 
             (See Figure~\ref{fig:1.c}.)

            Consider the path $p'_{uv}=p_{us}\cup p'_{sv}$. 
            Note that $\deg_{p'_{uv}}(s)=2\in \pi(s)$.
             Then,  $p'_{uv}$ is a basic factor of $\Omega$.
             Since, $F$ is a basic factor of $\Omega$ with the largest weight,
             we have
$$\omega(F)=\omega(p_{us})+\omega(p_{sv})\geq \omega(p_{us})+\omega(p'_{sv})=\omega(p'_{uv}).$$
         Thus,   $\omega(p_{sv})\geq \omega(p'_{sv})$.
         Let $F^\ast$ be the tadpole graph  $q_{tv^3}=p_{sv}\cup p'_{sv} \cup p_{vt}$.
         Note that $\deg_{F^\ast}(s)=2\in \pi(s)$ and $\deg_{F^\ast}(v)=3\in \pi(v)$.
         Then, $F^\ast$ is a basic factor of $\Omega$
         and $\deg_{F^\ast}(u)=0$.
         Also, the path $p_{vt}$ is a basic factor of $\Omega$ since
         $\deg_{p_{vt}}(v)=1\in \pi(v)$,
         and $\deg_{p_{vt}}(u)=0$.
         Then, $$\omega(F^\ast)+\omega(p_{vt})=\omega(p_{sv})+\omega(p'_{sv})+\omega(p_{vt})+\omega(p_{vt})\geq 2(\omega(p'_{sv})+\omega(p_{vt}))=2\omega(F')>0.$$
         Thus, among $\omega(F^\ast)$ and $\omega(p_{vt})$, at least one is positive.
         Thus, $F^\ast$ or $p_{vt}$ is a desired basic factor of $\Omega$ that satisfies the requirements. 
         \item\label{case1.1.d} $V_{\cap}=\{s, t\}$.
            \begin{figure}[!h]
             \centering
  \includegraphics[height=2.7cm]{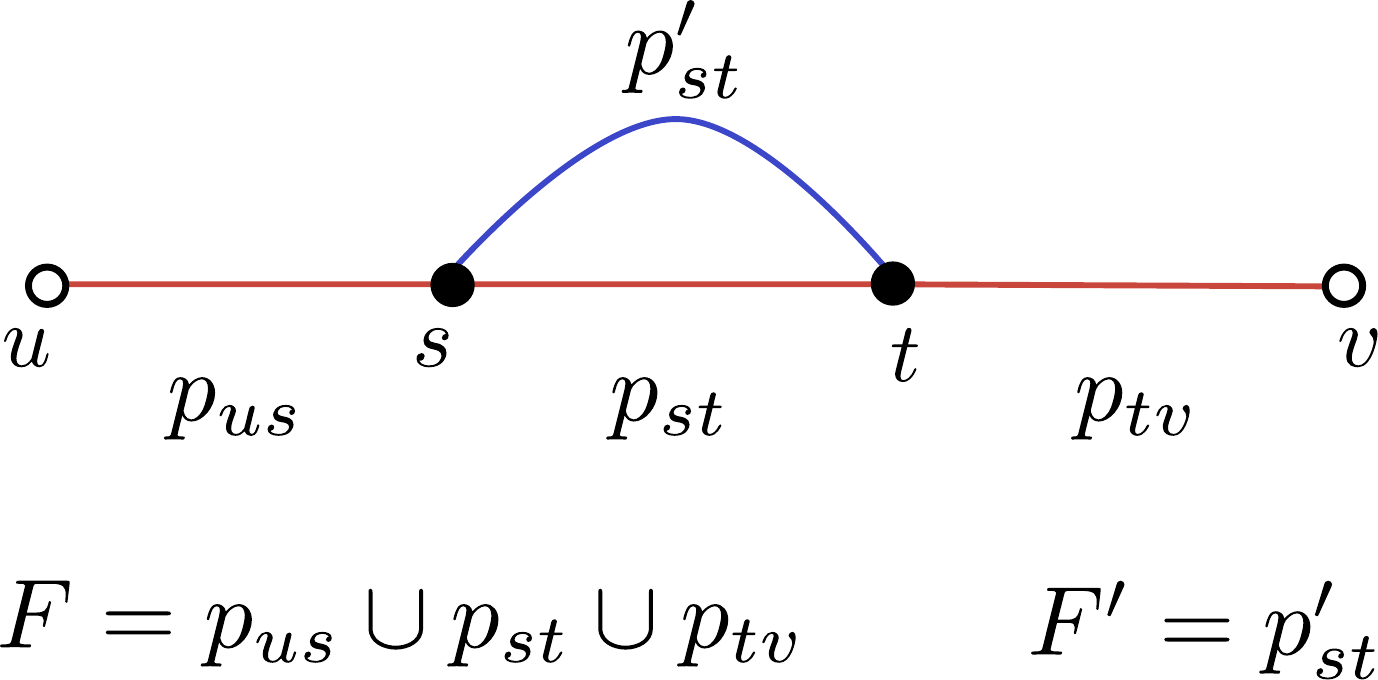}
             \caption{The graph $H$ in Case \Rmnum{1}.1.(d)}
             \label{fig:1.d}
         \end{figure}      
         
         In this case, $\deg_H(u)=\deg_H(v)=1$, $\deg_H(s)=\deg_H(t)=3$, and $\pi(s)=\pi(t)=\{0, 2, 3\}$. 
         The points $s$ and $t$ split $F$ into three paths. 
         Without loss of generality, we may assume that $s$ is closer to $u$ and $t$ is closer to $v$.
         Then, the three paths are $p_{us}$, 
         $p_{st}$, 
         and $p_{tv}$, and $F=p_{us}\cup p_{st} \cup p_{ts}$.
         Also, $F'$ is a path with endpoints $s$ and $t$, which is disjoint with $p_{st}$.
          (See Figure~\ref{fig:1.d}.)
         
         Consider the path $p'_{uv}=p_{us}\cup F'\cup p_{tv}$.
         One can check that $p'_{uv}$ is a basic factor of $\Omega$.
         Then, $ \omega(F)\geq \omega(p'_{uv})$.
         Thus, $ \omega(p_{st})\geq \omega(F')>0$.
         Consider the cycle $F^\ast=F'\cup p_{st}$.
         Also, one can check that $F^\ast$ is a basic factor of $\Omega$. 
         Moreover, $\omega(F^\ast)=\omega(F')+\omega(p_{st})>0$ and $\deg_{F^\ast}(u)=0$. We are done.
         \item $V_{\cap}=\{v, s, t\}.$
         
         \begin{figure}[!h]
              \centering
    \includegraphics[height=2.7cm]{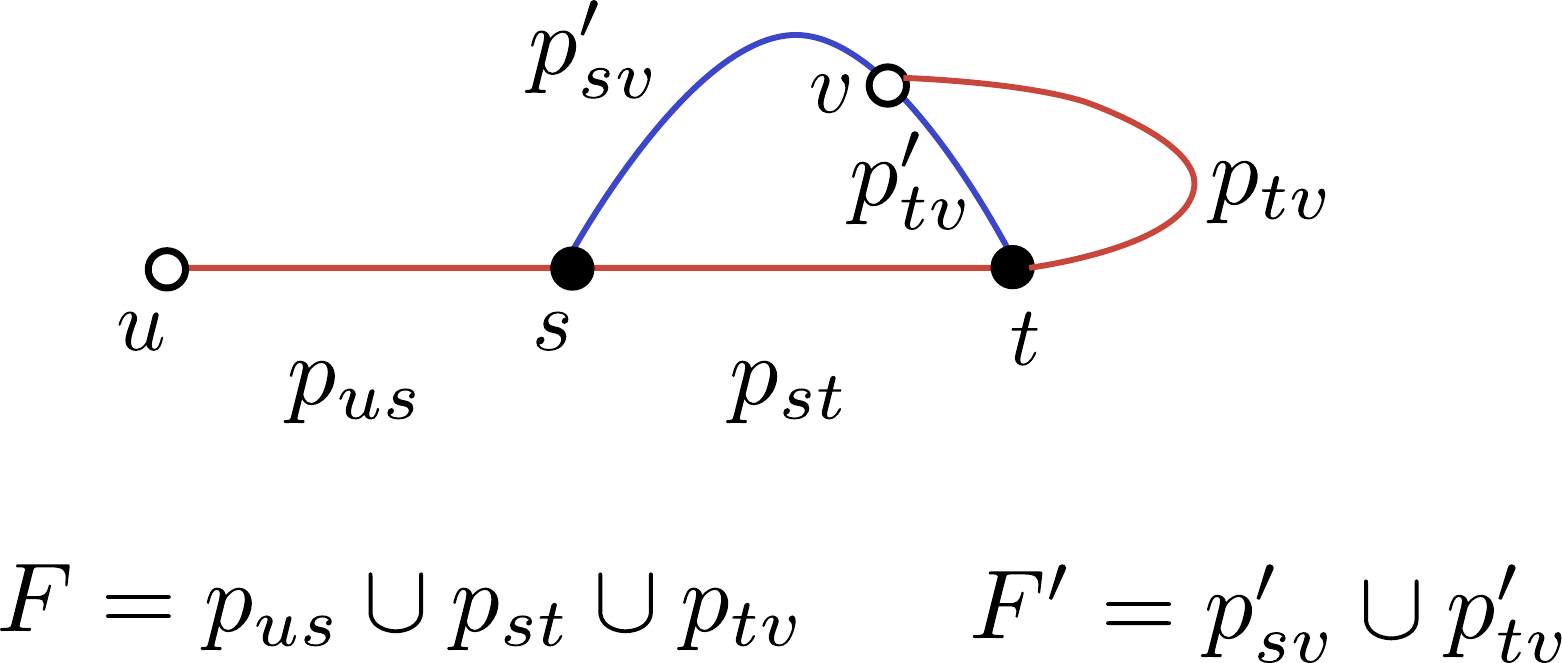}
              \caption{The graph $H$ in Case \Rmnum{1}.1.(e)}
              \label{fig:1.e}
          \end{figure}

         In this case,  $\deg_H(u)=1$, $\deg_H(v)=\deg_H(s)=\deg_H(t)=3$, $\pi(v)=\{0, 1, 3\}$, and $\pi(s)=\pi(t)=\{0, 2, 3\}$. 
          The points $s$ and $t$ split $F$ into three paths. 
          Without loss of generality, we assume that they are $p_{us}$, 
          $p_{st}$, 
          and $p_{tv}$. 
          Then, $F=p_{us}\cup p_{st} \cup p_{tv}$.
          The point $v$ splits $F'$ into two paths, $p'_{sv}$ 
          and $p'_{tv}$.
         Then, $F'=p'_{sv}\cup p'_{tv}$.
         (See Figure~\ref{fig:1.e}.)

          Consider the path $p'_{uv}=p_{us}\cup p'_{sv}$.
          One can check that it is a basic factor of $\Omega$.
          Since $\omega(F)\geq \omega(p'_{uv})$, we have $$\omega(p_{st})+\omega(p_{tv})\geq \omega(p'_{sv}).$$
          Consider the path $p''_{uv}=p_{us}\cup p_{st}\cup p'_{tv}$.
         One can check that it is also a basic factor of $\Omega$.
          Since $\omega{(F)}\geq \omega(p''_{uv})$, we have 
          $$\omega(p_{tv})\geq \omega(p'_{tv}).$$
          Consider the tadpole graph $q_{uv^3}=p_{us}\cup p'_{sv}\cup p'_{tv} \cup p_{tv}$.
          One can check that it is also a basic factor of $\Omega$.
           Since $\omega{(F)}\geq \omega(q_{uv^3})$, we have
          $$\omega(p_{st})\geq \omega(p'_{sv})+\omega(p'_{tv}).$$
          Sum up the above three inequalities, and we have 
          $$2(\omega(p_{st})+\omega(p_{tv}))\geq 2(\omega(p'_{sv})+\omega(p'_{tv}))=2\omega(F')>0.$$
          Consider the theta graph $F^\ast=p_{st}\cup p_{tv}\cup p'_{sv}\cup p'_{tv} .$
          Still one can check that it is a basic factor of $\Omega$.
          Moreover, $$\omega(F^\ast)=\omega(p_{st})+\omega(p_{tv})+\omega(p'_{sv})+\omega(p'_{tv})>0$$ and $\deg_{F^\ast}(u)=0$.
          We are done.
          \end{enumerate}
           We are done with Case~\href{case1.1}{\Rmnum{1}.1} where $F$ and $F'$ are both paths. 
           \item[\Rmnum{1}.2.]\label{case1.2} $F$ is a path and $F'$ is a tadpole graph. Without loss of generality, we may assume  that $\deg_{F'}(s)=1$ and $\deg_{F'}(t)=3$. 
           In other words, $F'$ consists of a path with endpoints $s$ and $t$, and a cycle $C$ containing the vertex $t$.
           Then, $V_{\cap}\subseteq \{v, s\}$.
           There are three subcases: $V_{\cap}=\{v\}$, $V_{\cap}=\{s\}$, and $V_{\cap}=\{v, s\}$. 
           \begin{enumerate}
               \item $V_{\cap}=\{v\}$.
               
               In this case, $\deg_H(u)=\deg_H(s)=1$, $\deg_H(v)=\deg_H(t)=3$, $\pi(v)=\{0, 1, 3\}$, and   $\pi(t)=\{0, 1, 3\}$ or $\{0, 2, 3\}$.
               There are two subcases depending on whether the intersection point $v$ appears in the path part or the cycle part of  $F'$.

               \begin{enumerate}
                   \item $v$ appears in the path part. 

                   Note that for every $x\in V_C\backslash \{t\}$, $\deg_H(x)=2$, and $\deg_H(t)=3$.
                   We say such a cycle with exactly one vertex of degree $3$ in $H$ is a \emph{dangling} cycle in $H$.
                   Let $e_t$ be the edge incident to $t$ where $e_t\notin E_C$. 
                   We call the vertex $t$ 
                   the \emph{connecting point} of $C$, and the edge $e_t$ 
                   the \emph{connecting bridge} of $C$.

                   Consider the graph $H'=H\backslash C$.
Notice that $\deg_{H'}(x)=\deg_H(x)$ for every $x\in V_{H'}\backslash \{t\}$ and $\deg_{H'}(t)=1$.
Consider the instance $\Omega_{H'}=(H', \pi_{H'}, \omega_{H'})$ where $\pi_{H'}(x)=\pi_H(x)$ for every $x\in V_{H'}\backslash\{t\}$ and $\pi_{H'}(t)=\{0, 1\}$, and $\omega_{H'}(e)=\omega(e)$ for every $e\in E_{H'}\backslash \{e_t\}$ and $\omega_{H'}(e_t)=\omega(e_t)+\omega(C)$.
In other words, the instance $\Omega_{H'}$ is obtained from $\Omega_H$ by contracting the dangling cycle $C$ to its connecting point $t$ and adding the total weight of $C$ to its connecting bridge $e_t$.
Clearly, $\Omega_{H'}$ is a key instance and $\omega_{H'}(H')=\omega(H')+\omega(C)=\omega(H)>0$.

For every factor $K'\in \Omega_{H'}$, we can recover a factor $K\in \Omega_H$ from $K'$ as follows: $K=K'$ if $e_t\notin E_{K'}$ and $K=K'\cup C$ if $e_t\in E_K$. 
One can check that $K$ is a factor of $\Omega_H$, and $\omega(K)=\omega_{H'}(K')$.
If $e_t\notin K$, then $K'=K$.
Clearly, $K'$ is a basic factor of $\Omega_{H'}$ if and only if $K$ is a basic factor of $\Omega$.
Now, suppose that $e_t\in K$.
Remember that $\deg_{H'}(t)=1$.
Then, $K'$ is a path with $t$ as an endpoint if and only if $K=K'\cup C$ is a tadpole graph with $t$ as the vertex of degree $3$, and
 $K'$ is a tadpole with $t$ as the vertex of degree $1$ if and only if $K=K'\cup C$ is a dumbbell graph.
Thus,
$K'$ is a basic factor of $\Omega_{H'}$ if and only if $K$ is a basic factor of $\Omega_H$. 

Notice that the instance $\Omega_H'$ has a similar structure to the instance $\Omega_H$ in Case \href{case.1.1.a}{\Rmnum{1}.1.(a)}.
By replacing the vertex $v$ 
                   in Case \href{case.1.1.a}{\Rmnum{1}.1.(a)} by the cycle $C$ (and re-arranging the weights between the cycle $C$ and its connecting bridge), 
                   one can check that the proof of Case \href{case.1.1.a}{\Rmnum{1}.1.(a)}   works here.
                   Note that after this replacement, the path $p_{vt}$  in Case \href{case.1.1.a}{\Rmnum{1}.1.(a)} becomes a tadpole graph $q_{vt^3}$ which is still a basic factor. 

                   \item $v$ appears in the cycle part. 
                   
                                        \begin{figure}[!htpb]
                \centering
                \includegraphics[height=3.3cm]{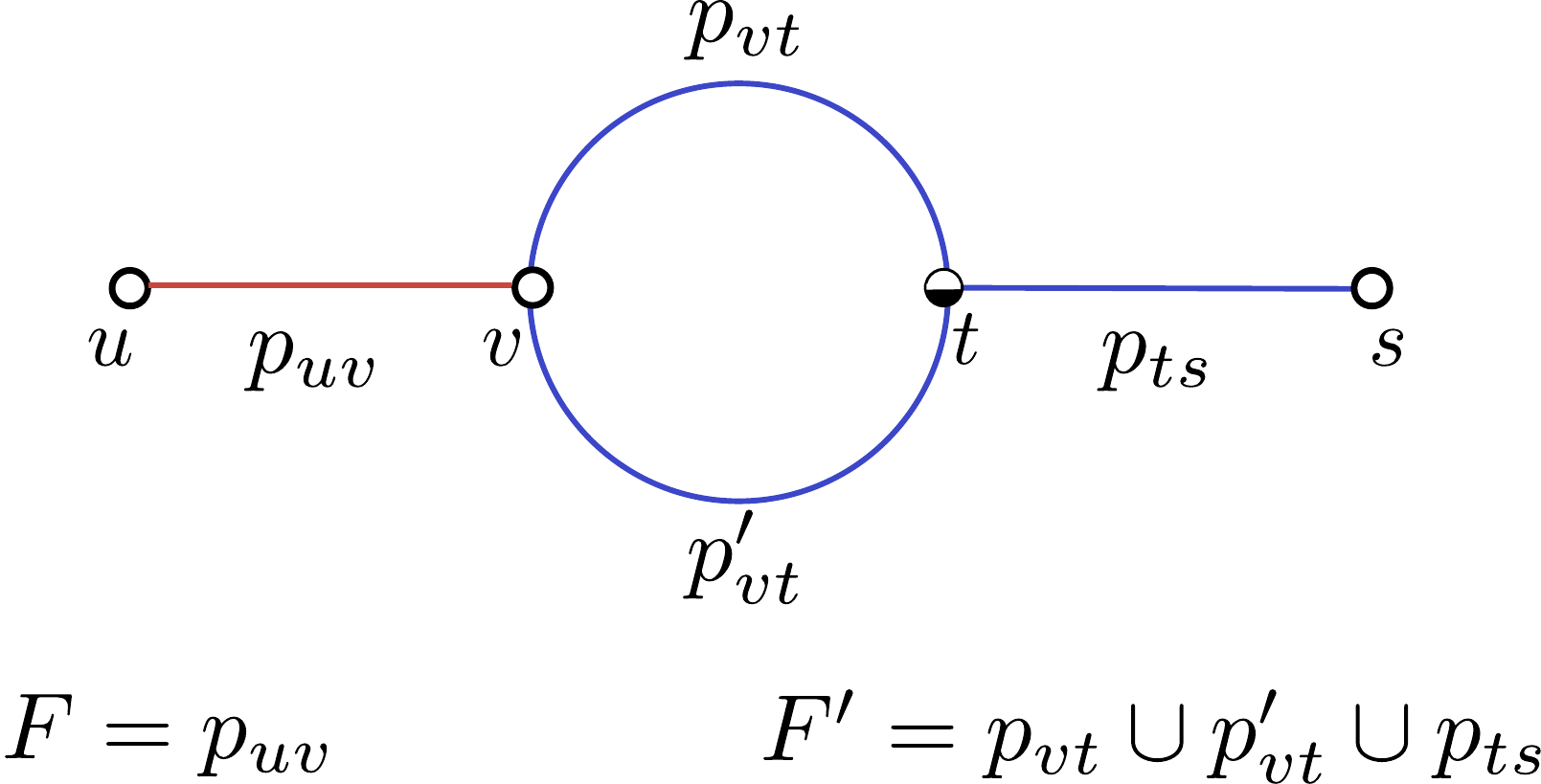}
                \caption{The graph $H$ in Case \Rmnum{1}.2.(a)}
                \label{fig:2.a}
            \end{figure}
            
                   Together with the point $t$, the point $v$ splits the cycle in $F'$ into two paths $p_{vt}$ and $p'_{vt}$.
                   Let $p_{ts}$ denote the path in $F'$ with endpoints $t$ and $s$.
                   Then, $F'=p_{vt}\cup p'_{vt}\cup p_{ts}$.
                   (See Figures~\ref{fig:2.a}.)

                   \begin{itemize}
                       \item If $\pi(t)=\{0, 1, 3\}$, then  the paths $p_{vt}$, $p'_{vt}$ and $p_{ts}$ are all basic factors of $\Omega$. 
                       Moreover, the vertex $u$ does not appear in any of these paths. 
                       Also, since $\omega(F')=\omega(p_{vt})+\omega(p'_{vt})+\omega(p_{ts})>0$,
                       there is at least one path with positive weight.
                       We are done.
                       \item If $\pi(t)=\{0, 2, 3\}$, then  the tadpole graph $q_{u v^3}=F\cup p_{vt} \cup p'_{vt}$ is a basic factor or $\Omega$.
                       Since $\omega(F)\geq \omega(q_{uv^3})$, we have $\omega(p_{vt})+\omega(p'_{vt})\leq 0.$
                       Without loss of generality, we assume that $\omega(p'_{vt})\leq 0$. 
                       Consider the path $F^\ast=p_{vt}\cup p_{ts}$.
                        We have $\deg_{F^\ast}(u)=0$, and $F^\ast$
                        is a basic factor of $\Omega$.
                       Since $$\omega(F')=\omega(p_{vt})+\omega(p'_{vt})+\omega(p_{ts})=\omega(F^\ast)+\omega(p'_{vt})>0$$ and $\omega(p'_{vt})\leq 0,$
                      we have $\omega(F^\ast)>0$.
                      We are done.
                   \end{itemize}
               \end{enumerate}
               \item $V_{\cap}=\{s\}$. 
               
               Still,  
               the cycle $C$ in $F'$ is a dangling cycle with the connecting point $t$.
               We  can contract $C$ to 
               $t$ and add the weight $\omega(C)$ to its connecting bridge. 
               Then,  this case is similar to Case \href{case1.1.b}{\Rmnum{1}.1.(b)}.
                By replacing the vertex $t$ in
                Case \href{case1.1.b}{\Rmnum{1}.1.(b)} by 
                the $C$,
                one can check that the proof of Case \href{case1.1.b}{\Rmnum{1}.1.(b)} works here. 
                Note that the path $p_{st}$ in 
                Case \href{case1.1.b}{\Rmnum{1}.1.(b)} is replaced by a tadpole graph $q_{st^3}$ which is still a basic factor.
               \item $V_{\cap}=\{v, s\}$.
               
               In this case, $\deg_H(u)=1$, $\deg_H(v)=\deg_H(s)=\deg_H(t)=3$, 
               $\pi(v)=\{0, 1, 3\}$,
               $\pi(s)=\{0, 2, 3\}$,
               and $\pi(t)=\{0, 1, 3\}$ or $\{0, 2, 3\}$.
               There are two subcases depending on whether the intersection point $v$ appears in the path part or the cycle part of the tadpole graph $F'$.
               Note that there is only one way for the intersection point $s$ to appear in the path $F$, and $s$ always splits  $F$ into two paths $p_{us}$ and $p_{st}$.
               \begin{enumerate}
                   \item $v$ appears in the path part. 
                   
                   Still, the cycle $C$ is a dangling cycle with the connecting point $t$.
                   This case is similar to Case  \href{case1.1.c}{\Rmnum{1}.1.(c)}.
                   By replacing  the vertex $t$
                   in Case \href{case1.1.c}{\Rmnum{1}.1.(c)} by  the cycle $C$,
                   one can check that the proof of Case \href{case1.1.c}{\Rmnum{1}.1.(c)} works here.
                  Note that  the path $p_{vt}$ and the tadpole graph $F^\ast=q_{tv^3}$
                   in Case  \href{case1.1.c}{\Rmnum{1}.1.(c)} are replaced by a tadpole graph and a dumbbell graph respectively. Both are still  basic factors. 
                   
                   \item $v$ appears in the cycle part. 
                   
                   \begin{figure}[!htpb]
                \centering
                \includegraphics[width=7.2cm]{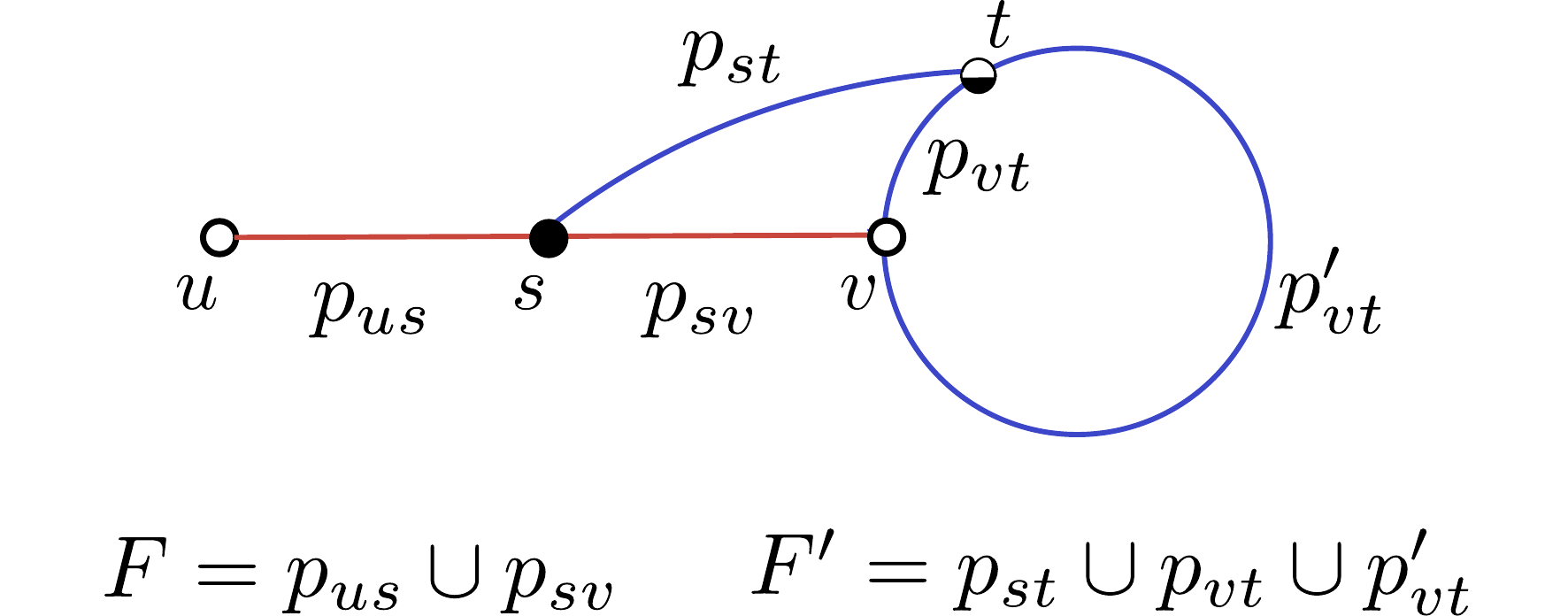}
                \caption{The graph $H$ in Case \Rmnum{1}.2.(c)}
                \label{fig:2.c}
            \end{figure}

                  Together with the point $t$, the point $v$ splits the cycle in $F'$ into two parts $p_{vt}$ and $p'_{vt}$. 
                  Let $p_{st}$ denote the path in $F'$ with endpoints $s$ and $t$. Then, $F'=p_{st}\cup p_{vt}\cup p'_{vt}$.
                  (See Figure~\ref{fig:2.c}.)
                  \begin{itemize}
                      \item If $\pi(t)=\{0, 1, 3\}$, then the paths $p_{vt}$ and $p'_{vt}$ are both basic factors of $\Omega$.
                      If $\omega(p_{vt})>0$ or $\omega(p'_{vt})>0$, then we have a basic factor satisfying the requirements. 
                      Thus, we may assume $\omega(p_{vt}), \omega(p'_{vt})\leq 0$.
                      Since $\omega(F')=\omega(p_{st})+\omega(p_{vt})+\omega(p'_{vt})>0$, 
                      we have $\omega(p_{st})>0$.
                      Consider the path $p_{ut}=p_{us}\cup p_{st}$.
                      It is a basic factor of $\Omega$.
                      Since $F$ is a basic factor of $\Omega$ with the largest weight,  $$\omega(F)=\omega(p_{us})+\omega(p_{sv})\geq\omega(p_{us})+\omega(p_{st})=\omega(p_{ut}).$$
                      Then, $\omega(p_{sv})\geq \omega(p_{st})>0$.
                      
                      Consider the path $p''_{vt}=p_{sv}\cup p_{st}$. 
                      It is a basic factor of $\Omega$ and $\deg_{p''_{vt}}(u)=0$.
                      Also, $\omega(p''_{vt})=\omega(p_{sv})+\omega(p_{st})>0$. We are done.
                      \item 
                      If $\pi(t)=\{0, 2, 3\}$, then  the tadpole graph $q_{uv^3}=F\cup p_{vt} \cup p'_{vt}$ is a basic factor of $\Omega$.
                       Since $\omega(F)\geq \omega(q_{uv^3})$, we have $\omega(p_{vt})+\omega(p'_{vt})\leq 0.$
                       Since $\omega(F')>0$, we have $\omega(p_{st})>0$.
                       Consider the path $p'_{uv}=p_{us}\cup p_{st}\cup p_{vt}$. 
                       It is a basic factor of $\Omega$.
                       Since $\omega(F)\geq \omega(p'_{uv})$,
                       we have $$\omega(p_{sv})\geq \omega(p_{st})+\omega(p_{vt}).$$
                       Similarly, consider the path $p''_{uv}=p_{us}\cup p_{st}\cup p'_{vt}$. 
                       We have $$\omega(p_{sv})\geq \omega(p_{st})+\omega(p'_{vt}).$$
                       Sum up the above two inequalities, we have $$2\omega(p_{sv})\geq 2\omega(p_{st})+\omega(p_{vt})+\omega(p'_{vt})\geq 2\omega(p_{st}) >0.$$
                       
                       Consider the theta graph $F^\ast=p_{sv}\cup F'$. 
                       Note that $F^\ast$ is a basic factor of $\Omega$ and $\deg_{F^\ast}(u)=0$. 
                       Also, $\omega(F^\ast)=\omega(p_{sv})+\omega(F')>0$. We are done. 
                  \end{itemize}
               \end{enumerate}
           \end{enumerate}
            We are done with Case~\href{case1.2}{\Rmnum{1}.2} where $F$ is a path and $F'$ is a tadpole graph.

                          \item[\Rmnum{1}.3.]
               $F$ is a path and $F'$ is a dumbbell graph. 
               Then, $V_{\cap}=\{v\}$. 
               
               Let $C_s$ and $C_t$ be the two cycles in $F'$ that contain vertices $s$ and $t$ respectively. 
               Clearly, among $V_{C_s}$ and $V_{C_t}$, there exists at least one such that it does not contain the intersection point $v$.
               Notice that vertices $s$ and $t$ are symmetric in this case. 
               Without loss of generality, we may assume that $v\notin V_{C_s}$.   
               Then, $V_{C_s}\cap V_{F}=\emptyset$.
               Thus, $C_s$ is a dangling cycle with the connecting point $s$.
               Then, 
               this case is similar to Case~\href{case1.2.a}{\Rmnum{1}.2.(a)}. 
               By replacing the vertex $s$ in Case~\href{case1.2.a}{\Rmnum{1}.2.(a)} by the cycle $C_s$, 
               one can check that the proof of Case~\href{case1.1.a}{\Rmnum{1}.2.(a)} works here. 
                          
             \item[\Rmnum{1}.4.]   $F$ is a path and $F'$ is a theta graph.
           Then, $V_{\cap}=\{v\}$. 
\begin{figure}[!htbp]
    \centering
    \includegraphics[height=3.2cm]{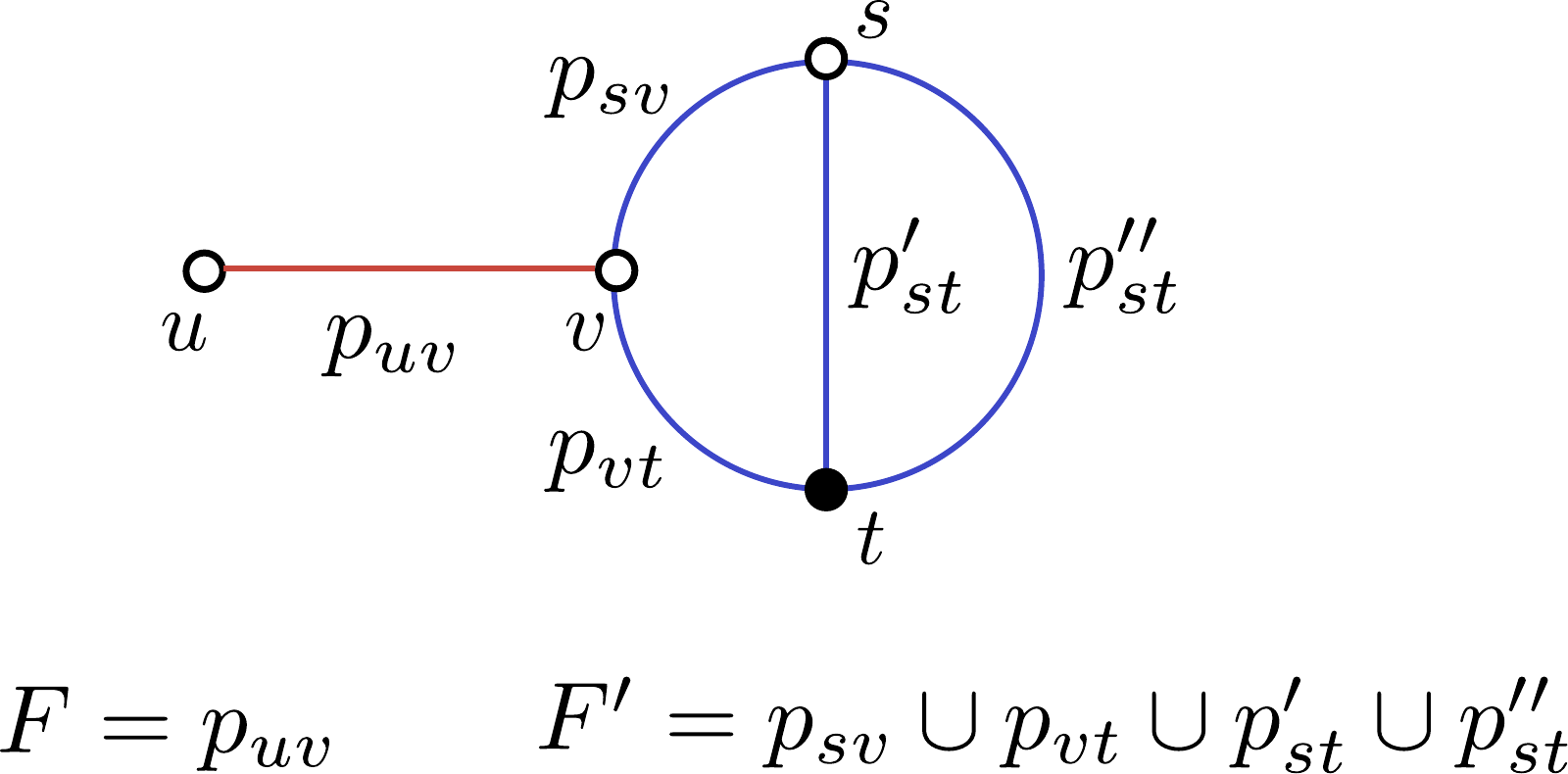}
    \caption{The graph $H$ in Case \Rmnum{1}.4}
    \label{fig:4}
\end{figure}
           
            In this case, $\deg_H(u)=1$, $\deg_H(v)=\deg_H(s)=\deg_H(t)=3$, and
               $\pi(v)=\{0, 1, 3\}$.
      Since $F'$ is a theta graph and $\deg_{F'}(s)=\deg_{F'}(t)=3$, without loss of generality, we may assume that $\pi(s)=\{0, 1, 3\}$ and $\pi(t)=\{0, 2, 3\}$. 
           $F'$ consists of three paths $p_{st}$, $p'_{st}$ and $p''_{st}$.
           Without loss of generality, we may assume that $v$ appears in the path $p_{st}$ and it splits the path into two paths $p_{sv}$ and $p_{vt}$.
           (see Figure~\ref{fig:4}.)
           
           Consider the paths $p'_{sv}=p'_{st}\cup p_{vt}$ and $p''_{sv}=p''_{st}\cup p_{vt}$, 
           and the tadpole graph $q_{vs^3}=p_{sv}\cup p'_{st}\cup p''_{st}$.
           It can be checked that $p'_{sv}$, $p''_{sv}$ and $q_{vs^3}$ are all basic factors of $H$.
           The vertex $u$ does not appear in any of them.
           Also, $$\omega(p_{sv})+\omega(p'_{sv})+\omega(p''_{sv})+\omega(q_{vs^3})=2\omega(F')>0.$$
           Then, among them at least one is positive.
           Thus, we can find a basic factor of $\Omega$ satisfying the requirements. 
           \end{enumerate}
 We are done with Case \Rmnum{1} where $F$ is a path.
 
 \vspace{2ex}
\noindent {\bf Case \Rmnum{2}:} $F$ is a tadpole graph and $\deg_F(u)=3$. 
By assumption, $\pi(u)=\{0, 2, 3\}$.
Also, since $\deg_F(v)=1\in \pi(v)$, $v$ is 1-feasible.
 Let $C$ be the cycle part of $F$. Consider $\{s, t\}\cap V_C$. 
 Here, we discuss possible cases depending  on intersection vertices belonging to $V_C$ instead of the entire set $V_\cap$ of vertices points as in Case \Rmnum{1}.
 There are three subcases.
 
 \begin{enumerate}
     \item[\Rmnum{2}.1]\label{case2.1}  $\{s, t\}\cap V_C=\emptyset$.
     
     In this case, $\deg_H(x)=2$ for every $x\in V_C\backslash\{u\}$.
     Thus, $C$ is a dangling cycle with in connecting point $u$ in $H$.
Then, the case is similar to {Case \Rmnum{1}}. 
 By replacing the vertex $u$ in Case \Rmnum{1} by the cycle $C$, one can check that the proof of Case \Rmnum{1} works here.
Note that after the above replacement, a path containing $u$ as an endpoint in Case \Rmnum{1} becomes a tadpole graph containing the cycle $C$, and a tadpole graph containing $u$ as the vertex of degree 1 in Case \Rmnum{1} becomes a dumbbell graph. 

\item[\Rmnum{2}.2]\label{case2.2} $\{s, t\}\cap V_C=\{s\}$ or $\{t\}$.

Without loss of generality, we may assume that $s\in V_C$.
Then, $\deg_H(u)=\deg_H(s)=3$ and $\pi(u)=\pi(s)=\{0, 2, 3\}$.
If $\omega(C)>0$, then we are done since $C$ is a basic factor of $\Omega$ and $\deg_C(u)=2$.
Thus, we may assume that $\omega(C)\leq 0$.
Vertices $s$ and $u$ split $C$ into two paths $p_{us}$ and $p'_{us}$.
Since $\omega(C)=\omega(p_{us})+\omega(p'_{us})\leq 0$, among them at least one is non-positive.
Without loss of generality, we assume that $\omega(p_{us})\leq 0$.

Consider the graph $H'=H\backslash p_{us}$.
Note that $V_{H'}=(V_H\backslash{V_{p_{us}}})\cup \{u, s\}$.
For every $x\in V_{H'}\backslash\{u, s\}$, we have $\deg_{H'}(x)=\deg_{H}(x)\in \pi(x)$ since $H$ is a factor of $\Omega$. 
Also, $\deg_{H'}(u)=2\in \pi(u)$ and $\deg_{H'}(s)=2\in \pi(s)$.
Thus, $H'$ is a factor of $\Omega$.
Also, $\omega(H')=\omega(H)-\omega(p_{us})>0$.
However, it is not clear whether $H'$ is  a \emph{basic} factor of $\Omega$.
Consider the sub-instance $\Omega_H'=(H', \pi_{H'}, \omega)$ of $\Omega$ induced by the factor $H'$.
Since $\omega(H')>0$, 
by Lemma~\ref{lem-MBFC}, there is a basic factor $F^\ast\in \Omega_{H'}$ such that $\omega(F^\ast)>0$.
Then, $\deg_{F^\ast}(u)\in \pi_{H'}(u)=\{0, 2\}$.
Clearly, $F^\ast$ is also a basic factor of $\Omega$. We are done.

Note that this proof works no matter  whether $F'$ is a path or a tadpole graph, and  whether $v\in V_\cap$ or $t\in V_\cap$.
In fact, this proof also works when $F$ is a dumbbell graph as long as $s$ (or symmetrically $t$) is the only vertex in $V_{F'}$ appearing in the cycle $C$ of $F$ that contains  the vertex $u$.

\item[\Rmnum{2}.3]\label{case2.3} $\{s, t\}\subseteq V_C$.

\begin{figure}[!htbp]
    \centering
    \includegraphics[height=2.9cm]{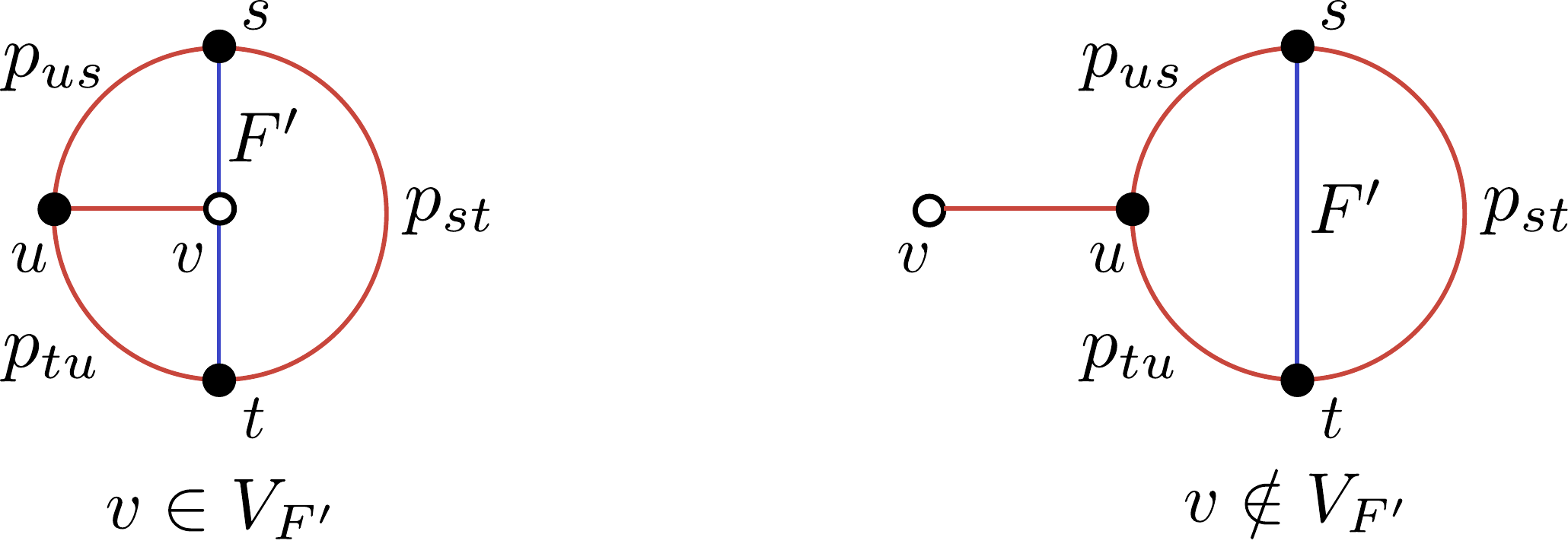}
    \caption{The two possible forms of graph $H$ in Case \Rmnum{2}.3.}
    \label{fig:case-2.3}
\end{figure}

In this case, $\deg_H(u)=\deg_H(s)=\deg_H(t)=3$ and $\pi(u)=\pi(s)=\pi(t)=\{0, 2, 3\}$.
Also,  $\deg_{F'}(s)=\deg_{F'}(t)=1$.
Thus, $F'$ is a path with endpoints $s$ and $t$. 
Note that in this case, it is possible that $v\in V_{F'}$.
If $v\in V_{F'}$, then $\deg_H(v)=3$ and $\pi(v)=\{0, 1, 3\}$; otherwise, $\deg_H(v)=1$ and $\pi(v)=\{0, 1\}$.
The points $u$, $s$, and $t$ split $C$ into three paths, $p_{us}$, $p_{st}$, $p_{tu}$. Then, $C=p_{us}\cup p_{st} \cup p_{tu}$. (See Figure \ref{fig:case-2.3}.) 
If $\omega(C)>0$, then  we are done. Thus, we may assume that $\omega(C)\leq 0$.

Consider the graph $H_1=H\backslash p_{st}=(F\backslash p_{st})\cup F'$.
Similar to the above Case \href{case2.2}{\Rmnum{2}.2}, one can check that $H_1$ is a factor of $\Omega$. 
Also, $H_1$ is a tadpole graph if $\deg_H(v)=1$  or a theta graph if
$\deg_H(v)=3$.
Thus, in both cases, $H_1$ is a \emph{basic} factor of $\Omega$. 
Since $F$ is a basic factor of $\Omega$ with the largest weight, we have $$\omega(F)\geq \omega(H_1)=\omega(F)-\omega(p_{st})+\omega(F').$$
Thus, $\omega(p_{st})\geq \omega(F')>0$.
Since $\omega(C)=\omega(p_{st})+\omega(p_{us})+\omega(p_{tu})\leq 0$, we have $\omega(p_{us})+\omega(p_{tu})<0$.
Without loss of generality, we may assume that $\omega(p_{us})<0$.
Then, consider the graph $H_2=H\backslash p_{us}$. 
Still, one can check that $H_2$ is a factor of $\Omega$, and $\deg_{H_2}(u)=2$.
Also, $H_2$ is a tadpole graph if  $\deg_H(v)=1$, or a theta graph if $\deg_H(v)=3$.
Thus, $H_2$ is a basic factor of $\Omega$. 
Moreover, $\omega(H_2)=\omega(H)-\omega(p_{us})>0$. We are done.

\end{enumerate}

\vspace{1ex}
\noindent{\bf Case \Rmnum{3}:}
$F$ is a tadpole graph and $\deg_F(v)=3$.
In this case, $\deg_F(u)=1$, $\pi(u)=\{0, 1\}$, $\deg_F(v)=3$, and $\pi(v)=\{0, 1, 3\}$ or $\{0, 2, 3\}$.
Recall that $\deg_H(u)=\deg_F(u)=1$, and $u\notin V_\cap$.
 Let $C$ be the cycle part of $F$. 
 Still consider $\{s, t\}\cap V_C$. There are three subcases.

 \begin{enumerate}
     \item[\Rmnum{3}.1]  $\{s, t\}\cap V_C=\emptyset$. 
     
In this case, $\deg_H(x)=2$ for every $x\in V_C\backslash\{v\}$.
Thus, in the graph $H$, the cycle $C$ is a dangling cycle with the connecting point $v$. 
Then, the case is similar to  Case \Rmnum{1}.
For a graph $H$ in Case \Rmnum{1} where $\deg_H(v)=1$ (i.e., $v \notin V_\cap$), 
by replacing the vertex $v$    by the cycle $C$, one can check that the proof of Case \Rmnum{1} works here.

\item[\Rmnum{3}.2]\label{case3.2} $\{s, t\}\cap V_C=\{s\}$ or $\{t\}$. 

Without loss of generality, we may assume that $s\in V_C$.
Then $\deg_H(s)=3$ and  $\pi(s)=\{0, 2, 3\}$.
Vertices $s$ and $v$ split $C$ into two paths $p_{vs}$ and $p'_{vs}$.
Let $p_{uv}$ be the path part in the tadpole graph $F$.
There are two subcases depending on 
whether $t\in V_F$.
Since $t\notin V_C$, $t\in V_F$ implies $t\in V_{p_{uv}}$.
\begin{enumerate}

    \item $t\notin V_{p_{uv}}$. 
    
    \begin{figure}[!htbp]
        \centering
        \includegraphics[height=3.2cm]{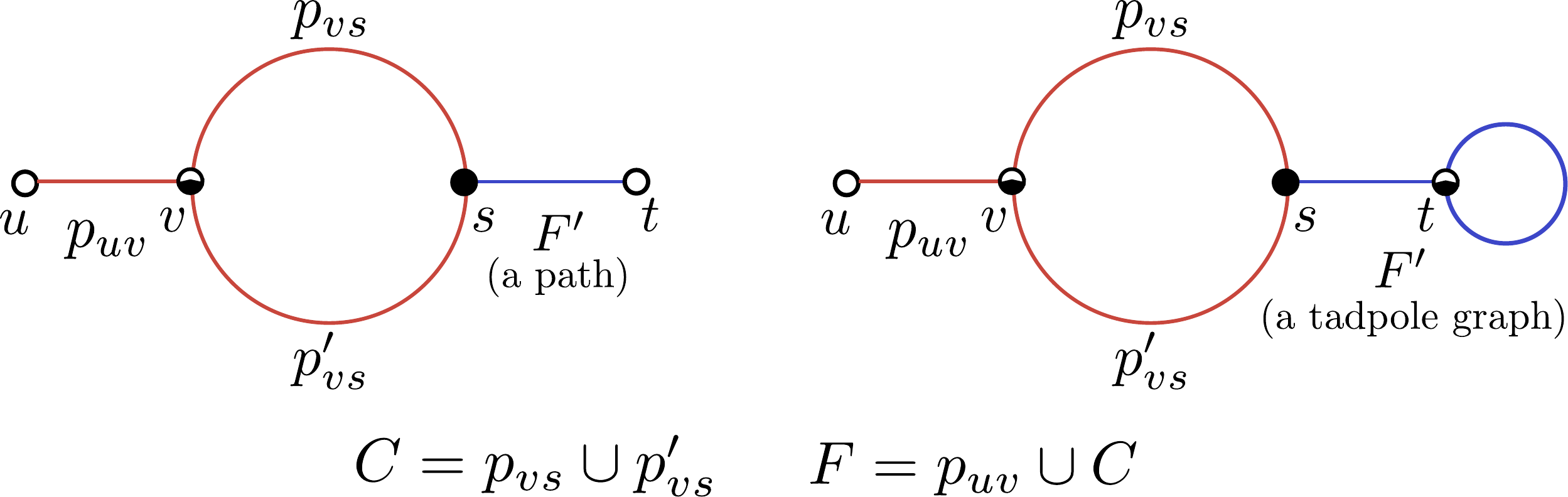}
        \caption{The two possible forms of graph $H$ in Case \Rmnum{3}.2 where $t\notin V_{p_{uv}}$.}
        \label{fig:3.2}
    \end{figure}

     In this case, $\deg_H(t)=1$ or $3$ depending on whether $F'$ is a path or a tadpole graph respectively (See Figure~\ref{fig:3.2}).
    If $F'$ is a tadpole graph, then the cycle in $F'$ containing $t$ is  a dangling cycle in $H$ with the connecting point $t$.
\begin{itemize}
    \item If $\pi(v)=\{0, 1, 3\}$, then
    $p_{uv}$ is a basic factor of $\Omega$ (this is true no matter whether $t\in V_{p_{uv}}$).
    Since $F$ is a basic factor of $\Omega$ with the largest weight,
$\omega(F)\geq \omega(p_{uv})$.
Thus, $\omega(p_{vs})+\omega(p'_{vs})=\omega(C)=\omega(F)-\omega(p_{uv})\geq 0$.
Without loss of generality, we may assume that $\omega(p_{vs})\geq 0$.
Consider the graph $H'=F'\cup p_{vs}$.
It is a path if $F'$ is a path, or a tadpole graph of $F'$ is a tadpole graph. 
Note that $\deg_{H'}(u)=0\in \pi(u)$, $\deg_{H'}(v)=1\in \pi(v)$, $\deg_{H'}(s)=2\in \pi(s)$, and $\deg_{H'}(t)=\deg_{H}(t)\in \pi(t)$.
Also, for every $x\in V_{H'}\backslash\{u, v, s, t\}$, $\deg_{H'}(x)=\deg_{H}(x)\in \pi(x)$.
Thus, $H'$ is a basic factor of $\Omega$. 
Also, $\omega(H')=\omega(F')+\omega(p_{vs})>0$.
We are done.

\item If $\pi(v)=\{0, 2, 3\}$, then the cycle $C$ is a basic factor of $\Omega$.
Consider $H_1=H\backslash p_{vs}$. 
Note that $\deg_{H_1}(u)=1\in \pi(u)$, $\deg_{H_1}(v)=2\in \pi(v)$, $\deg_{H_1}(s)=2\in \pi(s)$, and $\deg_{H_1}(t)=\deg_{H}(t)\in \pi(t)$.
One can check that $H_1$ is a factor of $\Omega$.
Also, $H_1$ is either a path with endpoints $u$ and $s$ if $F'$ is a path, or a tadpole graph with $u$ being the vertex of degree $1$ and $t$ being the vertex of degree $3$ if $F'$ is a tadpole graph. 
Thus, $H_1$ is a basic factor of $\Omega$.
Since $F$ is a basic factor with the largest weight, 
$$\omega(F)\geq \omega(H_1)=\omega(F)-\omega(p_{vs})+\omega(F').$$
Thus, $\omega(p_{vs})\geq \omega(F')>0$.
Similarly, by considering $H_2=H\backslash p'_{vs}$, we have $\omega(p'_{vs})>0$.
Then, $\omega(C)=\omega(p_{vs})+\omega(p'_{vs})>0$.
Thus, $C$ is a basic factor of $\Omega$ with positive weight and $\deg_C(u)=0$.
\end{itemize}
\item $t\in V_{p_{uv}}$. 

\begin{figure}[!htbp]
    \centering
    \includegraphics[height=2.8cm]{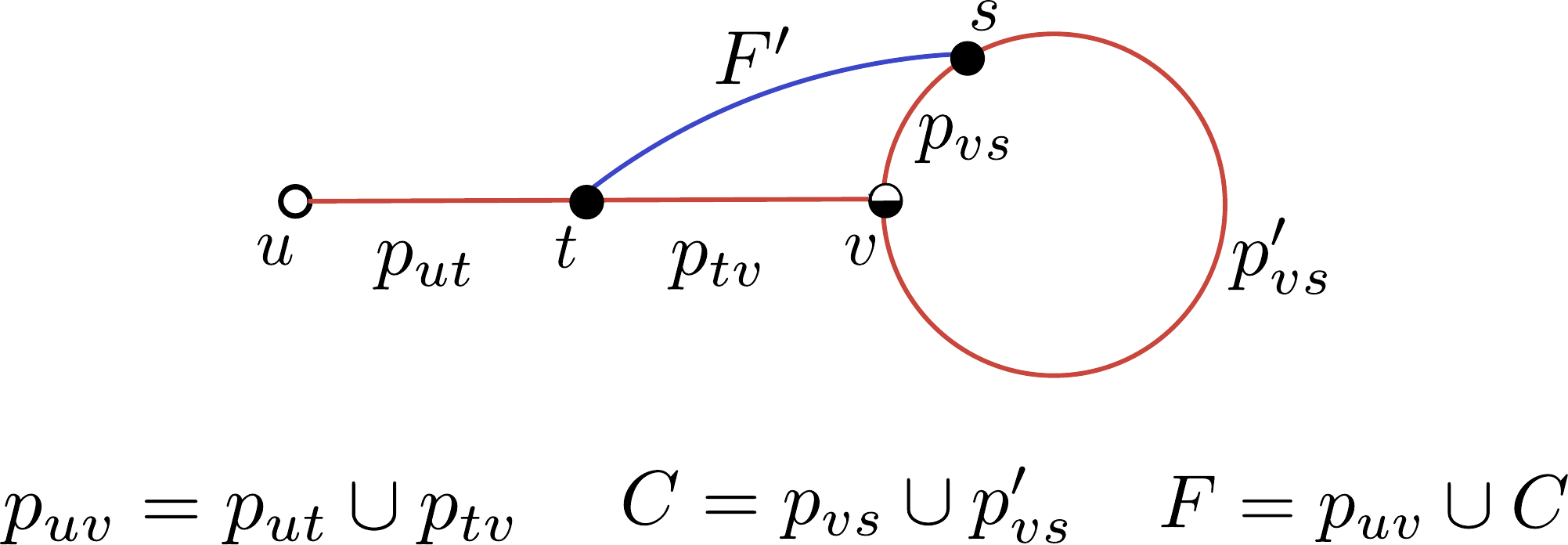}
    \caption{The graph $H$ in Case \Rmnum{3}.2 where $t\in V_{p_{uv}}$.}
    \label{fig:3.2.b}
\end{figure}

In this case, $\deg_H(t)=3$ and $\pi(t)=\{0, 2, 3\}$.
$F'$ is a path with endpoints $s$ and $t$. 
The vertex $t$ splits $p_{uv}$ into two parts $p_{ut}$ and $p_{tv}$ (see Figure~\ref{fig:3.2.b}).
\begin{itemize}
    \item If $\pi(v)=\{0, 1, 3\}$, then $p_{uv}$ is  a basic factor of $\Omega$.
    Since $\omega(F)\geq \omega(p_{uv})$, we have $\omega(C)\geq 0$.
    Consider the path $p'_{uv}=p_{ut}\cup F'\cup p_{vs}$.
    It is also a basic factor of $\Omega$.
    Still, since $\omega(F)\geq \omega(p'_{uv})$, we have $$\omega(p_{tv})\geq \omega(F')+\omega(p_{vs}).$$
    Similarly, by considering the path $p''_{uv}=p_{ut}\cup F'\cup p'_{vs}$, we have $$\omega(p_{tv})\geq \omega(F')+\omega(p'_{vs}).$$
    Thus, $2\omega(p_{tv})\geq 2\omega(F')+\omega(p_{vs})+\omega(p'_{vs})=2\omega(F')+\omega(C)>0$.
    Consider the theta graph $F^\ast=p_{tv}\cup C \cup F'$.
    Clearly, $\omega(F^\ast)>0$.
   Then, $F^\ast$ is a basic factor of $\Omega$ with $\deg_{F^\ast}(u)=0$. We are done.
    \item If $\pi(v)=\{0, 2, 3\}$, then $C$ is a basic factor of $\Omega$.
    Consider $H_1=H\backslash p_{vs}$.
    It is a  tadpole graph with the vertex $u$ of degree $1$ and the vertex $t$ of degree $3$. 
Note that $\deg_{H_1}(u)=1\in \pi(u)$, $\deg_{H_1}(v)=2\in \pi(v)$, $\deg_{H_1}(s)=2\in \pi(s)$, and $\deg_{H_1}(t)=3\in \pi(t)$.
One can check that $H_1$ is a basic factor of $\Omega$.
Thus, $H_1$ is a basic factor of $\Omega$.
Since $F$ is a basic factor with the largest weight, 
$$\omega(F)\geq \omega(H_1)=\omega(F)-\omega(p_{vs})+\omega(F').$$
Thus, $\omega(p_{vs})\geq \omega(F')>0$.
Similarly, by considering $H_2=H\backslash p'_{vs}$, we have $\omega(p'_{vs})>0$.
Then, $\omega(C)=\omega(p_{vs})+\omega(p'_{vs})>0$.
Thus, $C$ is a basic factor of $\Omega$ with positive weight and $\deg_C(u)=0$.
\end{itemize}
\end{enumerate}

\item[\Rmnum{3}.3]\label{case3.3} $\{s, t\}\in V_C$.

In this case, $\deg_H(u)=1$, $\pi(u)=\{0, 1\}$,  $\deg_H(v)=\deg_H(s)=\deg_H(t)=3$, and $\pi(s)=\pi(t)=\{0, 2, 3\}$.
Also,  $\deg_{F'}(s)=\deg_{F'}(t)=1$.
Thus, $F'$ is a path with endpoints $s$ and $t$. 
Let $p_{st}\subseteq C$ be the path with endpoints $t$ and $s$ such that $v\notin V_{p_{st}}$. 

Consider the tadpole graph $q_{uv^3}=(F\backslash p_{st})\cup F'$.
In other words, $q_{uv^3}$ is the tadpole graph obtained from $F$ by replacing the path $p_{st}$ by $F'$.
One can check that  $q_{uv^3}$ is also a basic factor of $\Omega$.
Since $F$ is a basic factor of $\Omega$ with the largest weight, $$\omega(F)\geq \omega(q_{uv^3})=\omega(F)-\omega(p_{st})+\omega(F').$$
Thus, $\omega(p_{st})\geq \omega(F')>0$.
Consider the cycle $C'=p_{st}\cup F'$.
Note that it is a basic factor of $\Omega$.
Also, $\deg_{C'}(u)=0$ and $\omega(C')=\omega(p_{st})+\omega(F')>0$.
We are done.
 \end{enumerate}

\vspace{1ex}

\noindent{\bf Case \Rmnum{4}:}
$F$ is a dumbbell graph.
Let $C_u$ and $C_v$ be the two cycles of $F$ containing vertices $u$ and $v$ respectively.

\vspace{1ex}

If $\{s, t\}\cap C_v=\emptyset$, 
then $C_v$ is a dangling cycle in $H$ with the connecting point $v$.  
This case is similar to Case \Rmnum{2}.
For a graph $H$ in Case \Rmnum{2} where $\deg_H(v)=1$ (i.e., $v\notin V_\cap$), 
by replacing the vertex $v$ by the cycle $C_v$, one can check that the proof of Case \Rmnum{2} works here. 

If $\{s, t\}\cap C_u=\emptyset$, 
then $C_u$ is a dangling cycle in $H$ with the connecting point $u$.  
This case is similar to Case \Rmnum{3}.
By replacing the vertex $u$ in Case \Rmnum{3} by the cycle $C_u$, one can check that the proof of Case \Rmnum{3} works here. 

If $\{s, t\}\cap C_u$ and $\{s, t\}\cap C_v$ are both non-empty, then without loss of generality, we may assume that $s\in C_u$ and $t\in C_v$.
Thus, $F'$ is a path with endpoints $s$ and $t$. 
As we have mentioned in  Case \href{case2.2}{\Rmnum{2}.2}, one can check that the proof of Case \href{case2.2}{\Rmnum{2}.2} works here. 

\vspace{2ex}

\noindent{\bf Case \Rmnum{5}:} $F$ is a theta graph. In this case, $\deg_H(u)=\deg_H(v)=3$. 
By assumption, $\pi(u)=\{0, 2, 3\}$.
Also,  by the definition of theta graphs, $\pi(v)=\{0, 1, 3\}$.
Then, $V_\cap\subseteq\{s, t\}$. There are two subcases.
\begin{enumerate}
    \item[\Rmnum{5}.1] $V_\cap=\{s\}$ or $\{t\}$. 
    
    Without loss of generality, we assume that $V_\cap=\{s\}$.
    Then, $\deg_H(s)=3$ and $\pi(s)=\{0, 2, 3\}$.
             The theta graph  $F$ consists of three paths $p_{uv}$, $p'_{uv}$ and $p''_{uv}$.
           Without loss of generality, we may assume that $s$ appears in the path $p_{uv}$ and it splits $p_{uv}$ into two paths $p_{us}$ and $p_{sv}$.
           
           Consider the paths $p_{sv}$, $p'_{sv}=p'_{uv}\cup p_{su}$ and $p''_{sv}=p''_{uv}\cup p_{su}$, 
           and the tadpole graph $q_{sv^3}=p_{sv}\cup p'_{uv}\cup p''_{uv}$.
           They are not factors of $H$ since the degree of $s$ is 1 in all these four graphs. 
           However, by taking the union of $F'$ with any one of them, we can get a basic factor of $H$ and the degree of $u$ in it is even. 
           Since $$\omega(p_{sv})+\omega(p'_{sv})+\omega(p''_{sv})+\omega(q_{sv^3})=2\omega(F')>0,$$
         among them at least one is positive.
         Also, $\omega(F')>0$.
         Then, by taking the union of it with $F'$, we can find a basic factor of $\Omega$ satisfying the requirements. 
         
         \item[\Rmnum{5}.2] $V_\cap=\{s, t\}$.
         
         In this case, $F'$ is a path with endpoints $s$ and $t$.
         Since $F$ is a theta graph which is 2-connected, we can find a path $p_{st}\subseteq F$ such that $v\notin V_{p_{st}}$.
         If $u\notin V_{p_{st}}$, then one can check that the theta graph $H'=(F\backslash p_{st})\cup F'$ is also a basic factor of $\Omega$.
         Since $\omega(F)\geq\omega(H')$, we have $\omega(p_{st})\geq\omega(F')>0$.
         Then, the cycle $C=p_{st}\cup F'$ is a basic factor of $\Omega$ where $\omega(C)=\omega(p_{st})+\omega(F')>0$ and $\deg_C(u)=0$. We are done.
         Otherwise, $u\in V_{p_{st}}$.
         The vertex $u$ splits $p_{st}$ into two paths $p_{us}$ and $p_{ut}$.
         Consider the theta graph $H'=(F\backslash p_{us})\cup F'$, where $\deg_{H'}(v)=\deg_{H'}(t)=3$, $\pi(v)=\{0, 1, 3\}$, and $\pi(t)=\{0, 2, 3\}$.
         One can check that $H'$ is a basic factor of $\Omega$.
         Since $\omega(F)\geq \omega(H')=\omega(F)-\omega(p_{us})+\omega(F')$, we have
         $\omega(p_{us})\geq \omega(F')>0$.
         Similarly, by considering the theta graph $H''=(F\backslash p_{ut})\cup F'$, 
         we have  $\omega(p_{ut})\geq \omega(F')>0$.
         Then the cycle $C=p_{st}\cup F'=p_{us}\cup p_{ut}\cup F'$ is  a basic factor of $\Omega$ where $\omega(C)=\omega(p_{us})+\omega(p_{ut})+\omega(F')>0$ and $\deg_C(u)=2$. We are done.
 \end{enumerate}
 We have taken care of all possible cases and finished the proof. 
\end{proof} 

Combining Lemmas~\ref{lem-MBFC} and~\ref{lem-2}, we finished the proof of Theorem~\ref{lem-basic-factor-star}.

\appendix

\section{$\Delta$-Matroids and Matching Realizability}\label{sec:matroid}

A $\Delta$-matroid is a family of sets obeying an axiom generalizing the 
matroid exchange axiom. Formally, a pair $M=(U,\sF)$ is a $\Delta$-matroid if
$U$ is a finite set and $\sF$ is a collection of subsets of $U$
satisfying the following: for any $X,Y\in\sF$ and any $u\in X\Delta Y$
in the symmetric difference of $X$ and $Y$, there exits a $v\in X\Delta Y$ such
that $X\Delta\{u,v\}$ belongs to $\sF$~\cite{bouchet1987greedy}.
A $\Delta$-matroid is \emph{symmetric} if, for every pair of $X,Y\subseteq U$ with
$|X|=|Y|$, we have $X\in\sF$ if and only if $Y\in\sF$.
A $\Delta$-matroid is \emph{even} if for every pair of  $X,Y\subseteq U$, $|X|\equiv |Y| \mod 2$.

Suppose that $U=\{u_1, u_2, \ldots, u_n\}$.
A subset $V\subseteq U$ can be encoded by a binary string $\alpha_V$ of $n$-bits where
the $i$-th bit of $\alpha_V$ is $1$ if $u_i\in V$ and $0$ if  $u_i\notin V$.
Then, a $\Delta$-matroid $M=(U,\sF)$ can be  represented by a relation $R_M$ of arity $|U|$ which consists of binary strings that encode all subsets in $\sF$.
Such a representation is unique up to a permutation of variables of the relation. 
A degree constraint $D$ of arity $n$ can be viewed  as an $n$-ary symmetric relation which consists of binary strings with the Hamming weight $d$ for every $d\in D$.
By the definition of $\Delta$-matroids, 
it is easy to check that a degree constraint $D$ (as a symmetric relation) represents a $\Delta$-matroid  if and only if $D$ has all gaps of length at most $1$.

\begin{definition}[Matching Gadget]\label{def:matchgate}
A \emph{gadget} using a set $\mathscr{D}$ of degree constraints
consists of a graph $G=(U\cup V, E)$ where $\deg_G(u)=1$ for every $u\in U$ and there are no edges between vertices in $U$,
and a mapping $\pi:V\rightarrow \mathscr{D}$.
A \emph{matching gadget} is a gadget where $\mathscr{D}=\{ \{0,1\}, \{1\}\}$.
A degree constraint $D$ of arity $n$ is \emph{matching realizable} if  there exists a  matching gadget $(G=(U\cup V, E), \pi:V\rightarrow \{ \{0,1\}, \{1\}\})$ 
such that $|U|=n$ and for every $k\in [n]$, $k\in D$ if and only if for every $W\subseteq U$ with $|W|=k$, there exists a matching $F=(V_F, E_F)$ of $G$ such that $V_F\cap U=W$ and 
for every $v\in V$ where $\pi(v)=\{1\}$, $v\in V_F$.
\end{definition}

The definition of matching realizability   can be extended to a relation $R$ of arity $n$ by requiring the set $U$ of $n$ vertices in a matching gadget to represent the $n$ variables of $R$. 
If $R$ is realizable by a matching gadget  $G=(U\cup V, E)$, then for every $\alpha\in\{0, 1\}^n$, $\alpha\in R$ if and only if there is a matching $F=(V_F, E_F)$ of $G$ such that $V_F\cap U$ is exactly the subset of $U$ encoded by $\alpha$
(i.e., for every $u_i\in U$, $u_i\in V_F$ if and only if $\alpha_i=1$), and for every $v\in V$ where $\pi(v)=\{1\}$, $v\in V_F$.
Note that the matching realizability of a relation is invariant under a permutation of its variables. 
We say that a $\Delta$-matroid is matching realizable if the relation representing it is matching realizable.\footnote{This definition of matching realizability for $\Delta$-matroids is different with the one that is usually used
for even $\Delta$-matroids~{\rm\cite{bouchet1989matchings, Dvorak15:icalp, kkr2018even}}, in which 
the gadget
is only allowed to use the constraint $\{1\}$ for perfect matchings, and hence the resulting $\Delta$-matroid must be even.}
The following is an equivalent definition for the matching realizability of  $\Delta$-matroids.

\begin{lemma}\label{def:matching-gadget2}
If a $\Delta$-matroid $M=(U, \sF)$ is matching
realizable, then there is a graph $G=(U\cup W \cup X, E)$ where $\deg(v)=1$ for
every $v\in U\cup X$ and there are no edges between vertices in $U\cup X$, such that  for every
$V\subseteq U$, $V\in \sF$ if and only if there exists $X_1\subseteq X$ such
that the induced subgraph of $G$ induced by the vertex set $V\cup W \cup X_1$ (denoted by $G(V\cup W \cup X_1)$) has a  \emph{perfect} matching. 

With a slight abuse of notation, we also say the graph $G=(U\cup W \cup X, E)$ realizes $M$. 
\end{lemma}

\begin{proof}
Let $G=(U\cup W, E)$ be the matching gadget realizing $M=(U, \sF)$. 
We construct the following graph $G'$ from $G$. 
For every $x\in W$ with $\pi(x)=\{0, 1\}$, 
we add a new edge incident to it. 
As the edge is added, a new vertex of degree of $1$ is also added to the graph.
We denote  these new vertices by $X$ and these new edges by $E_X$.
Then, one can check that the graph $G'=(U\cup W\cup X, E\cup E_X)$ satisfies the requirements.  
\end{proof}

The following result generalizes Lemma A.1 of~\cite{kkr2018even}. 

\begin{lemma}\label{lem-mixed-partition}
Suppose that $M=(U, \sF)$ is a matching realizable  $\Delta$-matroid, and $V_1, V_2 \in \sF$. 
Then, $V_1\Delta V_2$ can be partitioned into single variables  $S_1, \ldots, S_k$  and  pairs of variables  $P_1, \ldots, P_\ell$ 
such that for every $P=S_{i_1}\cup \cdots \cup S_{i_r}\cup P_{j_1}\cup \cdots \cup P_{j_t}$ $(\{i_1, \ldots, i_r\}\subseteq [k], \{j_1, \ldots, j_t\}\subseteq [\ell])$, 
$V_1\Delta P\in \sF$ and $V_2\Delta P\in \sF$.
\end{lemma}
\begin{proof}
By Lemma~\ref{def:matching-gadget2}, there is a graph $G=(U\cup W\cup X, E)$  realizing $M$. 
Since $V_1, V_2 \in \sF$, 
there exists $X_1\subseteq X$ and $X_2\subseteq X$ 
such that the induced subgraph $G(V_1\cup W \cup X_1)$ has a perfect matching $M_1$, and $G(V_2\cup W \cup X_2)$ has a perfect matching $M_2$. 
Let $E_1$ and $E_2$ be the edge sets of $M_1$ and $M_2$ respectively. 
Consider the graph $G'=(U\cup W\cup X, E_1\Delta E_2)$. 
Since $E_1$ covers each vertex in $V_1\cup W \cup X_1$ exactly once, and $E_2$ covers each vertex in $V_2\cup W \cup X_2$ exactly once,  
for every $v\in (V_1\cap V_2)\cup W \cup (X_1 \cap X_2)$ in $G'$, $\deg(v)=0$ or $2$, 
and for every $v\in (V_1\Delta V_2)\cup (X_1\Delta X_2)$ in $G'$, $\deg(v)=1$.
Thus, $G'$ is a union of induced cycles and paths, 
where each path connects two vertices in $(V_1\Delta V_2)\cup (X_1\Delta X_2)$. 
For every vertex $u\in V_1\Delta V_2$, 
if it is connected to another vertex $v\in V_1\Delta V_2$ by a path in $G'$, then we make $\{u, v\}$ a pair. 
Otherwise (i.e., $u$ is connected to a vertex in $X_1\Delta X_2$ by a path in $G'$), 
we make $u$ a single variable. 
Then, $V_1\Delta V_2$ can be partitioned into single variables $S_1, \ldots, S_k$ and pairs $P_1, \ldots, P_\ell$ according to the paths in $G'$.

Moreover, each path in $G'$ is an alternating path with respect to both matchings  $M_1$ and $M_2$. 
Pick a union of  such paths (note that they are edge-disjoint).
Suppose that there are $r$ many paths that connect single variables in $S_{i_1}, \ldots, S_{i_r}$ with variables in $X$, 
and $t$ many paths that connect pairs $P_{j_1}, \ldots, P_{j_t}$.
Let $P=S_{i_1}\cup \cdots \cup S_{i_r} \cup P_{j_1}\cup \cdots \cup P_{j_t}$.
After altering the matchings $M_1$ and $M_2$ according to these $t$ many alternating paths, 
we obtain two new matchings that cover exactly $(V_1\Delta P)\cup W \cup X'_1$ for some $X'_1\subseteq X$ and $(V_2\Delta P) \cup W \cup X'_2$ for some $X'_2\subseteq X$ respectively. 
Thus, $V_1\Delta P\in \sF$ and $V_2\Delta P\in \sF$.
\end{proof}

\begin{theorem}\label{thm:real}
A degree constraint $D$ of gaps of length at most $1$
is matching realizable if and only if all its
  gaps are of the same length $0$ or $1$.
\end{theorem}
\begin{proof}
 By the gadget constructed in the proof 
  of~\cite[Theorem~2]{cornuejols1988general}, if a degree constraint
  has all gaps of length  $1$ then it is 
  matching realizable.\footnote{We remark that~\cite{cornuejols1988general}
  includes gadgets for other types of degree constraints, including type-1 and type-2, but only under a more general notion of gadget constructions that involve edges  and triangles. 
  The gadget that only involves edges
  is a matching gadget defined in this paper.}
  We give the following gadget (Figure~\ref{fig:matchgate}) to realize a degree constraint $D$ with all gaps of length $0$, which generalizes the gadget in~\cite{tutte1954short}.
  Suppose that $D=\{p, p+1, \ldots, p+r\}$ of arity $n$ where $n\geq p+r\geq p\geq 0$. 
  Consider the following graph $G=(U\cup V, E)$:  $U$ consists of $n$ vertices of degree $1$, and $V$ consists of two parts $V_1$ with  $|V_1|=n$ and $V_2$ with  $|V_2|=n-p$;
the induced subgraph $G(V)$ of $G$  induced by $V$ is a  complete bipartite graph between $V_1$ and $V_2$, and 
the induced subgraph $G(U\cup V_1)$ of $G$ induced by $U\cup V_1$ is a  bipartite perfect matching between $U$ and $V_1$. 
  Every vertex in $V_1$ is labeled by the constraint $\{1\}$.
  There are $r$ vertices in $V_2$ labeled by $\{0, 1\}$ and the other $n-p-r$ vertices in $V_2$ labeled by $\{1\}$.
  One can check that this gadget realizes $D$.
  \begin{figure}[!h]
      \centering
      \includegraphics[height=5cm]{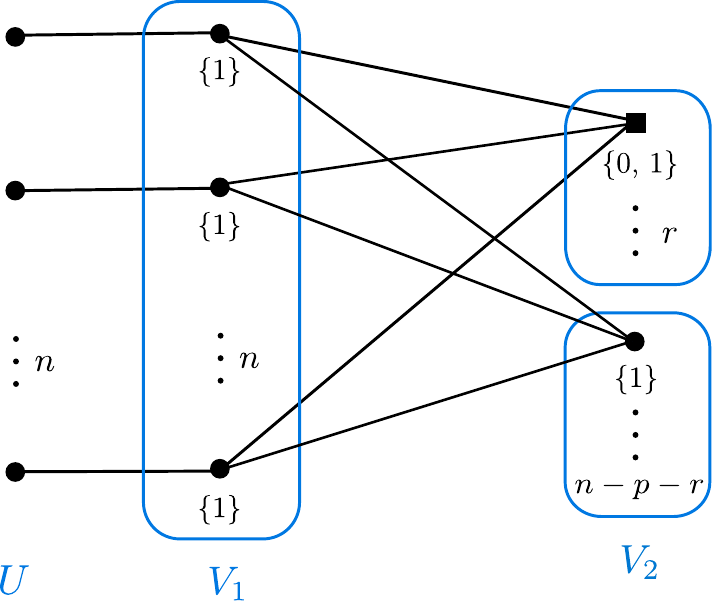}
      \caption{A matching gadget realizing $D=\{p, p+1, \ldots, p+r\}$ of arity $n$}
      \label{fig:matchgate}
  \end{figure}
  
For the other direction, without loss of generality, we may assume that $\{p, p+1, p+3\}\subseteq D$ and $p+2\notin D$. 
Since $D$ has gaps of length at most $1$, 
it can be associated with a symmetric $\Delta$-matroid $M=(U, \sF)$. 
Then, there is $V_1\in \sF$ with $|V_1|=p$ and $V_2\in \sF$ with $|V_2|=p+3$.
Since $M$ is symmetric, we may pick $V_2=V_1\cup \{v_1, v_2, v_3\}$ for some $\{v_1, v_2, v_3\}\cap V_1=\emptyset$. Let $S=V_1\Delta V_2=\{v_1, v_2, v_3\}$.
By Lemma~\ref{lem-mixed-partition}, $S$ can be partitioned into single variables and/or pairs of variables such that for any union $P$ of them, $V_2\backslash P\in \sF$.
Since $|S|=3$, there exists at least a single variable $x_i$ in the partition of $S$
  such that $V_2\backslash\{v_i\}\in \sF.$
Note that $|V_2\backslash\{v_i\}|=p+2$. Thus, $p+2\in D$. A  contradiction.
\end{proof}

\bibliographystyle{alpha}
\bibliography{sz}
\end{document}